\documentclass[11pt,journal, onecolumn]{IEEEtran}
\usepackage{amsmath,amssymb,amsthm}
\usepackage{xifthen}
\usepackage{mathtools}
\usepackage{enumerate}
\usepackage{microtype}
\usepackage{xspace}
\usepackage{bm}
\usepackage[T1]{fontenc}
\usepackage{fancyhdr}
\usepackage{lastpage}
\usepackage{bbm}

\newcommand{\T}{{\scriptscriptstyle\mathsf{T}}}
\renewcommand{\H}{{\scriptscriptstyle\mathsf{H}}}

\newcommand{\cond}{\,\vert\,}
\renewcommand{\Re}[1][]{\ifthenelse{\isempty{#1}}{\operatorname{Re}}{\operatorname{Re}\left(#1\right)}}
\renewcommand{\Im}[1][]{\ifthenelse{\isempty{#1}}{\operatorname{Im}}{\operatorname{Im}\left(#1\right)}}

\newcommand{\tv}{\vect{t}}

\newcommand{\Sigmam}{\pmb{\Sigma}}

\newcommand{\Phim}{\pmb{\Phi}}

\newcommand{\Thetam}{\pmb{\Theta}}

\newcommand{\Am}{\mat{a}}
\newcommand{\Bm}{\mat{b}}

\newcommand{\Mm}{\mat{M}}

\newcommand{\Pm}{\mat{p}}

\newcommand{\Rm}{\mat{r}}

\newcommand{\Um}{\mat{u}}
\newcommand{\Vm}{\mat{V}}

\newcommand{\Xm}{\mat{x}}

\newcommand{\Cc}{{\mathcal C}}
\newcommand{\Dc}{{\mathcal D}}

\newcommand{\Jc}{{\mathcal J}}
\newcommand{\Kc}{{\mathcal K}}

\newcommand{\Nc}{{\mathcal N}}

\newcommand{\Vc}{{\mathcal V}}

\newcommand{\CC}{\mathbb{C}}

\newcommand{\EE}{\mathbb{E}}

\newcommand{\Id}{\mat{\mathrm{I}}}

\newcommand{\CN}[1][]{\ifthenelse{\isempty{#1}}{\mathcal{N}_{\mathbb{C}}}{\mathcal{N}_{\mathbb{C}}\left(#1\right)}}

\renewcommand{\P}[1][]{\ifthenelse{\isempty{#1}}{\mathbb{P}}{\mathbb{P}\left(#1\right)}}
\newcommand{\E}[1][]{\ifthenelse{\isempty{#1}}{\mathbb{E}}{\mathbb{E}\left[#1\right]}}
\newcommand{\I}[1][]{\ifthenelse{\isempty{#1}}{\mathbb{I}}{\mathbb{I}\left\{#1\right\}}}
\renewcommand{\det}[1][]{\ifthenelse{\isempty{#1}}{\mathrm{det}}{{\rm det}\left(#1\right)}}
\newcommand{\trace}[1][]{\ifthenelse{\isempty{#1}}{\mathrm{tr}}{{\rm tr}\left(#1\right)}}
\newcommand{\rank}[1][]{\ifthenelse{\isempty{#1}}{\mathrm{rank}}{{\rm rank}\left(#1\right)}}
\newcommand{\diag}[1][]{\ifthenelse{\isempty{#1}}{\mathrm{diag}}{{\rm diag}\left(#1\right)}}
\newcommand{\Cov}[1][]{\ifthenelse{\isempty{#1}}{\mathsf{Cov}}{\mathsf{Cov}\left(#1\right)}}

\newcommand{\defeq}{\triangleq}
\newcommand{\eqdef}{\triangleq}

\newtheorem{proposition}{Proposition}
\newtheorem{remark}{Remark}
\newtheorem{definition}{Definition}
\newtheorem{theorem}{Theorem}
\newtheorem{example}{Example}
\newtheorem{corollary}{Corollary}
\newtheorem{lemma}{Lemma}
\newtheorem{property}{Property}

\newcounter{enumi_saved}
\usepackage{answers}
\Newassociation{solution}{Solution}{solutionfile}

\AtBeginDocument{\Opensolutionfile{solutionfile}[\jobname]}
\AtEndDocument{\Closesolutionfile{solutionfile}\clearpage
}

\IfFileExists{MinionPro.sty}{
}{
}

\usepackage{cite}
\usepackage{graphicx}
\usepackage{subfigure}
\usepackage{enumitem}
\usepackage[table,dvipsnames]{xcolor}
\usepackage{amsmath,amssymb,amsfonts,amsthm,mathtools,array,stfloats}
\usepackage{multirow}
\usepackage{tikz,pgfplots,environ}
\usepackage{epstopdf}
\usetikzlibrary{shapes}
\usetikzlibrary{plotmarks}
\usetikzlibrary{spy}

\definecolor{LightCyan}{rgb}{0.8,1,1}
\definecolor{LightGray}{gray}{0.9}
\definecolor{LightYellow}{rgb}{1,1,0.5}

\DeclarePairedDelimiter\floor{\lfloor}{\rfloor}

\allowdisplaybreaks
\IEEEoverridecommandlockouts

\def\rrone{s_{1}}
\def\rrtwo{s_{2}}
\def\rronetwo{s_0}

\def\vvone{{\mathbf V}_{1}}
\def\vvtwo{{\mathbf V}_{2}}
\def\vvonetwo{{\mathbf V}_0}
\def\spanone{{{\rm Span}}({\mathbf U}_1)}
\def\spantwo{{{\rm Span}}({\mathbf U}_2)}
\def\spanvzero{{{\rm Span}}({\mathbf V}_0)}
\def\spanvone{{{\rm Span}}({\mathbf V}_1)}
\def\spanvtwo{{{\rm Span}}({\mathbf V}_2)}
\def\spanvonetwo{{{\rm Span}}({\mathbf V}_0)}
\def\Rank{{\rm rank}}
\def\eqtilde{ \tilde{{\bf \Omega}}}
\def\eqtildeone{ \tilde{{\bf \Omega}}_1}
\def\eqtildetwo{ \tilde{{\bf \Omega}}_2}
\def\eqtildeonezero{ \tilde{{\bf \Omega}}_{10}}

\def\eqhat{ \hat{{\bf \Omega}}}
\def\eqhatone{ \hat{{\bf \Omega}}_1}
\def\eqhattwo{ \hat{{\bf \Omega}}_2}
\def\eqhatonezero{ \hat{{\bf \Omega}}_{10}}
\def\eqhatoneone{ \hat{{\bf \Omega}}_{11}}
\def\eqhattwozero{ \hat{{\bf \Omega}}_{20}}
\def\eqhattwotwo{ \hat{{\bf \Omega}}_{22}}
\newcommand{\Null}[1]{{{\rm null}\left(#1\right)}}
\newcommand{\Span}[1][]{\ifthenelse{\isempty{#1}}{{\rm Span}}{{\rm Span}\left(#1\right)}}

\def\r{r}

\def\ZeroMat{\mathbf 0}

\def\tX{{\mathbf X}}
\def\tx{{\mathbf x}}
\def\tH{{\textbf H}}
\def\tY{{\mathbf Y}}
\def\tW{{\mathbf W}}
\def\tw{{\mathbf w}}
\def\tI{{\mathbf I}}
\def\tA{{\mathbf A}}
\def\tB{{\mathbf B}}
\def\tR{{\mathbf R}}
\def\tC{{\mathbf C}}
\def\tD{{\mathbf D}}
\def\tE{{\mathbf E}}
\def\tU{{\mathbf U}}
\def\tG{{\mathbf G}}
\def\tv{{\mathbf v}}
\def\tV{{\mathbf V}}
\def\tT{{\mathbf T}}
\def\tS{{\mathbf S}}
\def\ts{{\mathbf s}}
\def\ty{{\mathbf y}}
\def\tP{{\mathbf P}}
\def\tM{{\mathbf M}}

\def\tg{{\mathbf g}}
\def\th{{\mathbf h}}

\newcommand{\rvVec}[1]{{\mathbf{#1}}}
\newcommand{\rvMat}[1]{{\mathbf{#1}}}
\renewcommand{\Vm}{\tV}
\renewcommand{\Um}{\tU}
\renewcommand{\Rm}{\tR}
\renewcommand{\Xm}{\tX}
\renewcommand{\Am}{\tA}
\renewcommand{\Bm}{\tB}
\renewcommand{\Pm}{\tP}
\renewcommand{\Mm}{\tM}

\renewcommand{\Id}{\tI}

\graphicspath{{figures/}}

\newif\ifShort
\Shortfalse  
\newcommand{\ShortLongVersion}[2]{\ifShort #1 \else #2 \fi}

\makeatletter%
\if@twocolumn%
\newcommand{\Figwidth}{0.95\columnwidth}%

\def\twocolbreak{\nonumber\\ &}%
\def\twocolAlignMarker{&}%
\else
\newcommand{\Figwidth}{4.5in}%

\def\twocolbreak{}%
\def\twocolAlignMarker{}%
\fi%
\makeatother%

\pgfplotsset{compat=1.15}
\begin{document}
\title{Transmit Correlation Diversity: Generalization, New Techniques, and Improved Bounds}
\author{Fan~Zhang,~\IEEEmembership{Student Member,~IEEE}, Khac-Hoang~Ngo,~\IEEEmembership{ Member,~IEEE},  Sheng~Yang,~\IEEEmembership{Member,~IEEE},  and Aria~Nosratinia,~\IEEEmembership{Fellow,~IEEE}\thanks{Fan Zhang and Aria Nosratinia are with the University of Texas at Dallas, Texas, USA. Khac-Hoang Ngo is with Chalmers University of Technology, Gothenburg, Sweden and Sheng Yang is with CentraleSup\'{e}lec, Paris-Saclay University, Gif-sur-Yvette, France.}\thanks{This work was supported in part by the grants 1527598 and 1718551 from the National Science Foundation.}
\thanks{The material in this paper was presented in part at the IEEE International Symposium on Information Theory (ISIT), Aachen, Germany, 2017,
the IEEE Information Theory Workshop (ITW), Kaohsiung, Taiwan, 2017,
and the IEEE International Symposium on Information Theory (ISIT), Colorado, USA, 2018.%
}
}

\maketitle

\begin{abstract}
When the users in a MIMO broadcast channel experience different spatial transmit correlation matrices, a class of gains is produced that is denoted {\em transmit correlation diversity.} This idea was conceived for channels in which transmit correlation matrices have mutually exclusive eigenspaces, allowing non-interfering training and transmission.
This paper broadens the scope of transmit correlation diversity to the case of {\em partially and fully overlapping} eigenspaces and introduces techniques to harvest these generalized gains. For the two-user MIMO broadcast channel, we derive achievable degrees of freedom (DoF) and achievable rate regions with/without channel state information at the receiver~(CSIR). When CSIR is available, the proposed achievable DoF region is tight in some configurations of the number of receive antennas and the channel correlation ranks. We then extend the DoF results to the $K$-user case by analyzing the interference graph that characterizes the overlapping structure of the eigenspaces. 
Our achievability results employ a combination of product superposition in the common part of the eigenspaces, and pre-beamforming (rate splitting) to create multiple data streams in non-overlapping parts of the eigenspaces.
Massive MIMO is a natural example in which spatially correlated link gains are likely to occur. We study the achievable downlink sum rate for a frequency-division duplex massive MIMO system under transmit correlation diversity.

\end{abstract}
\section{Introduction} \label{sec:introduction}

The effect of spatial correlation on the capacity of MIMO links has been a subject of long-standing interest. Spatial correlation arises in part from propagation environments producing stronger signal gains in some spatial directions than others, and in part from spatially dependent patterns of the antennas. The interest in spatial correlation was sharpened by its experimental validation~\cite{Kermoal2002stochasticMIMO,KaiYu2004modeling}, and more recently by the increasing attention to higher microwave frequencies and larger number of antennas. 

Shiu {\em et al.}~\cite{Shiu_2000} proposed an abstract ``one-ring'' model for the spatial fading correlation and its effect on the MIMO capacity. In single-user channels with channel state information at the receiver~(CSIR) but no channel state information at the transmitter~(CSIT), channel correlation can boost power but may reduce the degrees of freedom~(DoF)~\cite{Jafar_2004,Jorswieck_2004}, thus it can be detrimental at high signal-to-noise ratio~(SNR) but a boon at low SNR. Tulino {\em et al.}~\cite{Tulino_2005} derived analytical characterizations of the capacity of correlated MIMO channels for the large antenna array regime.
Chang {\em et al.}~\cite{chang_2006} showed that channel rank deficiency due to spatial correlation lowers the diversity-multiplexing tradeoff curves from that of uncorrelated channel. Capacity bounds subject to channel estimation errors in correlated fading have been characterized~\cite{Anese_2009,Soysal_2010}. 
{Of the rich broader literature on MIMO spatial correlation, we are able to mention only a few representative examples~\cite{1237141,995514} in the interest of brevity.}

The sum-rate capacity under user-specific transmit correlations with CSIR was studied in~\cite{Lee_2009,Lee_2010}. Under the assumption that all users experience {\em identical} correlation, Al-Naffouri {\em et al.}~\cite{Naffouri_2009} showed that correlation is detrimental to the sum-rate scaling of the MIMO broadcast channels under certain transmission schemes. In practice, however, users may have non-identical correlation matrices because they are not co-located~\cite{Abdi_2002}, making it difficult to draw conclusions based on~\cite{Naffouri_2009}.  Furthermore, at higher frequencies or with large number of antennas, when spatial correlation is unavoidable, comparing capacity against a hypothetically uncorrelated channel may have limited operational impact. Instead, a more immediate question could be: how to maximize performance in the presence of spatial correlation? A useful tool for that purpose is {\em transmit correlation diversity}, i.e., leveraging the difference between the spatial correlation observed by different users in the system 
{in the interest of exploring and exploiting economies in training and pilots}.

Transmit correlation diversity was originally conceived for transmit spatial correlation matrices that have mutually exclusive eigenspaces.\footnote{
{The phrase {\em Transmit correlation diversity} is employed in a narrow sense, describing a class of gains that are related to economy of training and pilots, and have been a subject of relatively recent interest. This is in contrast with the broader set of beamforming techniques in the presence of antenna correlation, which have a longer pedigree in wireless communication.}}
Under this condition, a joint spatial division multiplexing~(JSDM) transmission scheme was proposed~\cite{Nam_2012,Nam_2014,Nam_2014_2,Nam_2017} that reduces the overhead needed for channel estimation. For multi-user networks with orthogonal eigenspace correlation matrices, Adhikary and Caire~\cite{Adhikary_2014} showed that transmit correlation helps in multi-cell network by partitioning the user spaces into clusters according to correlation.  It is also known that transmit correlation benefits the sum rate in the downlink performance of a heterogeneous cellular network (HetNet) where both macro and small cells share the same spectrum~\cite{Adhikary_2015}.  Non-overlapping transmit correlation eigenspaces have also been exploited in a two-tier system where a large number of small cells are deployed under a macro cell~\cite{Adhikary_2014_2}.

Except for severely rank-deficient MIMO links and relatively small number of users, in most other scenarios transmit correlation matrices corresponding to different receivers have eigenspaces whose intersection is non-trivial, i.e., they experience some overlap. This creates a natural motivation to explore and understand transmit correlation diversity in the more general setting.
%
%
This paper broadens the scope of correlation diversity and introduces methods to harvest correlation diversity gains under these broader channel conditions.


The main contributions of this work are summarized as follows.
\begin{enumerate}
\item We derive achievable DoF regions for the two-user broadcast channel in spatially correlated fading under the CSIR~(Theorem~\ref{thm:2user_CSIR_partiallyOverlapping}) and no free CSIR~(Theorem~\ref{thm:2user_CDIR_partiallyOverlapping}) assumptions. These regions are significantly larger than the time division multiple access~(TDMA) region, especially when the rank $r_0$ of the overlap between two correlation eigenspaces is large (see Fig.~\ref{fig:DoF_2user_CSIR} and Fig.~\ref{fig:dof_2pwo}). In the CSIR case, we also found an outer bound~(Theorem~\ref{thm:2user_CSIR_outerBound}) which shows that our achievable DoF region is tight under certain conditions. 

\item For the two-user broadcast channel, we propose an achievable rate region for arbitrary input distribution satisfying the power constraint~(Lemma~\ref{lemma:BC_rate_region}). We characterize this rate region with an explicit input distribution based on orthogonal pilots and Gaussian data symbols. We also derive the rate achieved with product superposition (Section~\ref{sec:2user_rate_prod}) and a hybrid of pre-beamforming and product superposition (Section~\ref{sec:2user_rate_hybrid}). As a by-product, we find the rate achieved with pilot-based schemes for the point-to-point channel (Theorem~\ref{thm:single_user}), which generalizes the result of Hassibi and Hochwald~\cite{Hassibi_2003} to correlated fading.

\item We derive achievable DoF regions for the $K$-user broadcast channel in spatially correlated fading in the presence of CSIR (Theorem~\ref{thm:K_wCSIR}), as well as without free CSIR under fully overlapping eigenspaces (Theorem~\ref{thm:kuser_fully}), symmetrically partially overlapping eigenspaces (Theorems~\ref{thm:K_sym_nCSIR}, \ref{thm:K_nCSIR}) or general correlation structure (Theorem~\ref{thm:Kuser_CDIT_DoF_hybrid}).

\item We analyze the sum rate of a massive MIMO system operating in FDD mode by investigating the pilot reduction and opportunistic additional data transmission that is made possible by spatial correlation.
\end{enumerate}
For the achievability results above, we employ pre-beamforming, product superposition, or a combination thereof, in the process demonstrating that these transmission techniques can harvest transmit correlation diversity gains under partially-overlapping eigenspaces. 
For the most part, our results do not require the fading to be Rayleigh; they hold for a wider class of fading such that the channel matrix has finite entropy and finite power.
Early versions of the results of this paper appeared in~\cite{Fan_2017,Ngo_2017,Fan_2018}.

\textit{Notation}: Bold lower-case letters, e.g. $\tx$, denote column vectors. Bold upper-case letters, e.g. $\tM$,  denote  matrices.
The Euclidean norm is denoted by $\|\tx\|$ and the Frobenius norm $\|\Mm\|_F$. 
The trace, conjugate,
transpose and conjugated transpose of $\Mm$ are denoted $\trace[\Mm]$, 
$\Mm^*$, $\Mm^\T$ and
$\Mm^\H$, respectively; $\Mm^{-\T} \defeq (\Mm^{-1})^\T$ and $\Mm^{-\H} \defeq (\Mm^{-1})^\H$; $\Id_m$ and $\mathbf{0}_{m \times n}$ denote the $m\times m$ identity matrix and $m\times n$ zero matrix, respectively, and the dimensions are omitted if cleared from the context; 
$\Mm_{[i:j]}$ denotes the sub-matrix containing columns from $i$ to $j$ of $\Mm$, and $\Mm_{[i]}$ denotes the $i$-th column; $(\tx)_{[i:j]}$ and $(\tx^\T)_{[i:j]}$ denotes respectively the column vector and row vector containing entries from $i$  to $j$ of a column vector $\tx$;  $\Span[\tU]$ denotes the subspace spanned by the columns of a truncated unitary matrix $\tU$ and $\Span[\tU]^{\perp}$ denotes the subspace that is orthogonal to $\Span[\tU]$; $\diag(x_1,\dots,x_n)$ is a diagonal matrix with diagonal entries $x_1,\dots,x_n$;
$[n] \defeq \{1,2,\dots,n\}$; $(x)^+ \defeq \max\{x,0\}$; 
$\mathbbm{1}\{A\}$ is the indicator function of event $A$.
Logarithms are in base $2$. All rates are measured in bits per channel use. 
 


\section{System Model}
\label{sec:sys}
Consider a MIMO broadcast channel in which a transmitter (also called as base station) equipped with $M$ antennas transmitting to $K$ receivers (also called as users), where User~$k$ is equipped with $N_k$ antennas, $k\in [K]$. The received signal at User~$k$ at channel use $j$ is
\begin{equation}
\label{eq:syssig}
\ty_k(j) = \tH_k(j)\tx(j) + \tw_k(j), \qquad \text{for~~}k\in [K],\ j  = 1,2,\dots
\end{equation}
where $\tx(j) \in \mathbb{C}^{M}$ is the transmitted signal at channel use $j$ and $\tw_k \in \mathbb{C}^{N_k}$ is the white noise with independent and identically distributed~(i.i.d.) $\Cc\Nc(0,1)$ entries. $\tH_k(j) \in \CC^{N_k\times M}$ is the channel matrix containing the random fading coefficients between $M$ transmit antennas of the base station and $N_k$ receive antennas of User~$k$. We assume that $\frac{1}{MN_k}\E[{\|\tH_k\|^2}] = 1, k\in [K]$. The transmitted signal is subject to the power constraint
\begin{align}
\frac{1}{J} \sum_{j = 1}^{J} \|\tx(j)\|^2 \le \rho,
\end{align}
where $J$ is the number of channel uses spanned by a codeword (of a channel code). Therefore, $\rho$ is the ratio between the average transmit power per antenna and the noise power, and is referred to as the SNR of the channel. Hereafter, we omit the channel use index $j$. 

\subsubsection{Channel Correlation} 
\label{sec:correlation_model}

We assume that the channel is spatially correlated according to the Kronecker model (a.k.a. separable model), and focus on the transmit-side correlation. Thus the channel matrices are expressed as
\begin{align}
\tH_k = \breve{\tH}_k \tR_k^{\frac{1}{2}}, \quad k\in [K], \label{eq:Kronecker_1}
\end{align}
where $\tR_k = \frac{1}{N_k} \E[\tH_k^\H \tH_k] {  \in \CC^{M\times M}}$, $\trace[\tR_k] = M$, is the transmit correlation matrix of User~$k$ with rank $r_k$, and $\breve{\tH}_k \in \mathbb{C}^{N_k \times M}$ is drawn from a generic distribution satisfying the conditions 
\begin{align}
h(\breve{\tH}_k) > -\infty, \quad \E[\breve{\tH}_k^\H \breve{\tH}_k] = N_k \tI_M, \quad k \in [K].
\end{align}

Since the correlation matrices might be rank-deficient, $\breve{\tH}_k$ is not necessarily a minimal representation of the randomness in $\tH_k$.  
The correlation eigenspace of User~$k$ is revealed via eigendecomposition of the correlation matrix:
\begin{align}
\tR_k = \tU_k \Sigmam_k \tU_k^\H,
\end{align}
where $\Sigmam_k$ is a $r_k \times r_k$ diagonal matrix containing $r_k$ non-zero eigenvalues of $\tR_k$, and $\tU_k$ is a $M \times r_k$ matrix whose orthonormal unit column vectors are the eigenvectors of $\tR_k$ corresponding to the non-zero eigenvalues. The rows of $\tH_k$ belong to the $r_k$-dimensional eigenspace $\Span(\tU_k)$ of $\tR_k$, also called as the eigenspace of User~$k$.

The channel expression \eqref{eq:Kronecker_1} can be expanded as
\begin{align}
\tH_k = \breve{\tH}_k \tU_k \Sigmam_k^{\frac12} \tU_k^\H = \tG_k \Sigmam_k^{\frac12} \tU_k^\H, \label{eq:Kronecker_2}
\end{align}
where $\tG_k \defeq \breve{\tH}_k \tU_k$ is equivalently drawn from a generic distribution satisfying $h(\tG_k) > -\infty$, $\E[\tG_k^\H \tG_k] = N_k \tI_{r_k}$, $k \in [K]$.

The eigenspaces $\Span(\tU_k)$ have a prominent role in transmit correlation diversity. For example, methods such as~\cite{Nam_2012,Nam_2014,Nam_2014_2,Nam_2017} are critically dependent on finding groups of users whose eigenspaces have no intersection. In contrast, in this paper, we propose transmission schemes that take advantage of both common and non-common parts of the eigenspaces. To this end, in several instances, we build an equivalent channel $\bar{\tH}_k$ that resides in a {\em subspace} of the eigenspace $\Span(\tU_k)$ via the linear transformation
\begin{align}
\bar{\tH}_k = \tH_k \tV_k,
\end{align}
for some truncated unitary matrix $\tV_k \in \CC^{M\times s_k}$, $s_k \le r_k$, such that~ $\Span(\tV_k) \subset \Span(\tU_k)$.
Unlike $\tU_k$, $k\in [K]$, that characterize the correlation eigenspaces of the links, the subspaces $\Span(\tV_k)$ also depend on the proposed transmission schemes and may be customized throughout the paper.

\subsubsection{Channel Information Availability}
We assume throughout the paper that the distribution of $\tH_k$, in particular the second-order statistic $\tR_k$ (and thus $\Sigmam_k$ and $\tU_k$), is known to both the base station and User~$k$. This is reasonable because $\tR_k$ represents long-term behavior of the channel that is stable and can be easily tracked. On the other hand, the realization of $\tH_k$ changes much more rapidly. We consider two scenarios: 
\begin{itemize}
\item CSIR~(channel state information at the receiver): User~$k$ knows perfectly the realizations of $\tH_k$.
\item No free CSIR: User~$k$ only knows the distribution of $\tH_k$.
In this case, for a tractable model of the channel variation, we assume a block fading model with equal-length and synchronous coherence interval (across the users) of $T$ channel uses. That is, $\tH_k$ remains constant during each block of length $T$ and changes independently across blocks \cite{Zheng_2002}. We assume that $T \ge 2\max (r_k,N_k), \forall k$. Let $\tX = [\tx(1)\ \dots\ \tx(T)]$ be the transmitted signal during a block, the received signal at User~$k$ during this block is
\begin{align} \label{eq:rx_signal_blockfading}
\tY_k = \tH_k \tX + \tW_k,
\end{align}
where $\tY_k = [\ty_k(1)\ \dots\ \ty_k(T)]$, $\tW_k = [\tw_k(1)\ \dots\ \tw_k(T)]$, and the block index is omitted for simplicity. User~$k$ might attempt to estimate $\tH_k$ with the help of known pilot symbols inserted in $\tX$.
\end{itemize}

\subsubsection{Achievable Rate and DoF}
Assuming $K$ independent messages are communicated (no common message), {and the corresponding rate tuple $(R_1(\rho),\dots,R_K(\rho))$ is achievable at SNR $\rho$, $\forall \rho \ge 0$, i.e., lie within the capacity region of the channel, then an achievable DoF tuple $(d_1,\dots,d_K)$ is defined as}
\begin{equation}
d_k \triangleq \lim_{\rho \rightarrow \infty} \frac{R_k(\rho)}{\log \rho},\quad k \in [K].
\end{equation}
The set of achievable rate (resp., DoF) tuples defines an achievable rate (resp., DoF) region of the channel.

For convenience, we denote $N_k^{\ast} \defeq \min (N_k,r_k)$.

\section{Preliminaries and Useful Results} \label{sec:preliminaries}

\begin{lemma} [The optimal single-user DoF] \label{lem:single-user}
For the correlated MIMO broadcast channel in Section~\ref{sec:correlation_model}, the optimal single-user DoF of User~$k$ is $d_k = N_k^*$ with CSIR and $d_k = N_k^* \Big(1-\frac{N_k^*}{T}\Big)$ without free CSIR.
\end{lemma}
The result in the CSIR case is well-known (see, e.g., \cite{Chiani2003capacity_spatially_correlated_MIMO}). The no free CSIR case was reported in \cite[Thm.~1]{Fan_2017}. The next lemma is used for the finite-SNR rate analysis.

\begin{lemma}[Worst case uncorrelated additive noise~{\cite{Hassibi_2003}}]	\label{lem:worst-case-noise}
Consider the point-to-point channel
\begin{align}
\rvVec{y} = \sqrt{\frac{\rho}{M}} \rvMat{H} \rvVec{x} + \rvVec{w}, \label{eq:channel_general_noise}
\end{align}
where the channel $\rvMat{H} \in \CC^{N\times M}$ is known to the receiver, and the signal $\rvVec{x} \in \CC^{M\times 1}$ and the noise $\rvVec{w}\in \CC^{N \times 1}$ satisfy the power constraints $\frac{1}{M}\E[\|\rvVec{x}\|^2] = 1$ and $\frac{1}{N}\E[\|\rvVec{w}\|^2] = 1$, are both complex Gaussian distributed, and are uncorrelated, i.e, $\E[\rvVec{x}\rvVec{w}^\H] = \mathbf{0}$. Let $\Rm_{\rvVec{x}} \defeq \E[\rvVec{x}\rvVec{x}^\H]$ and $\Rm_{\rvVec{w}} \defeq \E[\rvVec{w}\rvVec{w}^\H]$ and assume $\trace[\Rm_{\rvVec{x}}]=M$ and $\trace[\Rm_{\rvVec{w}}]=N$. Then the mutual information $I(\rvVec{y}; \rvVec{x} \cond \rvMat{H})$ is lower bounded as
\begin{align}
I(\rvVec{y}; \rvVec{x} \cond \rvMat{H}) &\ge \E[{\log\det[\Id_N + \frac{\rho}{M} \Rm_{\rvVec{w}}^{-1} \rvMat{H} \Rm_{\rvVec{x}} \rvMat{H}^\H]}] 
\label{eq:worst-case-noise-1}\\
&\ge \min_{\Rm_{\rvVec{w}},\trace[\Rm_{\rvVec{w}}] = N} \E[{\log\det[\Id_N + \frac{\rho}{M} \Rm_{\rvVec{w}}^{-1} \rvMat{H} \Rm_{\rvVec{x}} \rvMat{H}^\H]}]. 
\label{eq:worst-case-noise-2}
\end{align}
If the distribution of $\rvMat{H}$ is left rotationally invariant, i.e., $p(\Thetam\rvMat{H}) = p(\rvMat{H})$ for any deterministic $N\times N$ unitary matrix $\Thetam$, then the minimizing noise covariance matrix in \eqref{eq:worst-case-noise-2} is $\Rm_{\rvVec{w},{\rm opt}} = \Id_N$.

\end{lemma}
\begin{proof}
The proof follows from the proof of \cite[Thm.~1]{Hassibi_2003}. Specifically, the mutual information lower bound \eqref{eq:worst-case-noise-1} was stated in \cite[Eq.~(27)]{Hassibi_2003}. To show that $\Rm_{\rvVec{w},{\rm opt}} = \Id_N$, we diagonalize $\Rm_{\rvVec{w}}$ using the left rotational invariance of $\rvMat{H}$, and then use the convexity of $\E[{\log\det[\Id_N + \frac{\rho}{M} \Rm_{\rvVec{w}}^{-1} \rvMat{H} \Rm_{\rvVec{x}} \rvMat{H}^\H]}]$ in the diagonalized $\Rm_{\rvVec{w}}$. 
\end{proof}

The next lemma gives the MMSE estimator used for pilot-based channel estimation without free CSIR.
\begin{lemma}[MMSE estimator] \label{lem:MMSE-estimator}
Consider the following linear model
\begin{align}
\rvMat{Y} =  \rvMat{H} \Xm + \rvMat{W},
\end{align}
where $\rvMat{H} \in \CC^{N\times M}$ has correlation matrix $\Rm = \frac{1}{N}\E[\rvMat{H}^\H \rvMat{H}]$, $\Xm \in \CC^{M\times M}$ is known, and $\rvMat{W} \in \CC^{N\times M}$ has i.i.d. $\Cc\Nc(0,1)$ entries. The linear MMSE estimator for $\rvMat{H}$ is given by
\begin{align} \label{eq:MMSE-estimate}
\hat{\rvMat{H}} &= \rvMat{Y} (\Xm^\H \Rm \Xm + \Id_M)^{-1} \Xm^\H \Rm.
\end{align}
The MMSE estimate $\hat{\rvMat{H}}$ is also the conditional mean: $\hat{\rvMat{H}} = \EE [\rvMat{H} \cond \Xm,\rvMat{Y}]$. The estimate $\hat{\rvMat{H}}$ and the estimation error $\tilde{\rvMat{H}} = \rvMat{H} - \hat{\rvMat{H}}$ are uncorrelated, have zero mean and row covariance 
\begin{align}
\frac{1}{N}\EE [\hat{\rvMat{H}}^\H \hat{\rvMat{H}}] &= \Rm \Xm (\Xm^\H \Rm \Xm + \Id_M)^{-1} \Xm^\H \Rm, \label{eq:tmp181} \\
\frac{1}{N}\EE [\tilde{\rvMat{H}}^\H \tilde{\rvMat{H}}] &= \Rm - \Rm \Xm (\Xm^\H \Rm \Xm + \Id_M)^{-1} \Xm^\H \Rm 
. \label{eq:tmp182}
\end{align}
\end{lemma}
\begin{proof}
The linear MMSE channel estimator is given by $\hat{\rvMat{H}} = \rvMat{Y} \Am$ where $\Am$ is the minimizer of the MSE
\begin{align}
\frac{1}{N} \EE [\|\rvMat{H} - \hat{\rvMat{H}}\|_F^2] = \trace[\Rm] - \trace[\Rm \Xm \Am] - \trace[\Am^\H \Xm^\H \Rm] + \trace[\Am^\H(\Xm^\H \Rm \Xm + \Id_M)\Am].
\end{align}
Solving $\frac{\partial}{\partial \Am} \frac{1}{N} \EE [\|\rvMat{H} - \hat{\rvMat{H}}\|_F^2] = 0$ yields the optimal $\Am_{\rm opt} = (\Xm^\H \Rm \Xm + \Id_M)^{-1} \Xm^\H \Rm$. Some further simple manipulations give \eqref{eq:tmp181} and \eqref{eq:tmp182}.
\end{proof}

In the following, we introduce two main building blocks of our proposed achievable schemes.
\subsection{Rate Splitting and Precoder Design}
\label{sec:intro_ratesplit}
To illustrate the basic idea of rate splitting, we take a two-user broadcast for example. Define $r_0 = \rank(\Span(\tU_1) \cap \Span(\tU_2))$.  Let $\tV_k$ be the precoding matrix. The transmitted signal is
\begin{equation}
    \tX = \tV_1\tX_1 + \tV_2\tX_2 + \tV_0\tX_0.
\end{equation}
Each of the signals $\tX_i$ contains an information-carrying matrix. The precoder matrices $\tV_i$ are designed satisfying the following properties:
\begin{align}
    \rank(\tU_1\tV_1) & = r_1,~
    \rank(\tU_2\tV_2)  = r_2, \\
    \rank(\tU_1\tV_2) & = \rank(\tU_2\tV_1)= 0, \\
    \rank(\tU_1\tV_0) & = \rank(\tU_2\tV_0) = r_0.
\end{align}
 This property ensures that the receiver only sees the signal that transmit along the directions which are not orthogonal to its eigendirections. In this case, it indicates that receiver~1 can see $\tX_1$ and $\tX_0$, while receiver~2 can see $\tX_2$ and $\tX_0$. 
 The precoder $\tV_0$ can be calculated from $\Um_1$ and $\Um_2$ using, e.g., the Zassenhaus algorithm~\cite{Luks1997AlgoNilpotent}. Specifically, this algorithm uses elementary row operations to transform the $(r_1 + r_2) \times 2M$ matrix
$\Bigg[
\begin{matrix}
\Um_1^\T & \Um_1^\T \\
\Um_2^\T & \mathbf{0}_{r_2 \times M} 
\end{matrix}\Bigg]$
(or $\Bigg[\begin{matrix}
\Um_2^\T & \Um_2^\T \\
\Um_1^\T & \mathbf{0}_{r_1 \times M} 
\end{matrix}\Bigg]$) to the row echelon form
$\left[
\begin{matrix}
{\Vm}_0^\T & \pmb{*} \\
\mathbf{0}  & {\Vm}_0^\T \\
\mathbf{0} & \mathbf{0}
\end{matrix}\right]$,
where $\pmb{*}$ stands for a matrix which is not of interest. 
The precoders ${\Vm}_1$ and ${\Vm}_2$ can be found similarly by applying the Zassenhaus algorithm to $\Um_1$ and $\Null{\Um_2}$, and $\Null{\Um_1}$ and ${\Um_2}$, respectively, where $\Null{\Um_k}$ is the matrix such that $[\Um_k \ \Null{\Um_k}]$ is unitary.

\subsection{Product Superposition}
In \cite{Li_2012,Li_2015}, Li and Nosratinia studied a two-receiver broadcast  one static receiver has non-identical coherence times and proposed a product superposition scheme. In the earlier work of \cite{Fan_2017}, the product superposition scheme was implemented in a two-receiver broadcast channel when two receivers have non-identical transmit correlation. Assume a two-user broadcast channel has one receiver with uncorrelated channel and the other receiver with rank-deficient correlated channel with rank $r_2$, to apply product superposition, the transmitter sends the signal 
\begin{equation}
\tX = \sqrt{\rho}\tX_2 [\tI_{r_1}, \tS_1],
\end{equation}
where 
\begin{align}
\tX_2 =
\begin{bmatrix}
[\tI{r_2}, \tS_2] \\
\tS_0
\end{bmatrix},
\end{align}
$\tS_1 \in \mathbb{C}^{r_1 \times (T - M)}$ contains symbol intended for User~1, $\tS_0 \in \mathbb{C}^{(M - r_2) \times r_1}$ is designed to guarantee that $\tX_2$ is non-singular, and $\tS_2 \in \mathbb{C}^{r_2 \times M}$ includes symbol intended for Receiver~2. The received signal at User~1 is 
\begin{align}
\tY_1 = \sqrt{\rho}\tilde{\tH}_1  [\tI_{r_1}, \tS_1] + \tW_1,
\end{align}
where $\tilde{\tH}_1 = \bar{\tH}_1  {\boldsymbol \Sigma}_1 \tX_2$. Receiver~1 estimates the equivalent channel $\tilde{\tH}_1$ and decodes $\tS_1$, achieving $\min(N_1,M)(T - M)$ degrees of freedom. The received signal at Receiver~2, during the first $M$ time slots, is
\begin{align}
\tY_2^{\prime} = \sqrt{\rho}\tilde{\tH}_2 [\tI{r_2}, \tS_2] + \tW_2^{\prime},
\end{align}
where $\tilde{\tH}_2 = \bar{\tH}_2  {\boldsymbol \Sigma}_2$. Using the first $r_2$ columns, Receiver~2 estimates the channel, $\tilde{\tH}_2$, and furthermore using the remaining $(M - r_2)$ columns, Receiver~2 decodes the symbols, achieving $\min (N_2,r_2) (M - r_2)$ degrees of freedom.

\color{black}
\section{Two-user Broadcast Channel: DoF Analysis}
\label{sec:2-user_DoF}
Both with or without free CSIR assumption, we study first the special case of fully overlapping correlation eigenspaces, then the more general case of partially overlapping correlation eigenspaces.

\subsection{CSIR}
\label{sec:2w}
Consider the case where both users have spatially correlated channels, and User~$2$'s channel eigenspace is a subspace of User~$1$'s, which implies $r_2 \leq r_1 \leq  M$. 

\begin{proposition}
\label{thm:2fw} 
For the two-user broadcast channel with CSIR, when the eigenspace of User~$2$ is a subspace of User~$1$'s~(implying $r_2 \leq r_1 \leq M$), the DoF pairs $(N_1^*,0), (0,N_2^*)$, and $\big((N_1^*-r_2)^+,N_2^*\big)$ are achievable. Furthermore, if $r_1 \ge N_1 \ge r_1 - r_2$, the DoF pair $\big(r_1 - r_2, \min(N_1 - r_1 + r_2,N_2^*) \big)$ is also achievable. The convex hull of these pairs and the origin $(0,0)$ is an achievable DoF region.
\end{proposition}

\ShortLongVersion{The proof is available in~\cite{arXivVersion:Fan2021} and is omitted here for brevity.}{
\begin{proof}
According to Lemma~\ref{lem:single-user}, the DoF pairs $(N_1^*,0)$ and $(0,N_2^*)$ are achievable. 

When $N_1^* \ge r_2$, the pair $(N_1^*-r_2, N_2^*)$ can be achieved as follows.

Recall that the eigenspaces of channels $\tH_1$ and $\tH_2$ are $\Span(\tU_1)$ and $\Span(\tU_2)$, respectively, and in the present case, $\Span(\tU_2) \subset \Span(\tU_1)$. There exist transmit eigendirections $\tV_1 \in \CC^{M \times (N_1^* - r_2)}, \tV_0 \in \CC^{M \times N^*_2}$ that are aligned with the common and non-common parts of the two channel eigenspaces 
such that
\begin{align}
\Span(\tV_0) &\subset \spantwo,\\
\spanvone &\subset \big(\spanone \cap \spantwo^\perp\big). \label{eq:tmp315}
\end{align}

Define $\tV \triangleq [\tV_0 \ \tV_1]$. The proposed transmission scheme is $\tx = \tV \big[\ts_0^\T \ \ts_1^\T\big]^\T$
where the signals $\ts_1 \in \mathbb{C}^{N_1^*-r_2},\ts_0 \in \mathbb{C}^{N_2^*}$ are intended for User~$1$ and User~$2$, respectively. 
The received signal at User~$1$ is 
\begin{equation}
\begin{split}
\ty_1 & =\tH_1\tx+ \tw_1  = {\tH}_1  \tV \Bigg[\begin{matrix}
\ts_0 \\
\ts_1
\end{matrix}\Bigg]
+ \tw_1.
\end{split}
\end{equation} 
Since User~$1$ knows ${\tH}_1  \tV$, it can decode both $\ts_1$ and $\ts_0$, achieving respectively $N_1^*-r_2$ and $N_2^*$ DoF.
The received signal at User~$2$ is
\begin{align}
\ty_2 =\tH_2\tx+ \tw_2 
 = {\tH}_2 [\tV_0 \ \tV_1] 
\Bigg[\begin{matrix}
\ts_0\\
\ts_1
\end{matrix}\Bigg]
+ \tw_2
 ={\tH}_2 \tV_0 \ts_0 + \tw_2,
\end{align}
which uses $\tH_2 \tV_1 = \mathbf{0}$ due to \eqref{eq:tmp315}. Since User~$2$ knows $\tH_2 \tV_0$, it can decode $\ts_2$, achieving $N_2^*$ DoF. By dedicating $\ts_2$ to user 2, the DoF pair $(N_1^*-r_2, N_2^*)$ is achieved. 

The pair $\big(r_1 - r_2, \min(N_1 - r_1 + r_2,N_2^*) \big)$ can be achieved similarly when $r_1 \ge N_1 \ge r_1 - r_2$ by setting $\tV_1 \in \CC^{M \times (r_1 - r_2)}, \tV_0 \in \CC^{M \times \min(N_1 - r_1 + r_2,N_2^*)}$, and the dimensions of $\ts_1, \ts_0$ accordingly. 
\end{proof}
}

 
\begin{theorem}
\label{thm:2user_CSIR_partiallyOverlapping} {}
For the two-user broadcast channel with CSIR and ${\rm rank}(\spanone \, \cap\, \spantwo) = r_0 \ge 0$, the DoF pairs $(N_1^*,0), (0,N_2^*)$, $\big((N_1^* - r_0)^+, N_2^*\big)$, and $\big(N_1^*, (N_2^* - r_0)^+\big)$ are achievable. Furthermore, if $N_1 \le r_1$ and $N_2 \le r_2$, the DoF pairs
\begin{align}
&\Big(\min\big(N_1,r_1-r_0\big) + \min\big((N_1-r_1 + r_0)^+,(N_2-r_2 + r_0)^+\big), \ \min\big(N_2,r_2-r_0\big)\Big), \label{eq:2-user_CSIR_point3}\\
&\Big(\min\big(N_1,r_1-r_0\big), \ \min\big(N_2,r_2-r_0\big) + \min\big((N_1-r_1 + r_0)^+, (N_2-r_2 + r_0)^+\big)\Big), \label{eq:2-user_CSIR_point4}
\end{align}
are also achievable.
The convex hull of these pairs and the origin $(0,0)$ is an achievable DoF region.
\end{theorem}
\ShortLongVersion{The proof is available in~\cite{arXivVersion:Fan2021} and is omitted here for brevity.}{
\begin{proof}
The DoF pairs $(N_1^*,0)$ and $(0,N_2^*)$ are achievable according to Lemma~\ref{lem:single-user}. The achievable schemes for the other pairs are as follows. For non-negative integers $s_0\le r_0$, $s_1 \le r_1-r_0$, and $s_2 \le r_2-r_0$, there exist transmit eigendirections $\tV_0 \in \CC^{M \times s_0}$ aligned with the common part of the two channel eigenspaces, and eigendirections $\tV_1\in \CC^{M \times s_1}, \tV_2 \in \CC^{M \times s_2}$ aligned with the two non-common parts, such that
\begin{align}
\spanvonetwo &\subset \big(\spanone \cap \spantwo\big), \label{eq:tmp419}\\
\Span (\vvone) &\subset \big(\spanone \cap \spantwo^\perp\big), \label{eq:tmp420}\\
\Span (\vvtwo) &\subset \big(\spantwo \cap \spanone^\perp\big). \label{eq:tmp421}
\end{align} 

Define $\tV \triangleq [\tV_0 \ \tV_1 \ \tV_{2}]$. 
Let the transmitter send the signal $\tx = \tV \big[\ts_{0}^\T \	\ts_1^\T \	\ts_{2}^\T\big]^\T$, 
where $\ts_k\in \mathbb{C}^{s_k}$ contains symbols for User~$k$, $k\in \{1,2\}$, and  $\ts_0 \in \mathbb{C}^{s_0}$ contains symbols that both users can decode.

The received signal at User~$1$ and User~$2$ are respectively
\begin{align}
\ty_1 
& = \tH_1 \tx+ \tw_1
 = {\tH}_1 [\tV_0 ~ \tV_1]
\Bigg[\begin{matrix}
\ts_0\\
\ts_1
\end{matrix}\Bigg]
+ \tw_1, \\
\ty_2
&= \tH_2 \tx+ \tw_2 
= {\tH}_2 [\vvonetwo ~ \tV_{2}]
\Bigg[\begin{matrix}
\ts_0\\
\ts_{2}
\end{matrix}\Bigg]
+ \tw_2,
\end{align}
using $\tH_2 \tV_1 = \mathbf{0}$ and $\tH_1 \tV_2 = \mathbf{0}$ due to \eqref{eq:tmp420} and \eqref{eq:tmp421}, respectively. Then if $s_k + s_0 \le N_k$, User~$k$ can decode both $\ts_k$ and $\ts_0$, $k \in \{1,2\}$. 
\begin{itemize} [leftmargin=*]
\item If $N_1 \ge r_0$ and $N_2 \le r_0$, set $s_1 = N_1^*-r_0$, $s_2 = 0$, and $s_0 = N_2$. By dedicating $\ts_0$ to User~$2$, the DoF pair $(N_1^*-r_0,N_2)$ can be achieved. Similarly, if $N_1 \le r_0$ and $N_2 \ge r_0$, the DoF pair $(N_1,N_2^*-r_0)$ can be achieved.
\item If $N_1 \ge r_0$ and $N_2 \ge r_0$, set $s_1 = N_1^* - r_0$, $s_2 = N_2^* - r_0$, and $s_0 = r_0$. By dedicating $\ts_0$ to one of the users, the DoF pairs $\big(N_1^*-r_0,N_2^*\big)$ and $\big(N_1^*,N_2^*-r_0\big)$ are achievable.
\item When $N_1\le r_1$ and $N_2\le r_2$, by setting $s_1 = \min\big(N_1,r_1-r_0\big)$, $s_2 = \min\big(N_2,r_2-r_0\big)$, $s_0 = \min\big((N_1-r_1 + r_0)^+,(N_2-r_2 + r_0)^+\big)$, and dedicating $\ts_0$ to one of the users, the DoF pairs given in \eqref{eq:2-user_CSIR_point3} and \eqref{eq:2-user_CSIR_point4} are achievable.
\end{itemize}
Therefore, the proof is completed.
\end{proof}
}
An outer bound for the achievable DoF region is given as follows.
\begin{theorem} \label{thm:2user_CSIR_outerBound}
When $\Rank(\Span(\tU_1) \cap \Span(\tU_2)) = r_0 \ge 0$, the achievable DoF region is outer bounded by $d_k \le N_k^*$, $k\in \{1,2\}$, and
\begin{align}
d_1 + d_2 & \leq \min \{r_1 + r_2 - r_0, N_1 + N_2\} \label{eq:OB_sumDoF}.
\end{align}
When $\{r_1 \leq N_1, r_2 \leq N_2\}$ or $\{N_1 \leq r_1 - r_0, N_2 \leq r_2 - r_0\}$, this outer bound is tight.
\end{theorem}

\ShortLongVersion{The proof is available in~\cite{arXivVersion:Fan2021} and is omitted here for brevity.}{
\begin{proof}
The single-user bounds $d_k \le N_k^*$, $k\in \{1,2\}$, follow from Lemma~\ref{lem:single-user}. 

Denote by $\tV_1 \in \CC^{M \times (r_1 - r_0)}, \tV_2 \in \CC^{M \times (r_2 - r_0)}$ the non-unique transmit eigendirections that are aligned with the non-common parts, i.e., $\Span (\vvone) = \spanone \cap \spantwo^\perp$ and
$\Span (\vvtwo) = \spantwo \cap \spanone^\perp$, and $\tV_0 \in \CC^{M \times r_0}$ the common part, i.e., $\spanvonetwo = \spanone \cap \spantwo$, of the eigenspaces. Let $\tV_{\perp} \in \CC^{M \times (M - r_1 - r_ 2 + r_0)}$ denote the orthogonal complement of the total channel eigenspaces, i.e., $\tV \triangleq [\tV_0 ~\tV_1 ~ \tV_2 ~ \tV_{\perp}]$ is an unitary matrix. For a transmit vector $\tx \in \CC^{M}$, define $[\tx_0^\T \ \tx_1^\T \ \tx_{2}^\T \ \tx_\perp^\T]^\T \defeq \tV \tx$, 
where $\tx_0 \in \CC^{r_0}$, $\tx_1 \in \CC^{r_1 - r_0}$, $\tx_2 \in \CC^{r_2 - r_0}$ and $\tx_\perp \in \CC^{M - r_1 - r_2 + r_0}$. 

A cooperative cut-set upper bound is as follows, using invertibility of $\tV$:
\begin{equation}
\label{eq:OB_cutset}
R_1 + R_2 \leq I(\ty_1,\ty_2;\tx) = I(\ty_1,\ty_2;\tV \tx).
\end{equation}
The next step is to bound the right-hand side in \eqref{eq:OB_cutset}. To extract a full-rank representation of $\tH_1$ and $\tH_2$, 
\begin{align}
\begin{bmatrix}
\tH_1 \\ \tH_2
\end{bmatrix}
= \begin{bmatrix}
\tG_1 \Sigmam_1^{\frac12} \tU_1^\H \\ \tG_2 \Sigmam_2^{\frac12} \tU_2^\H
\end{bmatrix}
= \begin{bmatrix}
\tG_1 \Sigmam_1^{\frac12} \tT_{1c} & \tG_1 \Sigmam_1^{\frac12} \tT_{1p} & \ZeroMat \\
\tG_2 \Sigmam_2^{\frac12} \tT_{2c} & \ZeroMat & \tG_2 \Sigmam_2^{\frac12} \tT_{2p}
\end{bmatrix}
[\tV_0 \ \tV_1 \ \tV_2]^\H,
\end{align}
where $\tT_{1c}, \tT_{1p}, \tT_{2c}$, and $\tT_{2p}$ are matrices such that $[\tT_{ic} \ \tT_{ip}]$ is non-singular and
$
\tU_k^\H = [\tT_{ic} \ \tT_{ip}] [\tV_0 \ \tV_{i}]^\H, ~i\in \{1,2\}.
$
Replacing $\tx$ by $\tV \tx$, the concatenated received signal is
\begin{align}
\begin{bmatrix}
\ty_1 \\ \ty_2
\end{bmatrix}
= \begin{bmatrix}
\tH_1 \\ \tH_2
\end{bmatrix}
\tV \tx + 
\begin{bmatrix}
\tw_1 \\ \tw_2
\end{bmatrix}
= \tilde{\tH}
\begin{bmatrix}
\tx_0 \\
\tx_1 \\
\tx_2 
\end{bmatrix}
+ \begin{bmatrix}
\tw_1 \\
\tw_2
\end{bmatrix},
\end{align}
where
\begin{equation}
\tilde{\tH} \defeq 
\begin{bmatrix}
\tG_1 \Sigmam_1^{\frac12} \tT_{1c} & \tG_1 \Sigmam_1^{\frac12} \tT_{1p} & \ZeroMat \\
\tG_2 \Sigmam_2^{\frac12} \tT_{2c} & \ZeroMat & \tG_2 \Sigmam_2^{\frac12} \tT_{2p}
\end{bmatrix}
\in \CC^{(N_1 + N_2) \times (r_1 + r_2 - r_0)}.
\end{equation}
Because $\tilde{\tH}$ is known at the receivers, 
\begin{align}
I(\ty_1,\ty_2;\tV\tx)  &=  I(\ty_1,\ty_2;\tx_0,\tx_1,\tx_2) \\
&\leq \min \{r_1 + r_2 - r_0, N_1 + N_2\} \log\rho + o(\log \rho).
\end{align}
This yields the sum DoF bound $d_1 + d_2 \leq \min \{r_1 + r_2 - r_0, N_1 + N_2\}$. This outer bound is tight against the achievable region in Theorem~\ref{thm:2user_CSIR_partiallyOverlapping}.
\end{proof}
}
Fig.~\ref{fig_OB_tight} shows the regions  where the outer bound in Theorem~\ref{thm:2user_CSIR_outerBound} is tight.
\begin{figure}[!ht]
\centering
\includegraphics[width = \Figwidth]{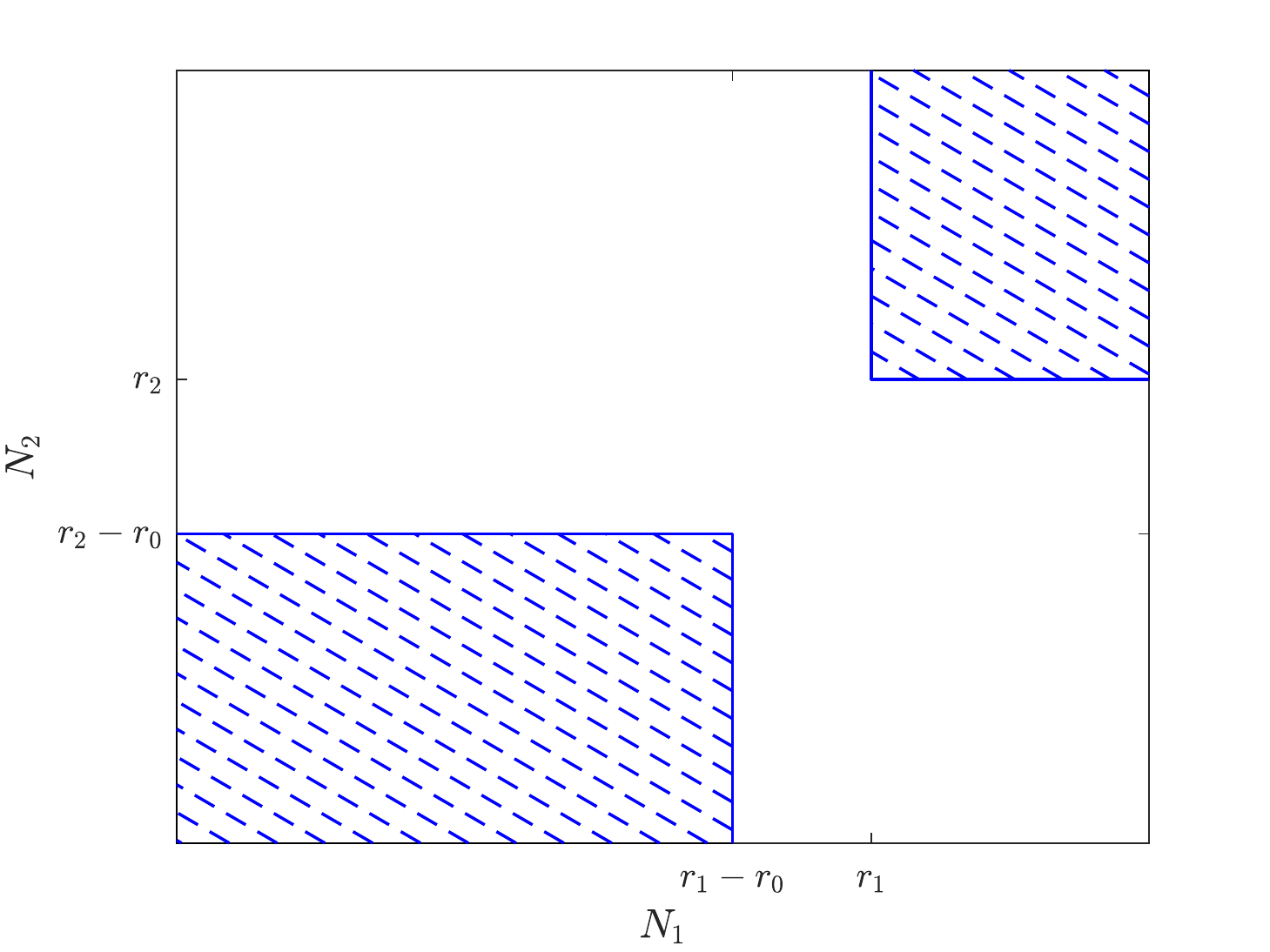}
\caption{Regions (the hashed part) where the outer bound for the DoF region with CSIR in Theorem~\ref{thm:2user_CSIR_outerBound} is tight.}
\label{fig_OB_tight}
\end{figure}

Fig.~\ref{fig:DoF_2user_CSIR} compares the achievable region proposed in Theorem~\ref{thm:2user_CSIR_partiallyOverlapping} and the achievable region achieved with TDMA (time sharing between $(N_1^*,0)$ and $(0,N_2^*)$) for  $r_1 = 12$, $r_2 = 10$, $r_0 \in \{0,3,6,9\}$ and $N_1 \ge r_1$, $N_2\ge r_2$. The proposed achievable region is much larger than the TDMA region, especially when $r_0$ is small. In this setting, according to Theorem~\ref{thm:2user_CSIR_outerBound}, the proposed region is optimal. 
\begin{figure}[!ht]
\centering
\includegraphics[width =\Figwidth]{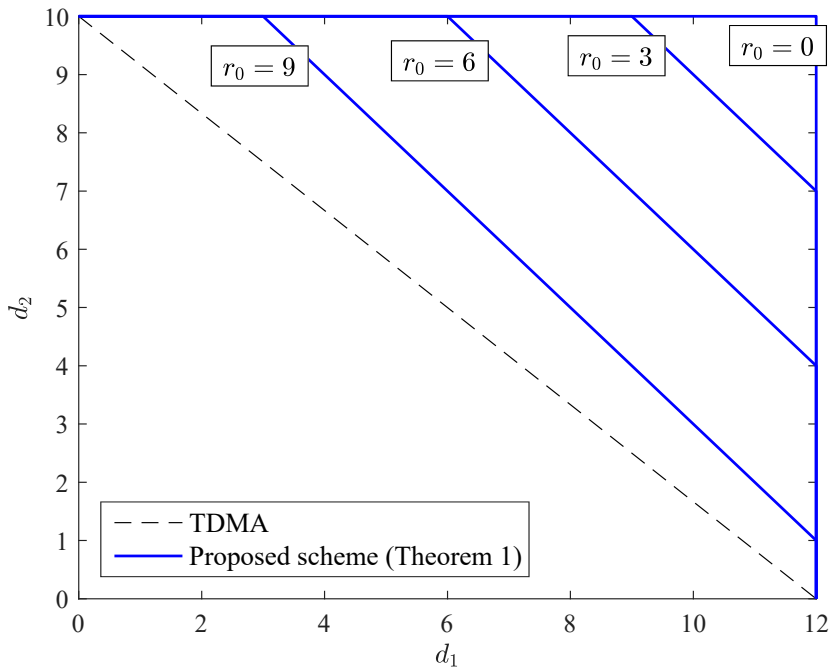} 
\caption{The achievable DoF region for two-users with CSIR, under TDMA and the proposed scheme (Theorem~\ref{thm:2user_CSIR_partiallyOverlapping}) for $r_1 = 12$, $r_2 = 10$, $r_0 \in \{0,3,6,9\}$ and $N_1 \ge r_1$, $N_2\ge r_2$. In this case, the latter region is optimal.}
\label{fig:DoF_2user_CSIR}
\end{figure}

\color{black}
\subsection{No free CSIR}
In this case, CSIR is not available a priori and must be acquired via pilot transmission. On the one hand, one needs to take into account the cost of CSI acquisition in both energy and DoF. On the other hand, pilot transmission enables product superposition~\cite{Li_2012} that can improve upon rate splitting. 

\subsubsection{Fully Overlapping Eigenspaces}
\label{sec:2fwo}
Consider the case where User~$2$'s eigenspace is a subspace of User~$1$'s, which implies $r_{2} \le r_{1} \leq M$. The following proposition presents achievable DoF with product superposition in this case.
\begin{proposition}
\label{thm:2fwo}
In a two-user broadcast channel without free CSIR, when the eigenspace of User~$2$ is a subspace of User~$1$'s (implying $r_{2} \le r_{1} \leq M$), the DoF pair $\Big( N_1^*\big(1 - \frac{r_{1}}{T}\big), N_2^* \frac{r_{1} - r_{2}}{T} \Big)$ is achievable with product superposition.
\end{proposition}

\begin{proof}
There exist transmit eigendirections $\tV_1 \in \CC^{M \times (r_1 -r_2)}$ and $\tV_0\in \CC^{M \times r_2 }$ that are aligned with the non-common and common parts, respectively, of the two channel eigenspaces such that
\begin{align}
\Span(\tV_0) & = \spantwo,\\
\spanvone &= \spanone \cap \spantwo^\perp. \label{eq:tmp803}
\end{align}
Define $\tV \triangleq [\tV_0\ \tV_1]$. 
Let the transmitter send the signal $\tX = \tV \tX_2 \tX_1$ during a coherence block,
with $\tX_1 = [\tI_{r_{1}} ~ \tS_{1}] \in \CC^{r_1 \times T}$ and $\tX_2 =
\Bigg[\begin{matrix}
\tI_{r_{2}} & \tS_{2} \\
\mathbf{0} & \tI_{r_1-r_2}
\end{matrix}\Bigg] \in \CC^{r_1 \times r_1}$, 
where 
$\tS_{1} \in \mathbb{C}^{r_{1} \times (T - r_{1})}$ contains symbols for User~$1$ and $\tS_{2} \in \mathbb{C}^{r_{2} \times (r_{1} - \rronetwo)}$ contains symbols for User~$2$. The received signal at User~$1$ is 
\begin{equation}
\begin{split}
\tY_1 & =\tH_1 \tV \tX_2 \tX_1 + \tW_1 
={\tH}_1 \tV \tX_2 [\tI_{r_{1}} ~ \tS_{1}] + \tW_1.
\end{split}
\end{equation}
User~$1$ estimates the equivalent channel ${\tH}_1 \tV \tX_2$ and then decodes $\tS_{1}$, achieving $N_1^*(T - r_{1})$ DoF.
The received signal at User~$2$ during the first $r_{1}$ channel uses is
\begin{equation}
\begin{split}
\tY_{2[1:r_1]} & ={\tH}_2  \tV 
\Bigg[\begin{matrix}
\tI_{r_{2}} & \tS_{2} \\
\mathbf{0} & \tI_{r_1-r_2}
\end{matrix}\Bigg]
\tI_{r_{1}}+\tW_{2[1:r_1]}
={\tH}_2 \tV_0 
    [\tI_{r_{2}} \quad \tS_{2}]
    + \tW_{2[1:r_1]},
\end{split}
\end{equation}
using $\tH_2 \tV_1 = \mathbf{0}$ due to \eqref{eq:tmp803}. User~$2$ estimates the equivalent channel $\tH_2 \tV_0$, and then decodes $\tS_2$, achieving $N_2^* (r_{1} \!-\! r_{2})$ DoF. 
Therefore, the normalized DoF pair $\Big( N_1^*\big(1 \!-\! \frac{r_{1}}{T}\big), N_2^* \frac{r_{1} - r_{2}}{T} \Big)$ is achievable.
\end{proof}

\subsubsection{Partially Overlapping Eigenspaces}
\begin{theorem}
\label{thm:2user_CDIR_partiallyOverlapping}
For the two-user broadcast channel without free CSIR and ${\rm rank}(\spanone \, \cap\, \spantwo) = r_0 \ge 0$, the DoF pairs $\Big(N_1^*\big(1-\frac{N_1^*}{T}\big),0\Big)$ and $\Big(0,N_2^*\big(1-\frac{N_2^*}{T}\big)\Big)$ are achievable. Furthermore, for any integers $(s_1,s_2,s_0)$ such that $0\le s_1 \le r_1-r_0$, $0\le s_2 \le r_2-r_0$, and $0\le s_0 \le r_0$,
the DoF pairs 
\begin{align}
\Dc_1 &= \bigg(\min(s_0,N_1)\frac{s_2}{T}, \min(s_2+s_0,N_2) \Big(1-\frac{s_2 + s_0}{T}\Big)\bigg), \\
\Dc_2 &= \bigg(\min(s_1+s_0,N_1) \Big(1-\frac{s_1 + s_0}{T}\Big), \min(s_0,N_2)\frac{s_1}{T}\bigg)
\end{align}
are achievable. On top of that, if $s_1 \ge s_2$, the DoF pairs
\begin{align}
\mathcal{D}_3 & = \Big(\min(\rrone + \rronetwo,N_1)\Big(1 - \frac{\rrone + \rronetwo}{T}\Big), \min(s_2,N_2) \frac{s_1-s_2}{T} + \min(\rrtwo,(N_2-\rronetwo)^+)\Big(1 - \frac{\rrone + \rronetwo}{T}\Big)\Big),\\
\mathcal{D}_4 & = \Big(\min(\rrone,(N_1 - \rronetwo)^+) \Big(1 - \frac{\rrone + \rronetwo}{T}\Big),\min(s_2,N_2) \frac{s_1-s_2}{T} + \min(s_2 + s_0,N_2) \Big(1-\frac{s_1+s_0}{T}\Big)\Big), \\
\mathcal{D}_5 & = \Big(\min(\rrone + \rronetwo,N_1)\Big(1 - \frac{\rrone + \rronetwo}{T}\Big),\min(s_2+s_0,N_2)\frac{s_1-s_2}{T} + \min(s_2,N_2)\Big(1-\frac{s_1+s_0}{T}\Big)\Big)
\end{align}
are achievable; if $s_1 \le s_2$, the DoF pairs
\begin{align}
\mathcal{D}_3 & = \Big(\min(s_1,N_1) \frac{s_2-s_1}{T} + \min(\rrone,(N_1-\rronetwo)^+)\Big(1 - \frac{\rrtwo + \rronetwo}{T}\Big), \min(\rrtwo + \rronetwo,N_2)\Big(1 - \frac{\rrtwo + \rronetwo}{T}\Big)\Big),\\
\mathcal{D}_4 & = \Big(\min(s_1,N_1) \frac{s_2-s_1}{T} + \min(s_1 + s_0,N_1) \Big(1-\frac{s_2+s_0}{T}\Big), \min(\rrtwo,(N_2 - \rronetwo)^+) \Big(1 - \frac{\rrtwo + \rronetwo}{T}\Big)\Big), \\
\mathcal{D}_5 & = \Big(\min(s_1+s_0,N_1)\frac{s_2-s_1}{T} + \min(s_1,N_1)\Big(1-\frac{s_2+s_0}{T}\Big), \min(\rrtwo + \rronetwo,N_2)\Big(1 - \frac{\rrtwo + \rronetwo}{T}\Big)\Big)
\end{align}
are achievable. The convex hull of these DoF pairs (over all feasible values of $s_1,s_2$, and $s_0$) and the origin $(0,0)$ is achievable.
\end{theorem}

\begin{remark}
The parameters $\rronetwo,\rrone,\rrtwo$ represent the allocation of available dimensions to encoding of messages for the two users. By tuning these parameters, we explore the trade-off between the number of data dimensions (indicating the amount of channel uses needed for pilot transmission) and the amount of channel uses for data transmission within each section of the eigenspaces. 
\end{remark}

\begin{proof}[Proof of Theorem~\ref{thm:2user_CDIR_partiallyOverlapping}]
The DoF pairs $\Big(N_1^*\Big(1-\frac{N_1^*}{T}\Big),0\Big)$ and $\Big(0,N_2^*\Big(1-\frac{N_2^*}{T}\Big)\Big)$ are achieved by activating only one user according to Lemma~\ref{lem:single-user}. 

For any non-negative integers $\rronetwo,\rrone,\rrtwo$ satisfying $\rronetwo \leq r_0$, $\rrone \leq r_1 - r_0$ and $\rrtwo \leq r_2 - r_0$, there exist eigendirections $\tV_0 \in \CC^{M \times \rronetwo},\tV_1 \in \CC^{M \times \rrone}, \tV_2 \in \CC^{M \times \rrtwo}$, such that User~1 can only see signals in the direction of $\tX_1$ and $\tX_0$, while User~2 can only see signals in the direction of $\tX_2$ and $\tX_0$. (See Section~\ref{sec:intro_ratesplit}.)

\color{black}
To achieve $\mathcal{D}_1$, the base station employs product superposition and transmits
\begin{equation}
\tX = [\tV_0 \ \tV_1] \tX_2 \tX_1,
\end{equation}
with $\tX_1 = [\tI_{\rrone + \rronetwo}~\tS_1]$ and $\tX_2 = 
\Bigg[\begin{matrix}
\tI_{\rronetwo} & \tS_2\\
\ZeroMat & \tI_{\rrone} 
\end{matrix}\Bigg],$
where $\tS_1 \in \mathbb{C}^{(\rrone + \rronetwo) \times (T - \rrone-\rronetwo)}$ and $\tS_2 \in \mathbb{C}^{\rronetwo \times \rrone}$ contain symbols for User~$1$ and User~$2$, respectively. Following steps similar to the proof of Proposition~\ref{thm:2fwo}, it can be shown that this achieves the DoF pair $\mathcal{D}_1$. The DoF pair $\Dc_2$ can be achieved similarly by switching the users' role.

When $s_1 \ge s_2$, the pairs $\mathcal{D}_3$ and $\mathcal{D}_4$ are achieved with rate splitting as follows. Let the transmitter send
\begin{equation}
\tX = [\vvonetwo ~\vvone ~\vvtwo]
\begin{bmatrix}
\tI_{\rronetwo} & [\ZeroMat_{s_0 \times s_1} \quad \tS_0]\\
\ZeroMat_{s_{1} \times s_0} & [\tI_{\rrone}  \quad \tS_1]\\
\ZeroMat_{s_{2} \times s_0} & [\tI_{\rrtwo} \quad \tS_{2}]
\end{bmatrix},
\end{equation}
where $\tS_0 \in \CC^{s_0\times (T-s_1 - s_0)}$ is a common signal to both users while $\tS_1 \in \CC^{s_1 \times (T-s_1 - s_0)}$ and $\tS_2 \in \CC^{s_2 \times (T-s_2 - s_0)}$  are private signals to User~$1$ and User~$2$, respectively. 

The received signal at User~$1$ is
\begin{equation}
\begin{split}
\tY_1 
= {\tH}_1 [\vvonetwo ~\vvone]
\begin{bmatrix}
\tI_{\rronetwo} & \ZeroMat & \tS_0\\
\ZeroMat & \tI_{\rrone} & \tS_1
\end{bmatrix}
+ \tW_1.
\end{split}
\end{equation}
\color{black}
User~$1$ estimates the equivalent channel ${\tH}_1[\vvonetwo ~\vvone]$ during the first $\rrone + \rronetwo$ channel uses and decodes both $\tS_{1}$ and $\tS_0$ during the remaining $T - s_1 -s_0$ channel uses, achieving $\min(\rrone + \rronetwo, N_1)\frac{T - \rrone - \rronetwo}{T}$ DoF.
The received signal at User~$2$ is
\begin{equation}
\begin{split}
\tY_2 
= {\tH}_2 [\vvonetwo ~ \vvtwo]
\begin{bmatrix}
\tI_{\rronetwo} & \ZeroMat & [ \ZeroMat_{s_0 \times (s_1 - s_2)} \ \tS_0]\\
\ZeroMat & \tI_{\rrtwo} & \tS_2
\end{bmatrix}
+ \tW_2.
\end{split}
\end{equation}
\color{black}
User~$2$ estimates the equivalent channel ${\tH}_2[\vvonetwo\ \vvtwo]$ and then decodes $\tS_0$ and $\tS_2$, achieving $\min(s_2,N_2) \frac{s_1-s_2}{T} + \min(s_2+s_0,N_2) \frac{T-s_1-s_0}{T}$ DoF. By dedicating $\tS_0$ to only User~$1$ or User~$2$, DoF pairs $\mathcal{D}_3$ and $\mathcal{D}_4$ are achieved, respectively. 

The degrees of freedom pair $\mathcal{D}_5$ can be achieved (still assuming $s_1 \ge s_2$), via a combination of rate splitting and product superposition as follows. The transmitted signal is 
\begin{equation}
\tX = [\tV_0 \ \tV_1] \tX'_2\tX_1 + \tV_2 \tX_2,
\end{equation}
with $\tX_2 = [\ZeroMat_{\rronetwo \times \rronetwo} ~ \tI_{\rrtwo} ~ \tS_2]$, $\tX_1 = [\tI_{\rrone + \rronetwo} ~ \tS_1]$, and $\tX'_2 = 
\begin{bmatrix}
\tI_{\rronetwo} & [\ZeroMat_{\rronetwo \times \rrtwo}~\tS'_2]\\
\ZeroMat_{s_1 \times s_0} & \tI_{s_1}
\end{bmatrix}$, 
where $\tS_1 \in \CC^{(s_1+ s_0) \times (T-s_1-s_0)}$ contains symbols intended for User~$1$ while $\tS \in \CC^{s_2 \times (T-s_2-s_0)}$ and $\tS'_2 \in \mathbb{C}^{\rronetwo \times (\rrone - \rrtwo)}$ contain symbols intended for User~$2$. 
The received signal at User~$1$ is
\begin{equation}
\begin{split}
\tY_1 
= {\tH}_1 [\vvonetwo ~ \vvone] \tX'_2[\tI_{\rrone + \rronetwo} ~ \tS_1] + \tW_1.
\end{split}
\end{equation}
\color{black}
User~$1$ estimates the equivalent channel ${\tH}_1 [\vvonetwo ~ \vvone] \tX'_2$, and then decodes $\tS_1$ to achieve $\min(\rronetwo + \rrone,N_1)\frac{T - \rronetwo - \rrone}{T}$ DoF. 
The received signal at User~$2$ is
\begin{equation}
\begin{split}
\tY_2 = {\tH}_2 [\vvonetwo ~ \vvtwo]
\begin{bmatrix}
\tI_{\rronetwo} & \ZeroMat_{\rrtwo \times \rrtwo} & [\tS'_2 ~ \mathbf{A}]\\
\ZeroMat_{\rronetwo \times \rronetwo} & \tI_{\rrtwo} & \tS_2
\end{bmatrix}
+ \tW_2,
\end{split}
\end{equation}
\color{black}
where ${\mathbf A} \triangleq [\tI_{\rronetwo}~\ZeroMat_{\rronetwo \times \rrtwo}~\tS'_2]\tS_1$. User~$2$ estimates its equivalent channel ${\tH}_2 [\vvonetwo ~ \vvtwo]$ in the first $\rrtwo+\rronetwo$ channel uses, and then decodes $\tS'_2$ and $\tS_2$, achieving $\min(s_2+s_0,N_2)\frac{s_1-s_2}{T} + \min(s_2,N_2)\frac{T-s_1-s_0}{T}$ DoF in total. Therefore, $\mathcal{D}_5$ is achieved.

Therefore, the proof for the case where $\rrone \geq \rrtwo$ is completed. A similar analysis applies to the case $\rrtwo \geq \rrone$ and completes the proof of Theorem~\ref{thm:2user_CDIR_partiallyOverlapping}. 
\end{proof}

In Figure~\ref{fig:dof_2pwo}, the achievable DoF region in Theorem~\ref{thm:2user_CDIR_partiallyOverlapping} is shown for the scenario where $T = 24$, $N_1 = 12$, $N_2 = 12$, $(r_1,r_2) \in \{(12,10),(12,12)\}$, and $r_0 \in \{0,3,6,9\}$. Similar to the CSIR case, exploiting the channel correlation improves significantly the DoF region upon TDMA, especially for small $r_0$. Note that TDMA was shown to be degrees of freedom optimal when the channel is uncorrelated \cite{Fadel_disparity}.
\begin{figure}
\centering
\hspace{-.3cm}
\subfigure[$r_1 = 12$, $r_2 = 10$, $r_0 \in \{0,3,6,9\}$]{
\includegraphics[width=.48\textwidth]{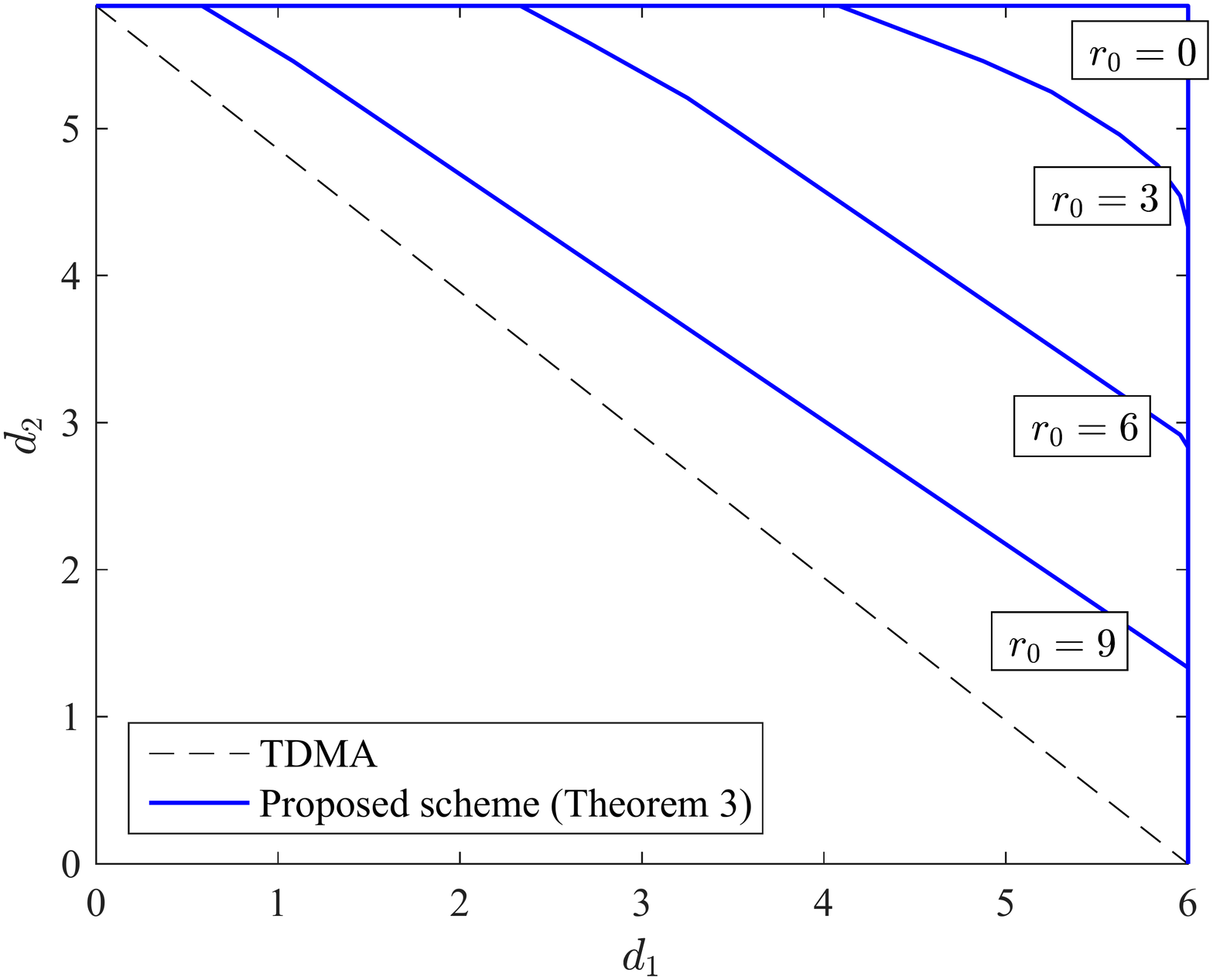} 
}
\subfigure[$r_1 = r_2 = 12$, $r_0 \in \{0,3,6,9\}$]{
\includegraphics[width=.49\textwidth]{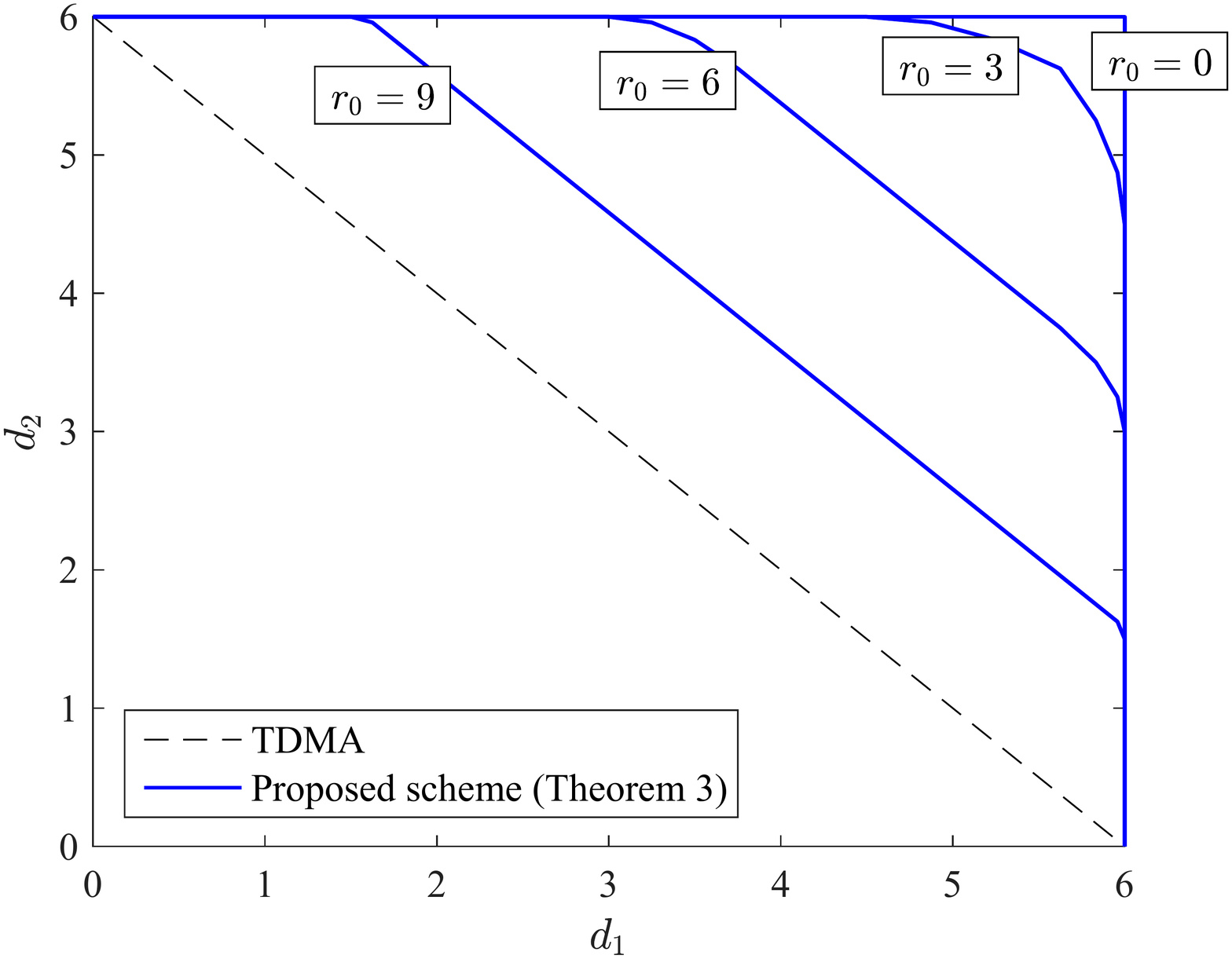}  
}
\caption{The DoF region for the two-user broadcast channel without free CSIR, achieved with TDMA or the proposed scheme (Theorem~\ref{thm:2user_CDIR_partiallyOverlapping}) for $T = 24$, $N_1 = 12$, $N_2 = 12$, $(r_1,r_2) \in \{(12,10),(12,12)\}$, and $r_0 \in \{0,3,6,9\}$.}
\label{fig:dof_2pwo}
\end{figure}

This completes the DoF analysis for the two-user case. By using both product superposition and rate splitting,  achievable DoF regions were calculated for a variety of correlation structures and antenna configurations. Also, an outer bound was calculated under perfect CSIR.

\section{Two-User Broadcast Channel: Rate Analysis}
\label{sec:2user_rate}
We assume no free CSIR under partially overlapping eigenspaces, and assume that $r_k \le N_k$, $k \in \{1,2\}$. In addition, without loss of generality $r_1 \ge r_2$.

\subsection{The Single-User Case} \label{sec:single-user}
Let us first consider the single-user case where, for simplicity, we omit the user's index. The received signal is
\begin{align}
\rvMat{Y} = \rvMat{H} \rvMat{X} + \rvMat{W},
\end{align}
where the assumptions for the transmitted signal $\rvMat{X}$, the Gaussian noise $\rvMat{W}$, and the channel  $\rvMat{H}$ are as before. In particular, $\rvMat{H}$ is block fading with coherence time $T$, and has correlation matrix $\Rm = \Um \Sigmam \Um^\H$, thus can be written as $\rvMat{H} = \rvMat{G} \Sigmam^{\frac12} \Um^\H$ with $\rvMat{G} \in \CC^{N \times r}$ drawn from a generic distribution. 
The following theorem states the achievable rate (in bits/channel use) for this channel.
\begin{theorem} \label{thm:single_user}
Achievable rates for a single-user spatially-correlated MIMO channel without free CSIR are as follows.
\begin{enumerate}
\item if the transmitter does not know the channel correlation matrix $\Rm$,
\begin{align}
R = \Big(1-\frac{M}{T}\Big) \EE \bigg[ \log \det \bigg(\Id_N + \frac{\rho_\delta \rho_\tau}{\rho_\delta \trace \big((\Sigmam^{-1} + \rho_\tau\Id_r )^{-1}\big) + M}  \hat{\rvMat{H}} \hat{\rvMat{H}}^\H \bigg) \bigg],
\label{eq:Rate_su_noRatTX}
\end{align}
where rows of $\hat{\rvMat{H}}$ are i.i.d. according to ${\Cc \Nc} \big({\bf0}^{T}, \Rm(\Id_M + \rho_\tau\Rm)^{-1}\Rm \big)$, and $\rho_\tau M + \rho_\delta(T-M) \le \rho T$;

\item if the transmitter knows the channel correlation matrix $\Rm$, under orthogonal pilots:

\begin{align}
\label{eq:Rate_su_RatTX_orthPilot}
R = \Big(1-\frac{r}{T}\Big) \EE \bigg[ \log \det \bigg(\Id_N + \frac{\rho_\delta \rho_\tau}{\rho_\delta\trace \big((\bar{\Rm}^{-1} + \rho_\tau \Id_{r}\big)^{-1}\big) + r}  \eqhat\eqhat^\H \bigg) \bigg],
\end{align}
where rows of $\eqhat$ are i.i.d.\ according to ${\Cc \Nc} \big({\bf 0}^T, \bar{\Rm}(\Id_r + \rho_\tau\bar{\Rm})^{-1}\bar{\Rm} \big)$ such that $\bar{\Rm} \defeq \Vm^\H \Rm \Vm$ for a truncated unitary matrix $\Vm \in \CC^{M \times r}$ such that $\Span[\Vm] = \Span[\Um]$. 

Allowing non-orthogonal pilots can improve the rate to: 

  \begin{align}
  \label{eq:Rate_su_RatTX_optPilot}
R = \Big(1-\frac{r}{T}\Big) \EE \bigg[ \log \det \bigg(\Id_N + \frac{\rho_\delta }{r\rho_\delta\big(\rho_\tau+\frac{1}{r}\trace (\bar{\Rm}^{-1}) \big)^{-1} + r}  \eqhat\eqhat^\H \bigg) \bigg],
\end{align}
where rows of $\eqhat$ are i.i.d. according to ${\Cc \Nc} \big({\bf 0}^T, \bar{\Rm} - \bigg(\rho_\tau + \frac{1}{r}\trace[\bar{\Rm}^{-1}] \bigg)^{-1} \Id_{r} \big)$.
\end{enumerate}
\end{theorem}
\begin{proof}
See Appendix~\ref{proof:rate_single}.
\end{proof}

\begin{remark}
\label{remark:powerallocation}
The optimal power allocation for the rate in \eqref{eq:Rate_su_RatTX_optPilot} is given by $\rho_\tau= \frac{(1-\alpha) \rho T}{r}$ and $\rho_\delta = \frac{\alpha \rho T}{T-r}$ with
\begin{align} \label{eq:p2p_power_allocation}
\alpha = \begin{cases}
\frac{1}{2}, & \text{if~~} T = 2r, \\
b - \sqrt{b(b-a)}, & \text{if~~} T > 2r,
\end{cases}
\end{align}
where $a \defeq 1 + \frac{\trace (\bar{\Rm}^{-1})}{\rho T} - \frac{r^2}{\rho T \trace (\bar{\Rm})}$ and $b \defeq \frac{T-r}{T-2r}\Big(1 + \frac{\trace (\bar{\Rm}^{-1})}{\rho T}\Big).$
\end{remark}

\begin{corollary}
If the channel is uncorrelated, i.e., $\Rm = \Id_M$, the achievable rate is
\begin{align}
R 
&= \Big(1-\frac{M}{T}\Big) \EE \bigg[\log\det\bigg(\Id_{N} + \frac{\rho_\delta \rho_\tau}{M(1 +\rho_\delta +\rho_\tau)} {\rvMat{H}} {\rvMat{H}}^\H \bigg)\bigg],
\end{align}
where ${\rvMat{H}} \in \CC^{N\times M}$ is the uncorrelated channel matrix.
This coincides with~\cite[Eq.(21)]{Hassibi_2003}.
\end{corollary}

\subsection{The Baseline TDMA Schemes} \label{sec:2user_rate_TDMA}
We consider TDMA without free CSIR. If only User~$k$ is activated and the base station does not exploit $\Rm_k$, according to Theorem~\ref{thm:single_user}, the following corollary demonstrates the achievable rate:
\begin{corollary}
For 2-user broadcast channel, when the transmitter does not know channel correlations $\{\Rm_1,\Rm_2\}$, the following single-user rates are achievable for users $k=1,2$:
\begin{align}
\label{eq:achieavable-rate-TDMA_fullspace}
R_k = \Big(1-\frac{M}{T}\Big) \EE \bigg[ \log \det \bigg(\Id_{N_k} + \frac{\rho_\delta \rho_\tau}{\rho_\delta \trace \big((\Sigmam_k^{-1} + \rho_\tau\Id_{r_k} )^{-1}\big) + M}  \eqhat_k\eqhat_k^\H \bigg) \bigg], \quad\quad k=1,2,
\end{align}
where rows of $\eqhat_k$  are i.i.d. according to ${\Cc \Nc} \big({\bf 0}^T, \Rm_k(\Id_M + \rho_\tau\Rm_k)^{-1}\Rm_k \big)$, and $\rho_\tau M + \rho_\delta(T-M) \le \rho T$;

\end{corollary}

If the base station transmits in the eigenspace of $\Rm_k$ using precoder $\Vm_k = \Um_k$, i.e., $\Um_k^\H\Vm_k = \Id_{r_k}$, and optimizes the pilot, the following corollary demonstrates the achievable rate:
\begin{corollary}
For 2-user broadcast channel, when the transmitter emits in the eigenspace of $\Rm_1$, the following single-user rate is achievable:

 \begin{align}  \label{eq:achieavable-rate-TDMA}
R_1 = \Big(1-\frac{r_1}{T}\Big) \EE \bigg[ \log \det \bigg(\Id_{N_1} + \frac{\rho_\delta }{r\rho_\delta\big(\rho_\tau+\frac{1}{r}\trace (\Sigmam_1^{-1}) \big)^{-1} + r_1}  \eqhat_1\eqhat_1^\H \bigg) \bigg],
\end{align}
where rows of $\eqhat_1$ are i.i.d.\ according to  ${\Cc \Nc} \big({\bf 0}^T, \Sigmam_1 - \big(\rho_\tau + \frac{1}{r_1}\trace (\Sigmam_1^{-1}) \big)^{-1} \Id_{r_1} \big)$, and $\rho_\tau r + \rho_\delta(T-r) \le \rho T$. A corresponding (single-user) rate applies for $R_2$. 
\end{corollary}
The optimal power allocation for \eqref{eq:achieavable-rate-TDMA} closely follows Remark~\ref{remark:powerallocation} and is ommited for brevity. The convex hull of $(0,0)$, $(R_1,0)$, and $(0,R_2)$ is achievable by TDMA.

\subsection{Rate Splitting}
\label{sec:2user_rate_split}
In the following, we analyze the rate achievable with the schemes achieving the DoF region in Theorem~\ref{thm:2user_CDIR_partiallyOverlapping}. Recall that for a set of non-negative integers $s_0 \le r_0$, $s_1 \le r_1 - r_0$, and $s_2 \le r_2 - r_0$, the precoding matrices $\Vm_0,\Vm_1, \Vm_2$, are defined in Section~\ref{sec:intro_ratesplit},
define
\begin{itemize}
\item $\Phim_k \defeq \Um_k^\H [\Vm_0 \ \Vm_k]$, $\Phim_{k0} \defeq \Um_k^\H \Vm_0$, $\Phim_{kk} \defeq \Um_k^\H \Vm_k$ (so $\Phim_k = [\Phim_{k0} \ \Phim_{kk}]$);
\item $\bar{\Rm}_k \defeq \Phim_k^\H \Sigmam_k \Phim_k$, $\bar{\Rm}_{k0} \defeq \Phim_{k}^\H \Sigmam_k \Phim_{k0}$, $\bar{\Rm}_{kk} \defeq \Phim_{k}^\H \Sigmam_k \Phim_{kk}$ (so $\bar{\Rm}_k = [\bar{\Rm}_{k0} \ \bar{\Rm}_{kk}]$);
\item $\breve{\Rm}_{k0} \defeq \Phim_{k0}^\H \Sigmam_k \Phim_{k0}$, $\breve{\Rm}_{kk} \defeq \Phim_{kk}^\H \Sigmam_k \Phim_{kk}$.
\end{itemize}

Let the base station transmit
\begin{align} \label{eq:rate-splitting}
\rvMat{X} = \Vm_0 \rvMat{X}_0 + \Vm_1 \rvMat{X}_1 + \Vm_2 \rvMat{X}_2,
\end{align}
where $\rvMat{X}_0$, $\rvMat{X}_1$, and $\rvMat{X}_2$ are independent and satisfy the power constraint $\E[\|\rvMat{X}_0\|_F^2 + \|\rvMat{X}_1\|_F^2 + \|\rvMat{X}_2\|_F^2] \le \rho T$. Thanks to the precoders, the private signal $\rvMat{X}_k$ is seen by User~$k$ only, while the common signal $\rvMat{X}_0$ is seen by both users. The received signals become
\begin{align} 
\hspace{-.2cm} \rvMat{Y}_1 &= \rvMat{G}_1 \Sigmam^{\frac12}_1 \Phim_{10} \rvMat{X}_0 + \rvMat{G}_1 \Sigmam^{\frac12}_1 \Phim_{11} \rvMat{X}_1 + \rvMat{W}_1, \label{eq:2MACequivalent1}\\
\hspace{-.2cm} \rvMat{Y}_2 &= \rvMat{G}_2 \Sigmam^{\frac12}_2 \Phim_{20} \rvMat{X}_0 +\rvMat{G}_2 \Sigmam^{\frac12}_2 \Phim_{22} \rvMat{X}_2 +  \rvMat{W}_2, \label{eq:2MACequivalent2}
\end{align} 
where the equivalent channels $\rvMat{G}_k \Sigmam^{\frac12}_k \Phim_{k0} \in \CC^{N_k \times s_0}$ and $\rvMat{G}_k \Sigmam^{\frac12}_k \Phim_{kk} \in \CC^{N_k \times s_k}$, $k\in\{1,2\},$ are correlated and unknown. It can be observed that the received signal at each user is similar to a non-coherent two-user MAC: \eqref{eq:2MACequivalent1} as the MAC $1$ with ($s_0,s_1$) equivalent transmit antennas and $N_1$ receive antennas, \eqref{eq:2MACequivalent2} as the MAC $2$ with ($s_0,s_2$) equivalent transmit antennas and $N_2$ receive antennas. The two MACs share a common signal $\rvMat{X}_0$. 

From the capacity region of multiple access channels~\cite{ElGamal}, we know that the rate pairs $(R_0,R_1^p)$ and $(R_0,R_2^p)$ are simultaneously achievable for the MAC $1$ and MAC $2$, respectively, if the rates $R_0 \ge 0,R_1^p \ge 0,R_2^p\ge 0$ satisfy 
\begin{align} \label{eq:rateMAC1}
R_0 &\le \frac{1}{T}I(\rvMat{Y}_1; \rvMat{X}_0 | \rvMat{X}_1), \\
R_1^p &\le \frac{1}{T}I(\rvMat{Y}_1; \rvMat{X}_1 | \rvMat{X}_0), \\
R_0+R_1^p &\le \frac{1}{T}I(\rvMat{Y}_1; \rvMat{X}_0,\rvMat{X}_1), \\
R_0 &\le \frac{1}{T}I(\rvMat{Y}_2; \rvMat{X}_0 | \rvMat{X}_2), \\
R_2^p &\le \frac{1}{T}I(\rvMat{Y}_2; \rvMat{X}_2 | \rvMat{X}_0), \\
R_0+R_1^p &\le \frac{1}{T}I(\rvMat{Y}_2 ; \rvMat{X}_0,\rvMat{X}_2). \label{eq:rateMAC2}
\end{align}
Then, User~$1$ achieves rate $R_1^p$ with private signal $\rvMat{X}_1$, user 2 achieves rate $R_2^p$ with private signal $\rvMat{X}_2$, and both users can achieve rate $R_0$ with common signal $\rvMat{X}_0$. Let $R_{0k}$ be the User~$k$'s share in $R_0$, then the rate pair $(R_1,R_2) = (R_{01}+R_1^p, R_{02}+R_2^p)$ is achievable. Replacing $R_0=R_{01}+R_{02}$, $R_1^p = R_1 - R_{01}$, and $R_2^p = R_2 - R_{02}$ in \eqref{eq:rateMAC1}-\eqref{eq:rateMAC2} and applying Fourier-Motzkin elimination leads to the following result.
\begin{lemma} 
\label{lemma:BC_rate_region}
With rate splitting and without free CSIR, rate pairs $(R_1,R_2)$ are achievable with:
\begin{align}
R_1 &\le \frac{1}{T}\min\{I(\rvMat{Y}_1;\rvMat{X}_1,\rvMat{X}_0), I(\rvMat{Y}_1;\rvMat{X}_1|\rvMat{X}_0)+I(\rvMat{Y}_2;\rvMat{X}_0|\rvMat{X}_2)\}, \label{eq:add:R1_bound}\\
R_2 &\le \frac{1}{T}\min\{I(\rvMat{Y}_2;\rvMat{X}_2,\rvMat{X}_0),I(\rvMat{Y}_2;\rvMat{X}_2|\rvMat{X}_0)+I(\rvMat{Y}_1;\rvMat{X}_0|\rvMat{X}_1)\}, \label{eq:add:R2_bound}\\
R_1+R_2 &\le \frac{1}{T}\min\{I(\rvMat{Y}_1;\rvMat{X}_1|\rvMat{X}_0)+I(\rvMat{Y}_2;\rvMat{X}_2,\rvMat{X}_0), I(\rvMat{Y}_1;\rvMat{X}_1,\rvMat{X}_0)+I(\rvMat{Y}_2;\rvMat{X}_2|\rvMat{X}_0)\},\label{eq:add:sum_rate_bound}
\end{align}
for input distributions $p(\rvMat{X}_0)$, $p(\rvMat{X}_1)$, and $p(\rvMat{X}_2)$ satisfying $\E[\|\rvMat{X}_0\|_F^2 + \|\rvMat{X}_1\|_F^2 + \|\rvMat{X}_2\|_F^2] \le \rho T$.
\end{lemma}

By bounding the mutual information terms in Lemma~\ref{lemma:BC_rate_region}, we have the following theorem:

\begin{theorem}
\label{thm:BC_rate_split}
Under rate splitting, the following rate region can be achieved in the two-user correlated broadcast channel with partially overlapped eigenspaces:
\begin{align}
R_1 & \leq \min \{R_1^{\prime}, R_1^p+ R_0^{\prime \prime} \} , \label{eq:RateSplitting1}\\
R_2 & \leq \min \{R_2^{\prime}, R_2^p+ R_0^{\prime}\}, \label{eq:RateSplitting2}\\
R_1 + R_2 & \leq \min \{R_1^p + R_2^{\prime},  R_1^{\prime} + R_2^p\} \label{eq:RateSplitting3},
\end{align}
where

\begin{align}
  R_1^{\prime} = 
  \Big(1 - \frac{s_1 + s_0}{T}\Big) \EE\bigg[\log \det \bigg( \Id_{N_1} +\frac{1}{\trace \big(\big(\bar{\Rm}_1^{-1} + \Pm_{1\tau} \big)^{-1} \Pm_{1\delta}\big) + 1}  \bar{{\bf \Omega}}_1 \bar{\Rm}_1 \Pm_{1\delta} \bar{\Rm}_1^\H \bar{{\bf \Omega}}_1^{\H} \bigg)\bigg],
\end{align}

\begin{align}
  R_1^{p} = 
  \Big(1 - \frac{s_1 + s_0}{T}\Big) \EE \bigg[\log \det \bigg( \Id_{N_1} + \frac{\rho_{1\delta}}{s_1 \big[\trace \big((\bar{\Rm}_1^{-1} + \Pm_{1\tau} )^{-1} \Pm_{1\delta}\big) + 1\big]} \bar{{\bf \Omega}}_1 \bar{\Rm}_{11} \bar{\Rm}_{11}^\H\bar{{\bf \Omega}}_1^{\H}\bigg)\bigg],
\end{align}

\begin{align}
  R_1^{\prime \prime} = 
  \Big(1 - \frac{s_1 + s_0}{T}\Big) \EE \bigg[\log \det \bigg( \Id_{N_1} + \frac{\rho_{0\delta}}{s_0\big[\trace \big( (\bar{\Rm}_1^{-1} + \Pm_{1\tau} \big)^{-1} \Pm_{1\delta}\big) + 1\big]} \bar{{\bf \Omega}}_1 \bar{\Rm}_{10} \bar{\Rm}_{10}^\H \bar{{\bf \Omega}}_1^{\H}\bigg) \bigg],
\end{align}
where rows of $\bar{{\bf \Omega}}_1$  obey ${\Cc \Nc} \big({\bf 0}^T, \Pm_{1\tau}^{\frac12}(\Pm_{1\tau}^{\frac12}\bar{\Rm}_1\Pm_{1\tau}^{\frac12} + \Id_{s_1+s_0})^{-1}\Pm_{1\tau}^{\frac12}\big)$ and are independent of each other.

\begin{align}
   R_2^{\prime} = 
  & \frac{s_1 - s_2}{T} \EE \bigg[
\log \det \bigg(\Id_{N_2} + \frac{\rho_{2 \delta }}{\rho_{2 \delta}\trace \big(\bar{\Rm}_{22}^\H (\bar{\Rm}_2 +  \bar{\Rm}_2 \Pm_{2\tau} \bar{\Rm}_2)^{-1}\bar{\Rm}_{22} \big) +s_2 }\bar{{\bf \Omega}}_2 \bar{\Rm}_{22} \bar{\Rm}_{22}^\H \bar{{\bf \Omega}}_2^{\H} \bigg)\bigg]  \notag\\
& + \Big(1 - \frac{s_1 + s_0}{T}\Big) \EE \bigg[\log \det \bigg(\Id_{N_2} + \frac{1}{\trace \big((\bar{\Rm}_2^{-1} + \Pm_{2\tau})^{-1} \Pm_{2\delta}\big)+ 1} \bar{{\bf \Omega}}_2 \bar{\Rm}_{2} \Pm_{2\delta} \bar{\Rm}_{2}^\H \bar{{\bf \Omega}}_2^{\H}\bigg) \bigg],
\end{align}

\begin{align}
  R_2^{p} = 
   & \frac{s_1 - s_2}{T} \EE \bigg[ \log \det \bigg(\Id_{N_2} + \frac{\rho_{2 \delta }}{\rho_{2 \delta}\trace \big(\bar{\Rm}_{22}^\H (\bar{\Rm}_2 + \bar{\Rm}_2 \Pm_{2\tau} \bar{\Rm}_2)^{-1}\bar{\Rm}_{22}\big)+s_2} \bar{{\bf \Omega}}_2 \bar{\Rm}_{22} \bar{\Rm}_{2} \bar{\Rm}_{2}^\H\bar{\Rm}_{22}^\H \bar{{\bf \Omega}}_2^{\H}\bigg)\bigg] \notag\\
&+ \Big(1 - \frac{s_1 + s_0}{T}\Big) \EE \bigg[\log \det \bigg(\Id_{N_2} + \frac{\rho_{2 \delta }}{{s_2\big[ \trace \big( (\bar{\Rm}_2^{-1} + \Pm_{2\tau})^{-1}\Pm_{2\delta} \big) + 1\big]} } \bar{{\bf \Omega}}_2 \bar{\Rm}_{22} \bar{\Rm}_{2} \bar{\Rm}_{2}^\H \bar{\Rm}_{22}^\H \bar{{\bf \Omega}}_2^{\H} \bigg)\bigg],
\end{align}

\begin{align}
  R_0^{\prime \prime} = 
  \Big(1 - \frac{s_1 + s_0}{T}\Big) \EE \bigg[\log \det \bigg(\Id_{N_2} + \frac{\rho_{0 \delta }}{s_0\big[\trace \big( (\bar{\Rm}_2^{-1} +  \Pm_{2\tau} )^{-1} \Pm_{2\delta}\big) + 1\big]}  \bar{{\bf \Omega}}_2\bar{\Rm}_{20} \bar{\Rm}_{2} \bar{\Rm}_{2}^\H \bar{\Rm}_{20}^\H\bar{{\bf \Omega}}_2^{\H}  \bigg) \bigg],
\end{align}
where rows of $\bar{{\bf \Omega}}_2$  are i.i.d.\ according to ${\Cc \Nc} \big({\bf 0}^T,  \Pm_{2\tau}^{\frac12}(\Pm_{2\tau}^{\frac12}\bar{\Rm}_2\Pm_{2\tau}^{\frac12} + \Id_{s_2+s_0})^{-1}\Pm_{2\tau}^{\frac12}\big)$. Variables $s_0,s_1,s_2$ allocate degrees of freedom and satisfy $s_0 \leq r_0$, $s_1 \leq r_1 - r_0$ and $s_2 \leq r_2 - r_0$. The component powers $\rho_{\tau 0},\rho_{\tau 1},\rho_{\tau 2},\rho_{\delta}$ satisfy the power constraint 
\begin{align}
\rho_{0\tau}s_0 + \rho_{0 \delta } (T-s_1-s_0) + \sum_{i=1}^{2} \big[\rho_{i\tau} s_i + \rho_{i\delta} (T-s_i-s_0)\big] \le \rho T.
\label{eq:PowerAllocRateSplitting}
\end{align} 
The overall achievable rate region is the convex hull of \eqref{eq:RateSplitting1}, \eqref{eq:RateSplitting2} and \eqref{eq:RateSplitting3}  over all feasible values of $s_0, s_1,s_2$ and power allocations~\eqref{eq:PowerAllocRateSplitting}.
\end{theorem}
\begin{proof}
See Appendix~\ref{proof:rate_part}. 
\end{proof}

\subsection{Product Superposition} \label{sec:2user_rate_prod}

\begin{theorem}
\label{thm:BC_rate_product}
With product superposition, the following rate pair $(R_1,R_2)$ can be achieved:
\begin{align}
R_1 & = \frac{s_2}{T} \EE \bigg[ \log \det  \bigg(\Id_{N_1} + \frac{\nu_{1\delta} \rho_{2\tau}}{s_0 + \nu_{1\delta} \rho_{2\tau}\trace \big( (\breve{\Rm}_{10}^{-1} + \nu_{1\tau}\rho_{2\tau} \Id_{s_0} )^{-1}\big)} \eqhatonezero \eqhatonezero^\H \bigg)\bigg],
\end{align}
where rows of $\eqhatonezero$  are i.i.d.\ according to ${\Cc \Nc} \big({\bf 0}^T,   \nu_{1\tau}\rho_{2\tau} \breve{\Rm}_{10} (\nu_{1\tau}\rho_{2\tau} \breve{\Rm}_{10} + \Id_{s_0})^{-1} \breve{\Rm}_{10} \big)$;

\begin{align}
R_2 & = \Big(1-\frac{s_2+s_0}{T}\Big) \EE \bigg[ \log\det \bigg( {\Id_{N_2} + \frac{\rho_{2\delta}}{s_2 + s_0 + \rho_{2\delta}\trace \big( (\Rm_{2e}^{-1} + \rho_{2\tau} \Id_{s_2 + s_0} )^{-1}\big)} \hat{\rvMat{G}}_{2e} \hat{\rvMat{G}}_{2e}^\H } \bigg) \bigg],
\end{align}
where rows of $\hat{\rvMat{G}}_{2e}$  are i.i.d.,\ 
{zero mean, with covariance $\rho_{2\tau} \Rm_{2e} \big(\rho_{2\tau}\Rm_{2e} + \Id_{s_2+s_0}\big)^{-1} \Rm_{2e}$, where}
\begin{align}
\Rm_{2e} \defeq \begin{bmatrix}
\nu_{1\tau} \breve{\Rm}_{20} & \sqrt{\nu_{1\tau} \nu_{1a}} \Phim_{20}^\H \Sigmam_2 \Phim_{22} \\
\sqrt{\nu_{1\tau} \nu_{1a}} \Phim_{22}^\H \Sigmam_2 \Phim_{20} & \frac{\nu_{1\delta}}{s_0} \trace[\breve{\Rm}_{20}] \Id_{s_2} + \nu_{1a} \breve{\Rm}_{22}
\end{bmatrix};
\end{align}
 $s_0,s_1,s_2$ allocate degrees of freedom to signal components, and satisfy $s_0 \leq r_0$ and $s_2 \leq r_2 - r_0$ with the power constraint
\begin{align}
\big(s_0 \nu_{1\tau} + s_2(\nu_{1\delta}+\nu_{1a})\big)\Big(\rho_{2\tau} + \frac{T-s_2-s_0}{s_2+s_0}\rho_{2\delta} \Big) \le \rho T.
\label{eq:PowerAllocProdSuperposition}
\end{align}
By swapping the users' role, another achievable rate pair is obtained. The overall achievable rate region is the convex hull of these pairs over all feasible values of $s_0, s_1,s_2$ and feasible power allocations~\eqref{eq:PowerAllocProdSuperposition}.
\end{theorem}

\begin{remark}
The distribution of $\hat{\rvMat{G}}_{2e}$ is non-Gaussian. As clarified in~\eqref{eq:prod:MMSE_estimator}, it consists of a Gaussian matrix plus the product of two other Gaussian matrices.
\end{remark}
\color{black}

\begin{proof}
See Appendix~\ref{proof:rate_product}. 
\end{proof}

\subsection{Hybrid Superposition} \label{sec:2user_rate_hybrid}

Hybrid superposition in this paper refers to a composite scheme that involves both rate splitting and product superposition. 

\begin{theorem}
\label{thm:BC_rate_hybrid}
With hybrid superposition, the following rate pair $(R_1,R_2)$ can be achieved:
\begin{align}
R_1 = \Big(1-\frac{s_1+s_0}{T}\Big) \EE \bigg[\log\det \bigg(\Id_{N_1} + \frac{\rho_{1\delta}}{s_1 + s_0 + \rho_{1\delta}\trace \big((\Rm_{1e}^{-1} + \rho_{1\tau} \Id_{s_1+s_0} )^{-1}\big)} \hat{\rvMat{G}}_{1e} \hat{\rvMat{G}}_{1e}^\H \bigg)\bigg],
\end{align}
where rows of $\hat{\rvMat{G}}_{1e}$  are i.i.d., 
{zero mean, with covariance $ \rho_{1\tau} \Rm_{1e} \big(\rho_{1\tau}\Rm_{1e} + \Id_{s_1+s_0}\big)^{-1} \Rm_{1e}$}, where
\begin{align}
\Rm_{1e} &\defeq \begin{bmatrix}
\nu_{2\tau} \breve{\Rm}_{10} & \sqrt{\nu_{2\tau}\nu_{2a}} \Phim_{10}^\H \Sigmam_1 \Phim_{11}\\
\sqrt{\nu_{2\tau}\nu_{2a}} \Phim_{11}^\H \Sigmam_1 \Phim_{10} & \Bigg[\begin{matrix}
\mathbf{0} & \mathbf{0} \\
\mathbf{0} & \frac{\nu_{2\delta}}{s_0} \trace[\breve{\Rm}_{10}] \Id_{s_1-s_2}
\end{matrix}\Bigg]
+ \nu_{2a} \breve{\Rm}_{22}
\end{bmatrix},
\end{align}
and


\begin{align}
R_2 & = \frac{s_1 - s_2}{T} \EE \bigg[ \log \det \bigg(\Id_{N_2} + \frac{1}{\trace[\left(\bar{\Rm}_2^{-1} + \Pm_{2\tau} \right)^{-1} \Pm_{2\delta a}] + 1} \bar{\bf \Omega}_2 \Pm_{2\delta a} \bar{\bf \Omega}_2^\H\bigg) \bigg] \notag\\
& + \bigg(1 - \frac{s_1 + s_0}{T}\bigg) \EE \bigg[ \log \det \bigg(\Id_{N_2} + \frac{1}{\trace[\big(\bar{\Rm}_2^{-1} + \Pm_{2\tau} \big)^{-1} \Pm_{2\delta b}] + 1} \bar{\bf \Omega}_2 \Pm_{2\delta b} \bar{\bf \Omega}_2^\H \bigg) \bigg] \notag\\
& - \bigg(1 - \frac{s_1 + s_0}{T}\bigg) \EE \bigg[ \log \det \bigg(\Id_{N_2} + \rho_{1\delta}\Big(\nu_{2\tau} + \nu_{2\delta}\frac{s_1-s_2}{s_0}\Big) \bar{\bf \Omega}_{20} \bar{\bf \Omega}_{20}^\H\bigg) \bigg],
\end{align}
where rows of $\bar{{\bf \Omega}}_2$ are i.i.d.\ according to ${\Cc \Nc} \big({\bf 0}^T,   \bar{\Rm}_2^\H \big(\bar{\Rm}_2 + \Pm_{2\tau}^{-1}\big)^{-1} \bar{\Rm}_2 \big)$ and rows of $\bar{\bf \Omega}_{20}$ are i.i.d.\ according to ${\Cc \Nc} \big({\bf 0}^T, \breve{\Rm}_{20} \big)$, and they are independent of each other. Variables $s_0,s_1,s_2$ allocate degrees of freedom to signal components, and satisfy $s_0 \leq r_0$, $s_1 \leq r_1 - r_0$ and $s_2 \leq r_2 - r_0$ with the power constraint
\begin{align}
\big(s_0\nu_{2\tau} + s_1\nu_{2a} + (s_1-s_2) \nu_{2\delta} \big) \Big(\rho_{1\tau} + \frac{T-s_1-s_0}{s_1+s_0} \rho_{1\delta}\Big) + s_2\rho_{2\tau} + (T-s_2-s_0)\rho_{2\delta}\le \rho T.
\end{align}
The overall achievable rate region is the convex hull of these pairs over power allocations satisfying the power constraint and all feasible values of $s_0, s_1,s_2$.
\end{theorem}
\begin{remark}
The distribution of $\hat{\rvMat{G}}_{1e}$ is non-Gaussian. As clarified in~\eqref{eq:hybrid:MMSE_estimate}, it consists of a Gaussian matrix plus the product of two other Gaussian matrices.
\end{remark}
\color{black}
\begin{proof}
See the Appendix~\ref{proof:rate_hybrid}. 
\end{proof}

\begin{remark}
Hybrid superposition utilizes both rate splitting and product superposition but is not a generalization, in the sense that the results of pure rate splitting and product superposition cannot be recovered from the hybrid scheme. At very  high SNR under partially overlapped eigenspaces, hybrid superposition can improve over rate splitting and product superposition, but in other channel conditions, the hybrid superposition may in fact perform worse than the individual schemes. 
\end{remark}

\subsection{Numerical Results}

Simulations in this  section assume Rayleigh fading, i.e., $\rvMat{G}_k$ has independent $\Cc\Nc(0,1)$ entries. The correlation matrix $\Rm_k = \Um_k \Sigmam_k \Um_k^\H$, $k\in \{1,2\}$, is generated by assuming the same magnitude along all eigendirections, i.e., $\Sigmam_k = \tI$. Furthermore, we assume the eigendirections of transmit correlation matrices of the two users are either the same or orthogonal to each other. The simplicity of this configuration makes it suitable for a representative example. Assuming a constant magnitude along different eigendirections allows us to concentrate on gains that are {\em purely} due to correlation diversity rather than, e.g., water-filling.

When the eigenspaces of the two users are partially overlapped, in Fig.~\ref{fig:rate_regions}, we plot the rate regions  achieved with these schemes in a setting of $T = 24$, $M = 16$, $N_1 = N_2 = 12$, $r_1 = 16$, $r_0 = 6$, $r_2 = 10$ and $T = 32$, $M = N_1 = N_2 = 16$, $r_1 = 15$, $r_0 = 7$, $r_2 = 8$, at power constraint $\rho = 30$~dB. We observe that the performance of rate splitting and product superposition depends strongly on the rank of the eigenspaces. When the rank of the two individual eigenspaces is close to each other, rate splitting will obtain a better rate region since the gains achieved by product superposition come from the difference between the rank of the two eigenspaces. 
In the channel configuration in Fig.~\ref{fig:rate_regions}, the hybrid superposition scheme produced rates that are inferior to both product superposition and to rate splitting, therefore they are not displayed. Hybrid superposition becomes competitive at very high SNR, while the results of this section focus on moderate SNR.


\begin{figure}
\centering
\subfigure[$r_1 = 16$, $r_2 = 10$, $r_0 = 6$]{\includegraphics[width=.48\textwidth]{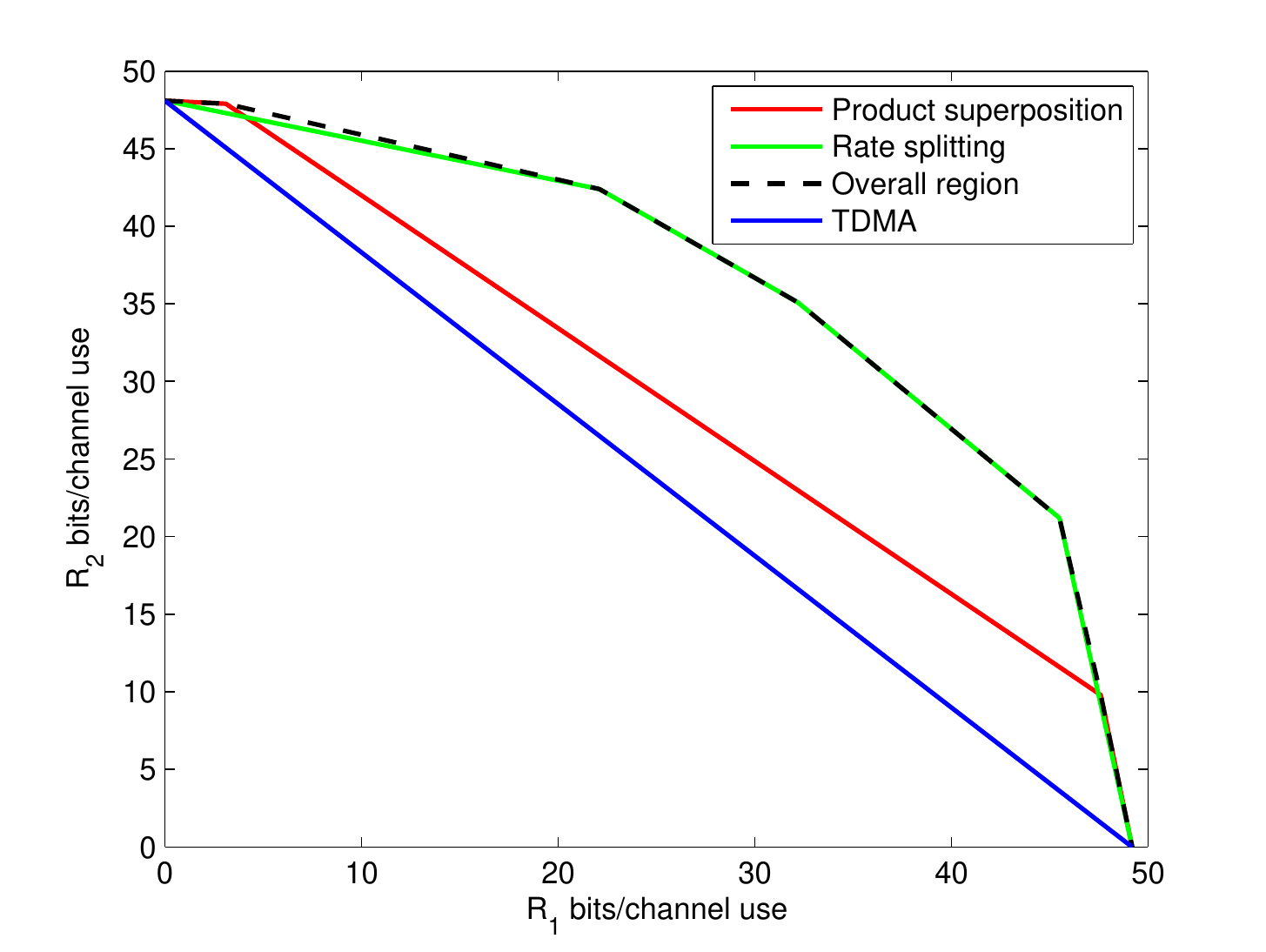}}
\subfigure[$r_1 = 15$, $r_2 = 8$, $r_0 = 7$]{\includegraphics[width=.48\textwidth]{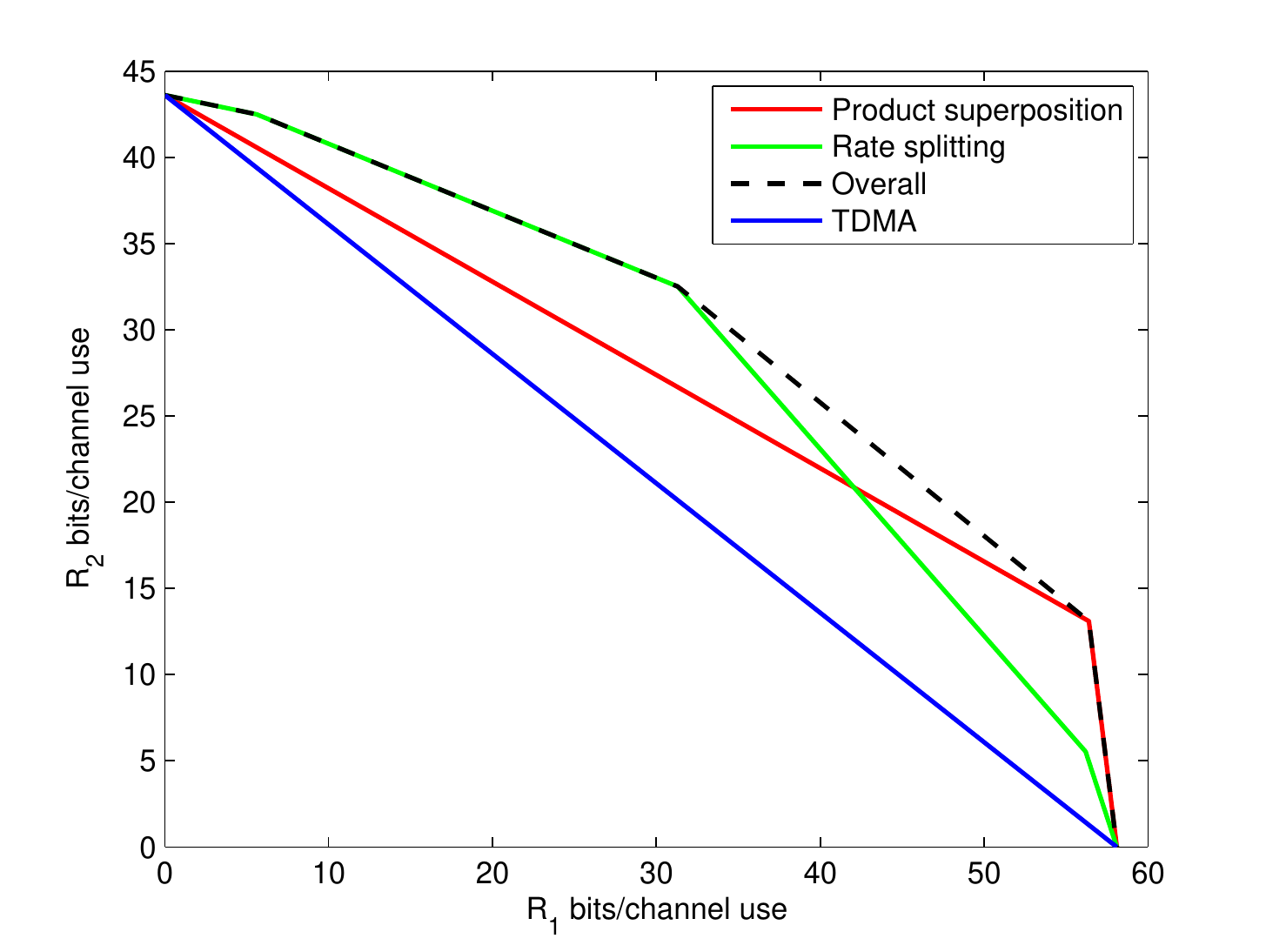}}
\caption{The rate regions of various schemes for the spatially correlated broadcast channel $\rho = 30$~dB.}
\label{fig:rate_regions}
\end{figure}

When one of the users' eigenspace is strictly a subspace of the other, rate splitting performs no better than TDMA. We plot the rate region for this scenario achieved via product superposition in a setting of $T = 20$, $M = N_1 = N_2 = 10$, $r_1 = 10,r_2 = 5,r_0 = 5$ and at power constraint $\rho = 30$~dB.
\begin{figure}
\centering
\subfigure[$r_1 = 5$, $r_2 = 0$, $r_0 = 5$]{\includegraphics[width=.48\textwidth]{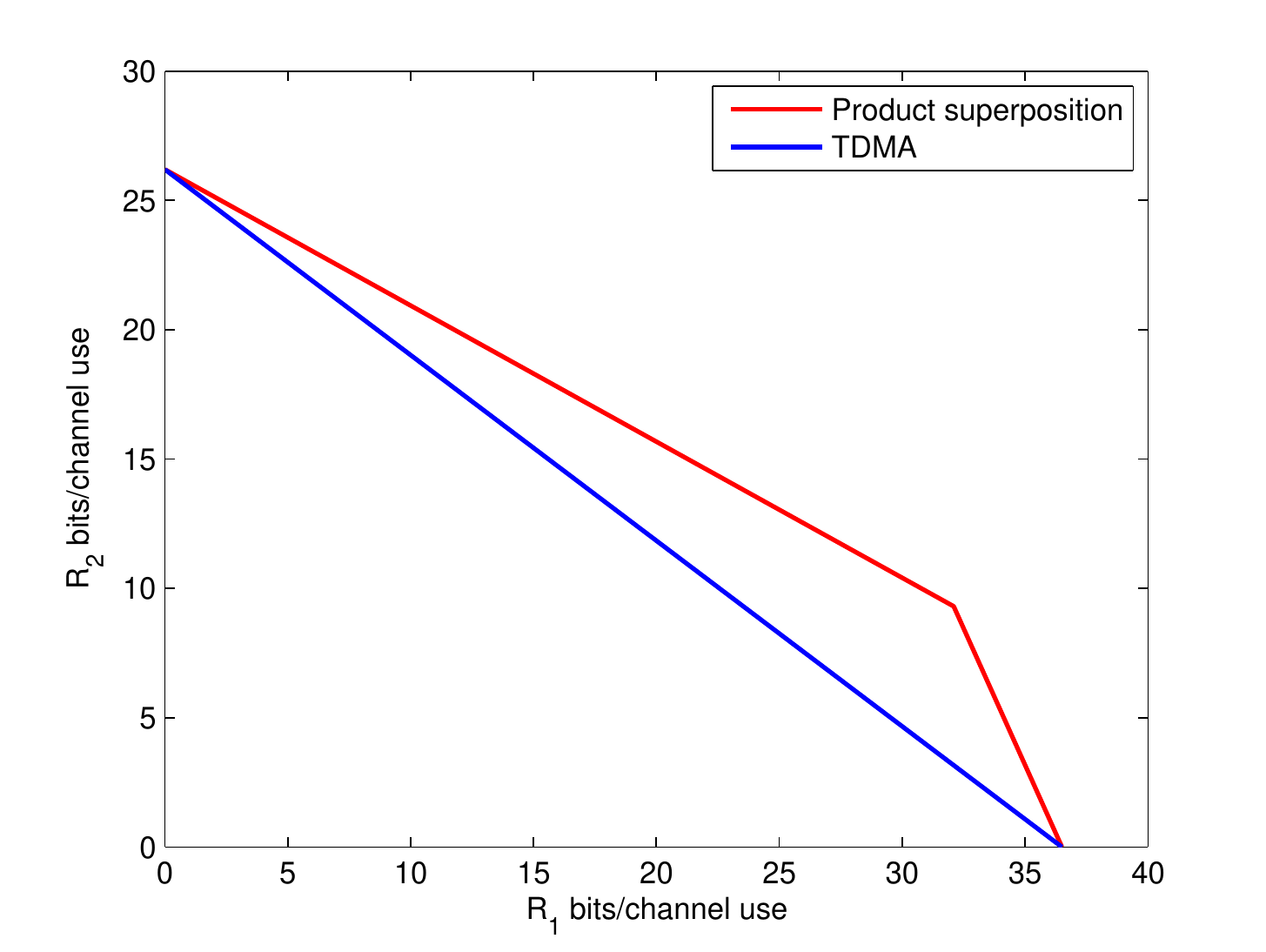}}
\caption{The rate region of $\rho = 30dB, M = N_1 = N_2 = 10, T = 20, r_1 = 10, r_0 = 5, r_2 = 5$}
\end{figure}

\section{$K$-user Broadcast Channel: DoF Analysis} \label{sec:Kuser_DoF}
To extend the study to the $K$-user scenario, some further assumptions on the correlation model are made as follows. Recall that the rows of $\tH_k$ belong to the eigenspace $\Span[\tU_k]$ of $\tR_k$. 

Denote the sum of all channel eigenspaces as follows\footnote{The sum of two subspaces is defined as $\Span(\tU)+\Span(\tV) \triangleq\Span(\tU\cup \tV)$. }
\begin{equation}
\Vc = \sum_{k \in [K]} \Span[\tU_k].
\end{equation}
Define $\mathcal{V}_{\mathcal{J}} = \bigcap_{k \in \mathcal{J}} \Span[\tU_k]$, for $|\mathcal{J}| > 1$, and $\mathcal{V}_{\{k\}} = \{\tv | \tv \in \Span[\tU_k], \tv \perp \mathcal{V}_{\mathcal{J}}, \forall \mathcal{J}, k \in \mathcal{J}, |\mathcal{J} | > 1\}$, $k \in [K]$.
Define $r_{\Jc} = \dim{(\Vc_{\Jc})}$. Obviously,
$
\sum_{\mathcal{J} \subset [K]} r_{\mathcal{J}} = \dim{(\mathcal{V})} \leq M$ and $\sum_{\mathcal{J} \subset [K]:\ k \in \mathcal{J}} r_{\mathcal{J}} = \dim{(\Span[\tU_k])} = r_k.
$ Therefore, we can generate $2^K - 1$ subspaces $\mathcal{V}_{\mathcal{J}}$ of $r_{\mathcal{J}}$ dimensional whose $r_{\mathcal{J}}$ basis vectors span the channel of every user in a non-empty group $\mathcal{J} \subset [K]$ and are linear independent to all vectors in $\Span[\tU_k]$ for $k \in \{[K] \setminus \mathcal{J}\}$.
\color{black}An example of the correlation structure for the case of three-user broadcast channel is shown in Fig.~\ref{fig:3-user-structure}.
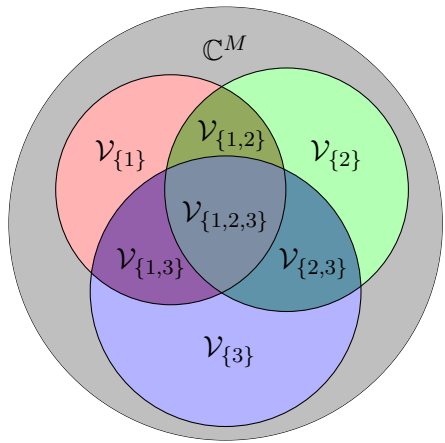
\begin{figure} 
\def \setA{ (-.21,0.2) circle (1.7cm) }
\def \setB{ (1.5,0.2) circle (1.8cm) }
\def \setC{ (.6,-1.3) circle (2cm) }
\def \myellipse { (.6, -.3) ellipse  (3.2cm and 3.2cm) }
\vspace{-.5cm}
\begin{center}
\begin{tikzpicture}[scale = .9]
\draw \myellipse;

\begin{scope}[even odd rule]
\clip \setA \setB \setC (-2.7,-3.5) rectangle (4,4);
\fill[lightgray] \myellipse;
\end{scope}
\draw (.6,2.3) node {$\CC^M$};

\begin{scope}[fill opacity=0.3]
\fill[red] \setA;
\fill[green] \setB;
\fill[blue] \setC;
\draw \setA;
\draw \setB;
\draw \setC;
\end{scope}
\draw (-.4,0.7) node[left] {$\Vc_{\{1\}}$};
\draw (1.7,.7) 	node[right] {$\Vc_{\{2\}}$};
\draw (1.2,-2.2) 	node[left] {$\Vc_{\{3\}}$};
\draw (.7,1) 	node {$\Vc_{\{1,2\}}$};
\draw (-.5,-0.9) node {$\Vc_{\{1,3\}}$};
\draw (1.9,-0.9) 	node {$\Vc_{\{2,3\}}$};
\draw (.6,-0.2) node {$\Vc_{\{1,2,3\}}$};
\end{tikzpicture}
\end{center}
\vspace{-.2cm}
\caption{The channel eigenspace overlapping structure of the three-user broadcast channel.} \label{fig:3-user-structure}
\end{figure}

In this way, the signal transmitted in the subspace $\Vc_\Jc$ can be seen by every user in $\Jc$ and is vague to all other users. On the other hand, the signals transmitted in $\mathcal{V_{\mathcal{J}}}$ and $\mathcal{V_{\mathcal{K}}}$ interfere each other at every user in $\mathcal{J} \cap \mathcal{K}$. To characterize the interfering relation between signals transmitted in different subspaces, we introduce the concept of {\em interference graph} as follows:
\begin{definition}
For $k \in [K]$, the {\em interference graph} of order $k$, denoted by $G(K,k)$, is an undirected graph for which: 
\begin{itemize}
\item the set of vertices is the set of unordered subsets of cardinality $k$ of $[K]$, i.e., {$\mathcal{J} \subset [K]:|\mathcal{J}|=k$}, hence a vertex is also denoted by a subset $\mathcal{J}$;
\item there exists an edge between two vertices $\mathcal{J}$ and $\mathcal{K}$  if and only if $\mathcal{J} \cap \mathcal{K} \neq \emptyset$.
\end{itemize}
\end{definition}
The interference graph $G(K,k)$ has $\binom{K}{k}$ vertices. It is a regular graph~\cite[Sec.~1.2]{Diestel2017graph_theory} of degree $\binom{K}{k} - \binom{K-k}{k} - 1$, with the convention $\binom{m}{n} = 0$ if $m < n$. 
Let $\mathcal{\chi}(G(K, k))$ denote the chromatic number of $G(K, k)$,
i.e., the minimum number of colors to color all the vertices such that adjacent vertices have different colors. 
We have the following property.
\begin{property}[The chromatic number of the interference graph] \label{prop:chromatic_number}
$\mathcal{\chi}\big(G(K, 1)\big) = 1$, $\mathcal{\chi}\big(G(K, k)\big) \leq \binom{K}{k} - \binom{K-k}{k} - 1$ when $1< k \le \floor{K/2}$, and $\mathcal{\chi}\big(G(K, k)\big) = \binom{K}{k}$ when $k > \floor{K/2}$. 
\end{property}
\begin{proof}
$\mathcal{\chi}\big(G(K, 1)\big) = 1$ since $G(K, 1)$ is edgeless. $\mathcal{\chi}\big(G(K, k)\big) = \binom{K}{k}$ when $k > \floor{K/2}$ because in this case, $G(K, k)$ is complete. The results for the case $1< k \le \floor{K/2}$ follows from Brook's theorem~\cite[Thm.~5.2.4]{Diestel2017graph_theory}.
\end{proof}

\begin{remark}
\label{remark:sketch_graph}
To avoid pilot interference, pilots in $\mathcal{V}_{\mathcal{J}}$ and $\mathcal{V}_{\mathcal{K}}$ need to be orthogonal in time if $\mathcal{J} \cap \mathcal{K} \neq \emptyset$, i.e, $\mathcal{J}$ and $\mathcal{K}$ are connected in the interference graph. Pilots in $\mathcal{V}_{\mathcal{J}}$ and $\mathcal{V}_{\mathcal{K}}$ can be transmitted simultaneously if $\mathcal{J} \cap \mathcal{K} = \emptyset$, i.e., $\mathcal{J}$ and $\mathcal{K}$ are not connected. Therefore, the problem of pilot alignment can be interpreted as interference graph coloring: pilots can be transmitted at the same time without interference in the subspaces corresponding to vertices with the same color. The minimum total amount of time for pilot transmission, normalized by the subspace dimension, is therefore the minimum number of colors, which is the chromatic number of the graph.

\end{remark}
\color{black}
\subsection{CSIR}
In this section, we assume the users have perfect CSIR. 
\begin{theorem}
\label{thm:K_wCSIR}
For the $K$-user broadcast channel with CSIR, for any integers $d_{\mathcal{J}}$ satisfy 
\begin{align}
d_{\mathcal{J}} &\leq r_{\mathcal{J}},\quad \forall \mathcal{J} \subset [K], \\
\sum_{\mathcal{J} \subset [K]: \ k \in \mathcal{J}} d_{\mathcal{J}} &\leq \min \big(r_k,N_k\big),\quad \forall k\in [K],
\end{align}
the DoF tuple $(d_1,\dots,d_K)$ given by
\begin{equation}
d_k = \sum_{\mathcal{J} \subset [K]:\ k \in \mathcal{J}} \tau_{k,\mathcal{J}}d_{\mathcal{J}},\quad k\in [K],
\end{equation}
for some time-sharing coefficients $\tau_{k,\mathcal{J}} \geq 0$ satisfying $\tau_{k,\mathcal{J}} = 0, \forall k \in \{[K] \setminus \mathcal{J}\}$ and $\sum_{i = 1}^{K} \tau_{k,\mathcal{J}} = 1, \forall \mathcal{J} \subset [K]$, is achievable.
\end{theorem}

\begin{proof}
For $\Jc \subset [K]$, let $\tV_{\mathcal{J}} \in \mathbb{C}^{M \times d_{\mathcal{J}}}$ be a matrix with orthonormal columns such that $\Span[\tV_\Jc] \subset \mathcal{V}_{\mathcal{J}}$. Then 
$
\mathbf{U}_k^\H \tV_{\mathcal{J}} = \mathbf{0},~ \forall k \notin \mathcal{J}$, and
$\rank[\mathbf{U}_k^\H \tV_{\mathcal{J}}] = d_{\mathcal{J}},~ \forall k \in \mathcal{J}.
$
Let the transmitter send the signal
\begin{equation}
\tX = \sum_{\mathcal{J} \subset [K]} \mathbf{V}_{\mathcal{J}}\ts_{\mathcal{J}},
\end{equation}
where 
$\ts_{\mathcal{J}} \in \mathbb{C}^{d_{\mathcal{J}}}$ contains data symbols. Let us consider User~$k$ and label the subsets in $\{ \mathcal{J} \subset [K]: k \in \mathcal{J} \}$ as $\{ \mathcal{J}_1,\dots,\mathcal{J}_l\}$. The received signal at User~$k$ is
\begin{equation}
\begin{split}
\tY_k 
& = {\tG}_k \Sigmam_k^{\frac12} \tU_k^\H [\tV_{\Jc_1} \ \dots \ \tV_{\Jc_l}]
\begin{bmatrix}
\ts_{\mathcal{J}_1} \\
\vdots \\
\ts_{\mathcal{J}_l}
\end{bmatrix}
+\tW_k.
\end{split}
\end{equation}
Because $\sum_{i=1}^l d_{\mathcal{J}_i} \leq \min (r_k,N_k)$, User~$k$ can decode $\ts_{\mathcal{J}_1},\dots,\ts_{\mathcal{J}_l}$, that is, $\{\ts_{\mathcal{J}} \subset [K]:\ k \in \mathcal{J} \}$, where the signal $\ts_\Jc$ provides $d_\Jc$ DoF. 
Signal $\ts_{\mathcal{J}}$ can be decoded by all the users in $\mathcal{J}$. By dedicating $\ts_{\mathcal{J}}$ to user $k \in \Jc$ in a fraction $\tau_{k,\mathcal{J}}$ of time, User~$k$ can achieve $\sum_{\mathcal{J} \in [K]: k \in \mathcal{J}} \tau_{k,\mathcal{J}}d_{\mathcal{J}}$ DoF. 
This completes the proof. 
\end{proof}

\subsection{No Free CSIR}
When the receivers have no free CSIR, we employ pilot-based schemes. As for the two-user case, we first consider the special case of fully overlapping eigenspaces and propose a product superposition scheme.
\subsubsection{Fully Overlapping Eigenspaces}
\begin{theorem}
\label{thm:kuser_fully}
For the $K$-user broadcast channel without free CSIR and the correlation eigenvectors are nested such that $\tU_{k-1} = [\bar{\tU}_{k} \ \tU_{k}]$ with $\bar{\tU}_{k}$ being a basis of the complement of $\Span[\tU_k]$ in $\Span[\tU_{k-1}]$, $k \in \{2,3,\dots,K\}$, 
the DoF tuple $(d_1,\dots,d_K)$ given by 
\begin{align}
d_1 = N_1^*\Big(1-\frac{r_1}{T}\Big) \quad \text{and} \quad d_k = N_k^*\frac{r_{k-1} - r_k}{T}, ~~k\in \{2,3, \ldots, K\} \label{eq:kuser}
\end{align}
is achievable.
\end{theorem}

\begin{proof} 
We develop the idea in the special case of 3 users, and then proceed to describe the $K$-user result. When $K = 3$, 
the transmitter sends 
\begin{equation}
\tX = \tU_1 \tX_2 \tX_1,
\end{equation}
with $\tX_1 = [\tI_{r_1} \ \tS_1] \in \CC^{r_1 \times T}$, $\tX_2 = 
\begin{bmatrix}
\bar{\tX}_{2} \\
\tX_3[\tI_{\rrtwo} ~ \tS_{2}]
\end{bmatrix} \in \mathbb{C}^{r_1 \times r_1}$, and 
$
\tX_3 = 
\begin{bmatrix}
\bar{\tX}_3 \\
[\tI_{r_3} ~ \tS_3]
\end{bmatrix} \in \mathbb{C}^{r_2 \times r_2},
$
where $\bar{\tX}_{k} \in \mathbb{C}^{(r_{k-1} - r_{k}) \times r_{k-1}}$ is designed to guarantee that $\tX_k$ is non-singular, $k\in \{2,3\}$; $\tS_1 \in \CC^{r_1\times (T-r_1)}$ contains symbols for User~$1$, and $\tS_k \in \mathbb{C}^{r_{k} \times (r_{k-1} - r_k)}$ contains symbols for User~$k$, $k \in \{2,3\}$. Because $\tU_1$ has orthogonal columns, the received signal at User~$1$ is
\begin{align}
\tY_1 & 
 \color{black} = {\tG}_1 \Sigmam_1^{\frac12} \tX_2 [\tI_{r_1} \ \tS_1] + \tW_1,
\end{align}
User~$1$ first estimates the equivalent channel ${\tG}_1 \Sigmam_1^{\frac12} \tX_2$ and then decodes $\tS_{1}$, achieving $N_1^*(T - r_1)$ DoF.

The received signal at User~$2$ during the first $r_1$ channel uses is
\begin{equation}
\begin{split}
\tY_{2[1:r_1]} &={\tG}_2 \Sigmam_2^{\frac12} \tX_3 [\tI_{r_2} ~ \tS_{2}] + \tW_{2[1:r_1]},
\end{split}
\end{equation}
\color{black}
User~$2$ estimates the equivalent channel ${\tG}_2 \Sigmam_2^{\frac12} \tX_3$ in the first $r_2$ channel uses, then decodes $\tS_{2}$ in the next $r_1 - r_2$ channel uses, achieving $N_2^* (r_1 - r_2)$ DoF.

The received signal at User~$3$ during the first $r_2$ channel uses is
\begin{equation}
\begin{split}
\tY_{3[1:r_2]} ={\tG}_3 \Sigmam_3^{\frac12} [\tI_{r_3} ~ \tS_3]+ \tW_{3[1:r_2]}.
\end{split}
\end{equation} 
\color{black}
During the first $r_3$ channel uses, User~$3$ estimates ${\tG}_3 \Sigmam_3^{\frac12}$, and then during the next $r_2 - r_3$ channel uses, User~$3$ decodes its symbols, achieving $N_3^* (r_2 - r_3)$ DoF. Therefore, for $K = 3$, the normalized DoF tuple \eqref{eq:kuser} is achieved. 

Now, we apply the same idea to the case of $K$ users. The transmitted signal is 
\begin{equation}
\tX = \tU_1 \tX_2 \tX_1,
\end{equation}
with $\tX_1 \!=\! [\tI_{r_1} \ \tS_1] \!\in\! \CC^{r_1 \times T}$, $\tX_k \!=\! 
\begin{bmatrix}
\bar{\tX}_{k} \\
\tX_{k+1}[\tI_{r_k} ~ \tS_k]
\end{bmatrix} \!\in\! \CC^{r_{k-1} \times r_{k-1}}$ for $k \!\in\! \{2, \ldots, K-1\}$, and
$\tX_K \!=\! 
\begin{bmatrix}
\bar{\tX}_{K} \\
[\tI_{r_K} ~ \tS_K]
\end{bmatrix} \!\in\! \CC^{r_{K-1} \times r_{K-1}}$.
User~$1$ uses the same decoding method as the case of $K=3$, achieving $N_1^*(T - r_1)$ DoF. For users $i = 2, \ldots, K-1$, consider the first $r_{k-1}$ channel uses, the received signal is 
\begin{equation}
\begin{split}
\tY_{k[1:r_{k-1}]} &
\color{black}={\tG}_k \Sigmam_k^\frac12 \tX_{k+1} [\tI_{r_k} ~ \tS_k]+ \tW_{k[1:r_{k-1}]},	  
\end{split}
\end{equation}
Therefore User~$k$ can achieve $N_k^*\frac{r_{k-1} - r_k}{T}$ DoF. With the same decoding method as User~$3$ in the $K=3$ case, User~$k$ can achieve $N_K^*\frac{r_{K-1} - r_K}{T}$ DoF. This completes the proof of Theorem~\ref{thm:kuser_fully}.
\end{proof}

\subsubsection{Partially Overlapping Eigenspaces} \label{sec:Kuser_noCSIR_partial}
We now consider the more general case of partially overlapping eigenspaces. We begin by analyzing symmetric $K$-user channels with overlapped eigenspaces, offering an achievable DoF region with rate splitting. Subsequently, the asymmetric case will also be analyzed.

For symmetric channels:
\begin{equation} \label{eq:symmetric_asumption}
r_{\mathcal{J}_i} = r_{\mathcal{J}_j}, \quad \forall \mathcal{J}_i,\mathcal{J}_j \subset [K]:|\mathcal{J}_i| = |\mathcal{J}_j|.
\end{equation}
That is, the rank of the common channel eigenspace $\mathcal{V}_{\mathcal{J}}$ is the same for all groups $\mathcal{J}$ containing the same number of
users. (In the two-user case, this corresponds to $r_1 = r_2$.) Define
\begin{equation}
p_k = r_{\mathcal{J}}, \quad \forall \mathcal{J} \subset [K]:|\mathcal{J}| = k,
\end{equation}
for $k\in [K]$. Then the set of parameters $(p_1,\dots,p_K)$ characterizes the correlation structure of the $K$-user symmetric broadcast channel. Furthermore, we assume that $r_k \le N_k$, $\forall k$.

\begin{theorem}
\label{thm:K_sym_nCSIR}
The $K$-user symmetric broadcast channel without free CSIR characterized by $(p_1,\dots,p_K)$ can achieve any permutation of the DoF tuple $D_{K,L}(p_1,\dots,p_K) = (d_1,d_2,\dots,d_K)$, for any $L \in \{0,1,\dots,K-1\}$, defined by
\begin{align}
d_k = \frac{1}{T} \sum_{i=1}^{K - \max (k-1,L)} \min\bigg(\binom{K - i}{i -1} p_i,N_k\bigg) \bigg(T - T_\tau (K,L) + \sum_{j = \floor{K/2} + 1}^{K-i} \binom{K - i}{k} p_j\bigg), \label{eq:Kuser_CDIR_DoF}
\end{align}
for $k \in [K]$, where
$
T_{\tau}(K,L) \defeq \sum_{k =1}^{K-L} \chi\big(G(K,k)\big)p_k.
$
\end{theorem}

Let us first describe the achievable scheme in the 3-user case for clarity, then go for the $K$-user case.

\begin{example} [Achievable scheme for Theorem~\ref{thm:K_sym_nCSIR} for $K = 3$] \normalfont
When $K = 3$, the correlation structure is illustrated
in Fig.~\ref{fig:3-user-structure}. Under the symmetry assumption, we have $r_{\{1\}} \!=\! r_{\{2\}} \!=\! r_{\{3\}} \triangleq p_1$, $\r_{\{1,2\}} \!=\! r_{\{1,3\}} \!=\! r_{\{2,3\}} \triangleq p_2$, $r_{\{1,2,3\}} \triangleq p_3$. The achievable scheme for 
\begin{align}
D_{3,0}(p_1,p_2,p_3) = \bigg((p_1 + 2p_2 + p_3)\Big(1 \!-\! \frac{T_\tau}{T}\Big)+ \frac{p_1 p_2}{T}, (p_1 + p_2)\Big(1 \!-\! \frac{T_\tau}{T}\Big) + \frac{p_1 p_2}{T}, p_1\Big(1 \!-\! \frac{T_\tau}{T}\Big) + \frac{p_1 p_2}{T}\bigg),
\end{align}
is based on rate splitting and channel training as illustrated in Table~\ref{tab:3user}.
\begin{table}[ht]
\caption{Illustration of pilot and data alignment for the scheme achieving $D_{3,0}(p_1,p_2,p_3)$ 
}
\begin{center}
\begin{tabular}{ r|c|c|c|c|c|c| } 
\hline
$\Vc_{\{1\}}$ & \cellcolor{LightCyan} Pilot & \cellcolor{LightGray} & \cellcolor{LightYellow} Data & \cellcolor{LightGray} & \cellcolor{LightGray} & \cellcolor{LightYellow} Data \\ \hline
$\Vc_{\{2\}}$ & \cellcolor{LightCyan} Pilot & \cellcolor{LightGray} & \cellcolor{LightGray} & \cellcolor{LightYellow} Data & \cellcolor{LightGray} & \cellcolor{LightYellow} Data\\ \hline
$\Vc_{\{3\}}$ & \cellcolor{LightCyan} Pilot & \cellcolor{LightYellow} Data & \cellcolor{LightGray} & \cellcolor{LightGray} & \cellcolor{LightGray} & \cellcolor{LightYellow} Data\\ \hline
$\Vc_{\{1,2\}}$ & \cellcolor{LightGray} & \cellcolor{LightCyan} Pilot & \cellcolor{LightGray} & \cellcolor{LightGray} & \cellcolor{LightGray} & \cellcolor{LightYellow} Data\\ \hline
$\Vc_{\{2,3\}}$ & \cellcolor{LightGray} & \cellcolor{LightGray} & \cellcolor{LightCyan} Pilot & \cellcolor{LightGray} & \cellcolor{LightGray} & \cellcolor{LightYellow} Data\\ \hline
$\Vc_{\{1,3\}}$ & \cellcolor{LightGray} & \cellcolor{LightGray} & \cellcolor{LightGray} & \cellcolor{LightCyan} Pilot & \cellcolor{LightGray} & \cellcolor{LightYellow} Data\\ \hline
$\Vc_{\{1,2,3\}}$ & \cellcolor{LightGray} & \cellcolor{LightGray} & \cellcolor{LightGray} & \cellcolor{LightGray} & \cellcolor{LightCyan} Pilot & \cellcolor{LightYellow} Data\\ \hline
& $\leftarrow \ p_1 \  \rightarrow$ & $\leftarrow p_2 \rightarrow$ & $\leftarrow p_2 \rightarrow$ & $\leftarrow p_2 \rightarrow$ & $\leftarrow \!p_3\! \rightarrow$ & $\xleftarrow{\hspace*{.3cm}} T - (p_1 \!+\! 3p_2 \!+\! p_3) \xrightarrow{\hspace*{.3cm}}$\\ 
\end{tabular}
\end{center}
\label{tab:3user}
\end{table}

Owing to linear precoding, choose a basis $\tV_\Jc$ of the subspace spanned by $\{\tv|\tv \in \Vc,\tv \perp \tV_\Kc, \forall \Kc \neq \Jc\}$. It can be proved that $\dim{(\Span[\tV_{\Jc}])} = r_{\Jc}$. We choose the precoder in this way but not directly choose a basis from $\tV_\Jc$, because for different $\Jc$, $\tV_\Jc$ is not guaranteed to be {\it orthogonal} with each other and we aim to remove the interference from the other channel component in $\tV_\Kc (\Kc \neq \Jc)$, so that all users in $\mathcal{J}$ can learn the channel directions in $\mathcal{V}_{\mathcal{J}}$. From Remark~\ref{remark:sketch_graph}, the required amount of pilot transmissions is identical with the chromatic number of the interference graph. 
\color{black}The interference graph $G(3, 1)$ has chromatic number $\chi\big(G(3, 1)\big) = 1$, which is also the amount of time, normalized by $p_1$, needed for pilot transmission without interference in $\mathcal{V}_{\{1\}}$, $\mathcal{V}_{\{2\}}$, and $\mathcal{V}_{\{3\}}$. Similarly, it takes $\chi\big(G(3, 2)\big)p_2 = 3p_2$ channel uses to transmit pilot interference-free in $\mathcal{V}_{\{1,2\}}$, $\mathcal{V}_{\{2,3\}}$, and $\mathcal{V}_{\{1,3\}}$, and takes $\chi\big(G(3, 3)\big)p_3 = p_3$ channel uses for pilot transmission in $\mathcal{V}_{\{1,2,3\}}$.

In this way, the total time for channel training is $T_\tau = \sum_{k=1}^{3} \chi\big(G(3,k)\big)p_k = p_1 + 3p_2 + p_3$ channel uses and there remains $T - T_{\tau}$ channel uses for simultaneous data transmission in all subspaces. By dedicating the data transmitted in $\mathcal{V}_{\{1,2\}}$, $\mathcal{V}_{\{1,3\}}$, and $\mathcal{V}_{\{1,2,3\}}$ to User~$1$, User~$1$ achieves $(p_1 + 2p_2 + p_3)(1 - \frac{T_\tau}{T})$ DoF. By dedicating the data transmitted in $\mathcal{V}_{\{2,3\}}$ to User~$2$, User~$2$ achieves $(p_1 + p_2)(1 - \frac{T\tau}{T})$ DoF. User~$3$ achieves $p_1(1 - \frac{T-\tau}{T})$ DoF from the data transmitted in $\mathcal{V}_{\{3\}}$. On top of that, the base station can transmit additional data to User~$3$ in $\mathcal{V}_{\{3\}}$ by superimposing it with the pilot for User~$1$ and User~$2$ in $\mathcal{V}_{\{1,2\}}$ without interference. Similarly, User~$1$ and User~$2$ can also receive additional data. With these additional data, each user achieves $\frac{p_1 p_2}{T}$ DoF. Therefore, $D_{3,0}(p_1,p_2,p_3)$ is achieved.

To achieve $D_{3,1}(p_1,p_2,p_3)$, which is
\begin{align}
\bigg(\big(p_1+2p_2\big)\Big(1-\frac{p_1+3p_2}{T}\Big) + \frac{p_1p_2}{T}, \ (p_1+p_2)\Big(1-\frac{p_1+3p_2}{T}\Big) + \frac{p_1p_2}{T}, \ p_1\Big(1-\frac{p_1+3p_2}{T}\Big)+\frac{p_1p_2}{T}\bigg),
\end{align}
we simply ignore the subspace $\mathcal{V}_{\{1,2,3\}}$. Then, we do not send pilot in this subspace and have more time to send data in all other subspaces. As a price for that, we lose the data we could send in $\mathcal{V}_{\{1,2,3\}}$ during the last $T - T\tau$ channel uses. When $\rank[\mathcal{V}_{\{1,2,3\}}] = p_3$ is small enough, this loss is not significant and we can gain DoF. The achievable scheme is illustrated in Table~\ref{tab:3-user-scheme-2}.
\begin{table} [ht]
\caption{Illustration of pilot and data alignment for the scheme achieving $D_{3,1}(p_1,p_2,p_3)$ 
}
\begin{center}
\begin{tabular}{ r|c|c|c|c|c| } 
\hline
$\Vc_{\{1\}}$ & \cellcolor{LightCyan} Pilot & \cellcolor{LightGray} & \cellcolor{LightYellow} Data & \cellcolor{LightGray} & \cellcolor{LightYellow} Data \\ \hline
$\Vc_{\{2\}}$ & \cellcolor{LightCyan} Pilot & \cellcolor{LightGray} & \cellcolor{LightGray} & \cellcolor{LightYellow} Data & \cellcolor{LightYellow} Data\\ \hline
$\Vc_{\{3\}}$ & \cellcolor{LightCyan} Pilot & \cellcolor{LightYellow} Data & \cellcolor{LightGray} & \cellcolor{LightGray} & \cellcolor{LightYellow} Data\\ \hline
$\Vc_{\{1,2\}}$ & \cellcolor{LightGray} & \cellcolor{LightCyan} Pilot & \cellcolor{LightGray} & \cellcolor{LightGray}  & \cellcolor{LightYellow} Data\\ \hline
$\Vc_{\{2,3\}}$ & \cellcolor{LightGray} & \cellcolor{LightGray} & \cellcolor{LightCyan} Pilot & \cellcolor{LightGray} & \cellcolor{LightYellow} Data\\ \hline
$\Vc_{\{1,3\}}$ & \cellcolor{LightGray} & \cellcolor{LightGray} & \cellcolor{LightGray} & \cellcolor{LightCyan} Pilot  & \cellcolor{LightYellow} Data\\ \hline
$\Vc_{\{1,2,3\}}$ & \cellcolor{LightGray} & \cellcolor{LightGray} & \cellcolor{LightGray} & \cellcolor{LightGray} & \cellcolor{LightGray}\\ \hline
& $\leftarrow \ p_1 \ \rightarrow$ & $\leftarrow p_2 \rightarrow$ & $\leftarrow p_2 \rightarrow$ & $\leftarrow p_2 \rightarrow$ & $\xleftarrow{\hspace*{1.3cm}} T \!-\! (p_1 \!+\! 3p_2) \xrightarrow{\hspace*{1.3cm}}$\\ 
\end{tabular}
\end{center}
\label{tab:3-user-scheme-2}
\end{table}

Similarly, $D_{3,2}(p_1, p_2, p_3) = \bigg(p_1\Big(1 - \frac{p_1}{T}\Big),p_2\Big(1 - \frac{p_2}{T}\Big) ,p_3\Big(1 - \frac{p_3}{T}\Big)\bigg)$
can be achieved by ignoring $\mathcal{V}_{\{1,2\}}$, $\mathcal{V}_{\{2,3\}}$, $\mathcal{V}_{\{1,3\}}$, and $\mathcal{V}_{\{1,2,3\}}$, as illustrated in Table~\ref{tab:3-user-scheme-3}.
\begin{table} [ht]
\caption{Illustration of pilot and data alignment for the scheme achieving $D_{3,2}(p_1,p_2,p_3)$ 
}
\begin{center}
\begin{tabular}{ r|c|c| } 
\hline
$\Vc_{\{1\}}$ & \cellcolor{LightCyan} Pilot & \cellcolor{LightYellow} Data \\ \hline
$\Vc_{\{2\}}$ & \cellcolor{LightCyan} Pilot & \cellcolor{LightYellow} Data\\ \hline
$\Vc_{\{3\}}$ & \cellcolor{LightCyan} Pilot & \cellcolor{LightYellow} Data\\ \hline
$\Vc_{\{1,2\}}$ & \cellcolor{LightGray} & \cellcolor{LightGray} \\ \hline
$\Vc_{\{2,3\}}$ & \cellcolor{LightGray} & \cellcolor{LightGray} \\ \hline
$\Vc_{\{1,3\}}$ & \cellcolor{LightGray} & \cellcolor{LightGray} \\ \hline
$\Vc_{\{1,2,3\}}$ & \cellcolor{LightGray} & \cellcolor{LightGray} \\ \hline
& $\leftarrow \ p_1\  \rightarrow$ & $\xleftarrow{\hspace*{4cm}} T \!-\! p_1 \xrightarrow{\hspace*{4cm}}$\\ 
\end{tabular}
\end{center}
\label{tab:3-user-scheme-3}
\end{table}

Due to symmetry, any permutation of $D_{3,L}$, $L\in \{0,1,2\}$ is achieved by permuting the users' indices.
\end{example}

\begin{proof}[Proof of Theorem~\ref{thm:K_sym_nCSIR}]
We first show the achievable scheme for $D_{K,0}(p_1,\dots, p_K)$ given by
\begin{align}
d_k = \frac{1}{T} \sum_{i = 1}^{K- k +1} \binom{K-i}{i-1} \twocolAlignMarker p_i \bigg(T - T_\tau(K,0) + \twocolbreak\sum_{j = \floor{K/2}+1}^K \binom{K - i}{j} p_j\bigg). \label{eq:tmp2662}
\end{align}
The scheme is based on rate splitting and channel training with two key elements: alignment
of pilots in different subspaces, and superposition of additional data on top of pilots without causing interference.

User~$k$ needs to learn the channel directions in all subspaces $\mathcal{V}_{\mathcal{J}}$ such that $k \in \mathcal{J}$ and is oblivious to signals (pilot or data) transmitted in other subspaces. 
From Remark~\ref{remark:sketch_graph}, the minimum total amount of time for pilot transmission in the common subspace by k users, normalized by the subspace dimension is given by the chromatic number of the interference graph $G(K,k)$. \color{black}Thus the total training time is $T_\tau (K,0) = \sum_{k=1}^{K} \chi\big(G(K,k)\big)p_k$ channel uses. In the remaining $T - T_\tau(K, 0)$ channel uses, data is transmitted in all subspaces. The DoF that User~$k$, $k \in \mathcal{J}$, can achieve with the message transmitted in $\mathcal{V}_{\mathcal{J}}$ is $\frac{1}{T} p_{|\mathcal{V}_{\mathcal{J}}|}\big(T - T_\tau(K,0)\big)$. 

Notice that for any $l > \floor{K/2}$, the interference graph $G(K, l)$ is fully connected, the pilots in subspaces $\mathcal{V}_{\mathcal{K}}$ for $|\mathcal{K}| = l$ cannot be transmitted at the same time. However, additional data can be transmitted in any subspace $\mathcal{V}_{\mathcal{J}}$ such that $\mathcal{V}_{\mathcal{K}} \cap \mathcal{V}_{\mathcal{J}} = \emptyset$. In this way, during the training of all subspaces $\mathcal{V}_{\mathcal{K}}$ with $|\mathcal{K}| = l$, for each subset $\mathcal{J}$ which does not intersect with $|\mathcal{K}|$, additional data can be transmitted in $\binom{K - |\mathcal{J}|}{l}p_l$ channel uses, enabling each user in $\mathcal{J}$ to achieve $\frac{1}{T}\binom{K - |\mathcal{J}|}{l}p_{\mathcal{J}} p_l$ more DoF.

Summing up the DoF, the number of DoF that each user in $\mathcal{J}$ can obtain from the message transmitted in $\mathcal{V}_{\mathcal{J}}$ is
\begin{align}
\twocolAlignMarker\frac{1}{T}p_{|\mathcal{J}|}\Big(T \!-\! T_\tau(K,0)\Big) + \frac{1}{T} \sum_{l = \floor{K/2}+1}^{K} \!\binom{K \!-\! |\mathcal{J}|}{l}p_{|\mathcal{J}|} p_l = \frac{1}{T} p_{|\mathcal{J}|}\bigg(T \!-\! T_\tau (K,0) + \sum_{l = \floor{K/2}+1}^{K}\! \binom{K \!-\! |\mathcal{J}|}{l} p_l\bigg).
\end{align}
By dedicating all the messages transmitted in $\mathcal{V}_{\mathcal{K}}$ such that $k \in \mathcal{J}$ and $\mathcal{J} \cap [k-1] = \emptyset$ to User~$k$, User~$k$ achieves $d_k$ DoF where $d_k$ is given in \eqref{eq:tmp2662}. Then $D_{K,0}(p_1,\dots, p_K)$ is achievable.

Similar to the $3$-user case, $D_{K,L}(p_1,\dots,p_K)$ with $L \in  [K - 1]$ is achieved by ignoring all the subspaces $\mathcal{V}_{\mathcal{J}}$ with $|\mathcal{J}| > K - L$.
Finally, due to symmetry, any permutation of $D_{K,L}(p_1,\dots,p_K)$ with $L = 1,\dots,K - 1$ can be achieved by
permutting the users' indices.
\end{proof}

\begin{remark}
We can improve the achievable scheme by sending additional data during the training of $\mathcal{V}_{\mathcal{K}}$ with $|\mathcal{K}| \leq \floor{K/2}$ also. However, the possibility for this additional data depends on the actual coloring of the interference graph and would not admit nice expressions of achievable DoF tuples. We therefore do not follow this direction in the interest of developing closed-form expressions.
\end{remark}

Computing the chromatic number $\chi\big(G(K, k)\big)$ is NP-complete in general~\cite{Garey1990NPcompleteness}. Therefore, one might confine to the achievable DoF tuples in the following corollary.
\begin{corollary}
The $K$-user symmetric broadcast channel without free CSIR can achieve the DoF tuple $D_{K,l}(p_1,\dots,p_K)$ given in Theorem~\ref{thm:K_sym_nCSIR}, with $T_\tau(K,L)$ replaced by $\sum_{k = 1}^{K -L}\Big(\binom{K}{k} - \binom{K - k}{k} - \mathbbm{1}\{1< k \le \floor{K/2}\}\Big)p_k$.
\end{corollary}
This corollary follows from Theorem~\ref{thm:K_sym_nCSIR} and Property~\ref{prop:chromatic_number}.

Based on Theorem~\ref{thm:K_sym_nCSIR}, we have the following achievable DoF region for the symmetric $K$-user channel.

\begin{theorem}
\label{thm:K_nCSIR}
The $K$-user symmetric MIMO broadcast channel without free CSIR characterized by $(p_1,\dots,p_K)$ can achieve the convex hull of all permutations of any DoF tuple of the form
\begin{align}
\Big(D_{k,L}(p_1^{\ast},\dots,p_k^{\ast}),0,\dots,0\Big), \quad \text{for~} k \in [K], L \in \{0,\dots,k-1\},
\label{eq:K_nCSIR}
\end{align}
with $D_{k,L}(\cdot)$ defined according to  \eqref{eq:Kuser_CDIR_DoF} and $p_l^{\ast} = \sum_{i = 0}^{K- k} \binom{K - k}{i} p_{l+i}$ for $l \in [k]$.
\end{theorem}

\begin{proof}
When $k = K$, \eqref{eq:K_nCSIR} becomes $D_{K,L}(p_1,\dots,p_K)$, which can be achieved as stated in Theorem~\ref{thm:K_sym_nCSIR}.

\color{black}
When $k < K$, by ignoring the last $K- k$ users, we construct a new symmetric channel with $k$ users. For example, by ignoring User~$3$ in the symmetric 3-user channel, we obtain a two-user channel in which the private subspace of User~$1$ and User~$2$ are $\mathcal{V}_{\{1\}} + \mathcal{V}_{\{1,3\}}$ and $\mathcal{V}_{\{2\}} + \mathcal{V}_{\{2,3\}}$, respectively, both of dimension $p_1^* = p_1 + p_2$; whereas the common subspace of two users is $\mathcal{V}_{\{1,2\}} + \mathcal{V}_{\{1,2,3\}}$ of dimension $p_2^* = p_2 + p_3$. In general, the new $K$-user channel is characterized by the new set of parameters $(p_1^*,\dots,p_k^*)$, where 
$
p_l^* = \sum_{i = 0}^{K - k} \binom{K - k}{i}p_{l + i}$, $l \in [k]$. Then, applying Theorem~\ref{thm:K_sym_nCSIR} to this $k$-user symmetric channel, the rate region $D_{k,L}(p_1^*,\dots,p_k^*)$ is achievable. Therefore, $\big(D_{k,L}(p_1^{\prime},\dots,p_k^{\prime}), 0 ,\dots,0\big)$ is achievable for the original $K$-user symmetric channel. Any permutation of \eqref{eq:K_nCSIR} can be achieved by permuting the users' indices. 
\end{proof}
Fig.~\ref{fig:3-user-p-region} demonstrates the achievable DoF region for the symmetric $3$-user broadcast channel given in Theorem~\ref{thm:K_nCSIR} with $T = 24, r_{\{1\}} = r_{\{2\}} = r_{\{3\}} = 4, r_{\{1,2\}} = r_{\{1,3\}} = r_{\{2,3\}} = 2$, and $r_{\{1,2,3\}} = 1$.
\begin{figure}
\centering
\includegraphics[width=\Figwidth]{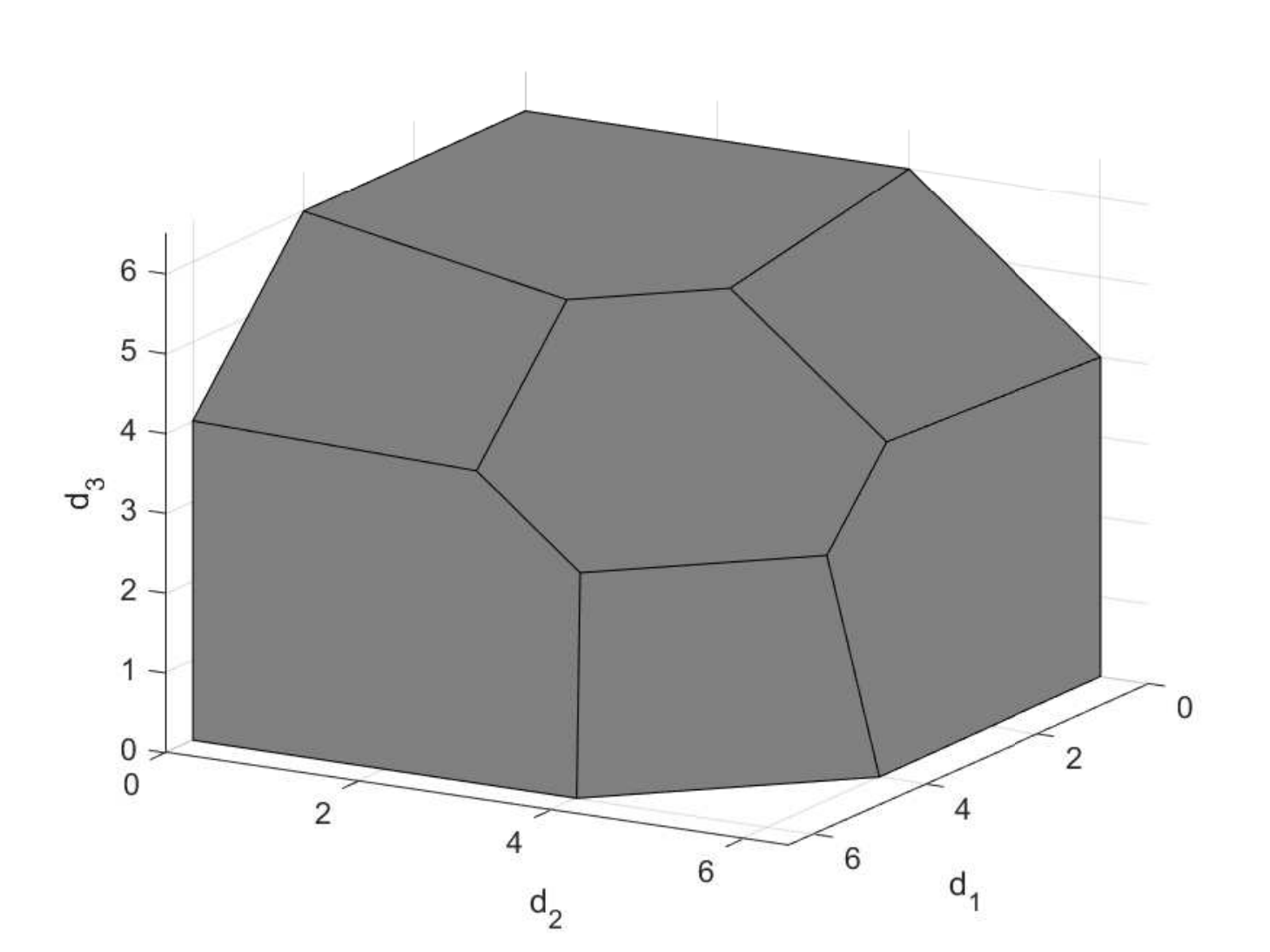}
\caption{An achievable DoF region of the symmetric $3$-user non-coherent broadcast channel with spatial correlation with $T = 24, r_{\{1\}} = r_{\{2\}} = r_{\{3\}} \eqdef p_1 = 4, r_{\{1,2\}} = r_{\{1,3\}} = r_{\{2,3\}} \eqdef p_2 = 2$, and $r_{\{1,2,3\}} \eqdef p_3 = 1$.}
\label{fig:3-user-p-region}
\end{figure}

We now broaden our analysis to $K$-user channels that may be {\em asymmetric}. 
The achievable scheme combines product superposition and rate splitting.
\begin{theorem} \label{thm:Kuser_CDIT_DoF_hybrid}
The $K$-user broadcast channel without free CSIR can achieve the DoF tuple $(d_1,\dots,d_K)$ given by
\begin{align} \label{eq:Kuser_CDIT_DoF_hybrid}
d_k = \sum_{\mathcal{J} \subset [k]:\ k \in \mathcal{J}} r_{\mathcal{J}}\Big(1 - \frac{r_k}{T}\Big) + \sum_{l = k + 1}^K \ \sum_{\mathcal{J} \subset [K]:\ k \in \mathcal{J},|\{k+1,\dots,K\} \cap \mathcal{J}| < 2} r_{\mathcal{J}}\frac{r_l - r_k}{T},
\end{align}
where it is assumed without loss of generality that $r_K \ge r_{K-1} \ge \dots \ge r_1$. 
\end{theorem}
\begin{proof}
For simplicity, let us focus on the $3$-user case. We assume without loss of generality that $r_3 \geq r_2 \geq r_1$. For each partition $\Vc_\Jc$, $\Jc \subset [3]$, we build a precoder $\tV_\Jc \in \CC^{M\times r_\Jc}$ as an orthonormal basis of $\Vc_\Jc$, thus $\tU_k^\H \tV_\Jc = \mathbf{0}$, $\forall k \notin \Jc$, and
$\rank[\tU_k^\H \tV_\Jc] = r_\Jc$, $\forall k \in \Jc$.
To combine rate splitting and product superposition, the transmitted signal is
\begin{equation}
\tX = [\tV_{\{1,2,3\}} \ \tV_{\{2,3\}} \ \tV_{\{1,3\}} \ \tV_{\{3\}}]\tilde{\tS}_2 \tS_3 + [\tV_{\{1,2,3\}} \ \tV_{\{2\}}]\tilde{\tS}_1\tS_2 + \tV_{\{1\}}\tS_1,
\end{equation}
with
\begin{align}
\tS_3 & = [\tI_{r_3} ~ \tS_{d,3}] \in \mathbb{C}^{r_3 \times T},\\
\tS_2 & = [\ZeroMat_{(r_{\{1,2\}} + r_{\{2\}}) \times (r_{\{1,2,3\}} + r_{\{2,3\}})} ~ \tI_{r_{\{1,2\}} + r_{\{2\}}} ~ \tS_{d2}] \in \mathbb{C}^{(r_{\{1,2\}} + r_{\{2\}}) \times T},\\
\tS_1 & = [\ZeroMat_{r_{\{1\}} \times (r_{\{1,2,3\}} + r_{\{1,3\}} + r_{\{1,2\}})} ~ \tI_r ~ \tS_{d1}] \in \mathbb{C}^{r_{\{1\}} \times T},\\
\tilde{\tS}_2 & = 
\begin{bmatrix}
\tI_{r_{\{1,2,3\}} + r_{\{2,3\}}} ~ \ZeroMat_{(r_{\{1,2,3\}} + r_{\{2,3\}}) \times (r_{\{1,2\}} + r_{\{2\}})} ~ \tilde{\tS}_{d2} \\
\ZeroMat_{r_{\{1,3\}} \times r_{\{1,2,3\}}} ~ \tI_{r_{\{1,3\}}} ~ \ZeroMat_{r_{\{1,3\}} \times (r_{\{1,2\}} + r_{\{1\}})} ~ \tilde{\tS}_{d21} \\
\bar{\tS}_2
\end{bmatrix} \in \mathbb{C}^{r_3 \times r_3},\\
\tilde{\tS}_1 & = 
\begin{bmatrix}
\ZeroMat_{r_{\{1,2\}} \times (r_{\{1,3\}} - r_{\{2,3\}})} ~ \tI_{r_{\{1,2\}}} ~ \ZeroMat_{r_{\{1,2\}} \times r_{\{1\}}} ~ \tilde{\tS}_{d1} \\
\bar{\tS}_1
\end{bmatrix} \in \mathbb{C}^{(r_{\{1,2\}} + r_{\{2\}}) \times (r_{\{1,2\}} + r_{\{2\}})},
\end{align} 
where $\bar{\tS}_2$ and $\bar{\tS}_1$ are designed to guarantee that $\tilde{\tS}_2$ and $\tilde{\tS}_1$ are respectively non-singular.

The received signal at User~$3$ is
\begin{equation}
\tY_3 = \tH_3 [\tV_{\{1,2,3\}} ~ \tV_{\{2,3\}} ~ \tV_{\{1,3\}} ~ \tV_{\{3\}}] \tilde{\tS}_2 [\tI_{r_3} ~ \tS_{d3}] + \tW_3.
\end{equation}
User~$3$ estimates the equivalent channel $\tH_3 [\tV_{\{1,2,3\}} ~ \tV_{\{2,3\}} ~ \tV_{\{1,3\}} ~ \tV_{\{3\}}]\tilde{\tS}_2$ in the first $r_3$ channel uses and then decode $\tS_{d3}$ to achieve full individual DoF $r_3 (1 - \frac{r_3}{T})$.

The received signal at User~$2$ is
\begin{align}
\tY_2 & = \tH_2 [\tV_{\{1,2,3\}} ~ \tV_{\{2,3\}}][\tI_{r_{\{1,2,3\}} + r_{\{2,3\}}} ~ \ZeroMat ~ \tilde{\tS}_{d2}]\tS_3 + \tH_2[\tV_{\{1,2\}} ~ \tV_{\{2\}}] \tilde{\tS}_1 \tS_2 + \tW_2 \\
& = \tH_2 [\tV_{\{1,2,3\}} ~ \tV_{\{2,3\}} ~ [\tV_{\{1,2\}} ~ \tV_{\{2\}}]\tilde{\tS}_1] 
\begin{bmatrix}
\tI_{r_{\{1,2,3\}} \times r_{\{2,3\}}} & \ZeroMat & [\tilde{\tS}_{d2} ~ \tB]  \\
\ZeroMat  & \tI_{r_{\{1,2\}} + r_{\{2\}}} & \tS_{d2}
\end{bmatrix} + \tW_2,
\end{align}
where $\tB \defeq [\tI_{r_{\{1,2,3\}} + r_{\{2,3\}}} ~ \ZeroMat ~ \tilde{\tS}_{d2}]\tS_{d3}$. User~$2$ can learn the equivalent channel $\tH_2 [\tV_{\{1,2,3\}} ~ \tV_{\{2,3\}} ~ [\tV_{\{1,2\}} ~ \tV_{\{2\}}]\tilde{\tS}_1] $ in the first $r_2$ channel uses and then decode both $\tilde{\tS}_{d2}$ and $\tS_{d2}$ to achieve $(r_{\{1,2,3\}} + r_{\{2,3\}})\frac{r_3 - r_2}{T} + (r_{\{1,2\}} + r_{\{2\}})(1 - \frac{r_2}{T})$ DoF in total.

The received signal at User~$1$ is
\begin{align}
\lefteqn{
\tY_1
}\notag \\
&= \tH_1[\tV_{\{1,2,3\}} ~ \tV_{\{1,3\}}] 
\begin{bmatrix}
\tI_{r_{\{1,2,3\}}} ~\ZeroMat_{r_{\{1,2,3\}} \times (r_{\{1,3\}} + r_{\{1,2\}} + r_{\{2\}})} \tilde{\tS}_{d2[1:r_{\{1,2,3\}}]} \\
\ZeroMat_{r_{\{1,3\}} \times r_{\{1,2,3\}}} ~ \tI_{r_{\{1,3\}}} ~ \ZeroMat_{r_{\{1,3\}} \times (r_{\{1,2\}} + r_{\{1\}})} ~ \tilde{\tS}_{d21}
\end{bmatrix} \tS_3 \notag\\
&\quad+ \tH_1 \tV_{\{1,2\}} [\ZeroMat_{r_{\{1,2\}} \times (r_{\{1,3\}} - r_{\{2,3\}})} ~ \tI_{r_{\{1,2\}}} \ZeroMat_{r_{\{1,2\}} \times r_{\{1\}}} ~ \tilde{\tS}_{d1}] \tS_2 + \tH_1 \tV_{\{1\}} \tS_1 + \tW_1
\\
&= \tH_1 [\tV_{\{1,2,3\}} ~ \tV_{\{1,3\}} ~ \tV_{\{1,2\}} ~ \tV_{\{1\}}]
\begin{bmatrix}
\tI_{\{1,2,3\}} & \ZeroMat & \ZeroMat & \ZeroMat & [\ZeroMat_{r_{\{1,2,3\}} \times (r_2 - r_1)} ~ \tilde{\tS}_{d2[1:r_{\{1,2,3\}}]} \ \tC] \\
\ZeroMat & \tI_{r_{\{1,3\}}} & \ZeroMat & \ZeroMat & [\tilde{\tS}_{d21} ~ \tD] \\
\ZeroMat & \ZeroMat & \tI_{r_{\{1,2\}}} & \ZeroMat & [\tilde{\tS}_{d1} ~ \tE] \\
\ZeroMat & \ZeroMat & \ZeroMat & \tI_{r_{\{1\}}} &  \tS_{d1} \\
\end{bmatrix} \!+\! \tW_1,
\end{align}
where 
$\tC \defeq [\tI_{\r_{\{1,2,3\}}} ~ \ZeroMat ~ \tilde{\tS}_{d2[1:r_{\{1,2,3\}}]}]\tS_{d3}$, $\tD \defeq [\ZeroMat_{r_{\{1,3\}} \times r_{\{1,2,3\}}} ~ \tI_{r_{\{1,3\}}} ~ \ZeroMat_{r_{\{1,3\}} \times (r_{\{1,2\}} + r_{\{1\}})} ~ \tilde{\tS}_{d21}]\tS_{d3}$, and $\tE \defeq [\ZeroMat_{r_{\{1,2\}} \times (r_{\{1,3\}} - r_{\{2,3\}})} \ \tI_{r_{\{1,2\}}} \ \ZeroMat_{r_{\{1,2\}} \times r_{\{1\}}} \ \tilde{\tS}_{d2}]$. User~$1$ learns the equivalent channel $\tH_1 [\tV_{\{1,2,3\}} ~ \tV_{\{1,3\}} ~ \tV_{\{1,2\}} ~ \tV_{\{1\}}]$ in the first $r_1$ channel uses then decode $\tilde{\tS}_{d21}$, $\tilde{\tS}_{d1}$ and $\tS_{d1}$ to achieve $r_{\{1,3\}}\frac{r_3 - r_1}{T} + r_{\{1,2\}}\frac{r_2 - r_1}{T} + r_{\{1\}}(1 - \frac{r_1}{T})$ DoF in total.
Therefore, the 3-user broadcast channel can achieve the DoF triple
\begin{equation}
\Big(r_3\big(1 - \frac{r_3}{T}\big),(r_{\{1,2,3\}} + r_{\{2,3\}})\frac{r_3 - r_2}{T} + (r_{\{1,2\}} + r_{\{2\}})\big(1 - \frac{r_2}{T}\big),r_{\{1,3\}}\frac{r_3 - r_1}{T} + r_{\{1,3\}}\frac{r_3 - r_1}{T} + r_{\{1\}}\big(1 - \frac{r_1}{T}\big)    \Big).
\end{equation}

Using similar reasoning, for the general $K$-user case such that $r_K \geq r_{K-1} \geq \dots \geq r_1$, the DoF in \eqref{eq:Kuser_CDIT_DoF_hybrid} is achievable.
\end{proof}

\color{black}

\section{Application in Massive MIMO}
\label{sec:mmimo}
In a massive MIMO system~\cite{massivemimobook}, the base station needs the CSI to beamform. However, due to the large number of antennas, the overhead for channel estimation is large. On the other hand, due to the limited space between the transmit antennas, the channel responses are normally spatially correlated. In this section, we exploit the spatial correlation to reduce the training overhead and compare the scheme with conventional training method.

We consider a multi-user massive MIMO system with a base station equipped with $M$ antennas communicating with $K$ single-antenna users with different spatial correlations. The channel vector corresponding to user $k \in [K]$ is $\rvVec{h}_k \in \CC^{M}$. 
The received signal of User~$k$ at time $t$ is $y(t) = {\bf h}_k^\T {\bf x}(t) + w(t)$, and during a coherence block is
\begin{equation}
\ty^\T_k = [y(1) \ y(2) \ \dots \ y(T)] = {\bf h}_k^\T {\bf X} + \tw^\T_k,
\end{equation}
where $\tX = [\tx(1)\ \tx(2) \ \dots \ \tx(T)]$ and $\tw_k = [w(1) \ w(2) \ \dots \ w(T)]^\T \sim \Cc\Nc(\mathbf{0},\Id_T)$. 
We assume that the system operates in FDD mode and focus on the downlink transmission. The transmission has two phases: the pilot phase and the data phase. During the pilot phase, pilot signal is sent so that the users can estimate the channel and then feedback the channel estimates to the base station. For simplicity and to focus on the gain of exploiting spatial correlation, we asume that feedback is perfect and instantaneous. After that, the base station sends data via beamforming.

\subsection{The Two-User Case} 
\label{sec:mmimo_2user}
We first consider the two-user scenario and assume that User~$1$ has uncorrelated channel and User~$2$ has spatially correlated channel of rank $r_2$. To extract an uncorrelated equivalent representation of ${\bf h}_2$, we define ${\bf g}_2 \in \CC^{r_2}$ via
\begin{equation}
{\bf h}_2 = \textbf{U} {\bf g}_2,
\end{equation}
where $\tU \defeq [{\bf u}_1 \ \dots \ {\bf u}_{M}]^\T \in \CC^{M \times r_2}$ is a truncated unitary matrix. 

Consider one coherence block. During the pilot phase, the transmitted signal is
\begin{equation}
\mathbf{X}_{[1:M]} = \sqrt{\rho} \ \diag[x_1,x_2,\dots,x_M],
\end{equation}
where $x_t=1$ for $t \in  \{1,2,...,r_2\}$, and $x_t$ is a Gaussian random variable following $\mathcal{CN}(0,1)$ for $t \in \{r_2+1,r_2+2,...,M\}$. In time slots $t = 1,2,\dots,r_2$, the received signal at User~$2$ is
$
y_2(t) = \sqrt{\rho} {\bf g}_2^\T {\bf u}_t  + w_2(t).
$
User~$2$ estimates ${\bf g}_2$ with a MMSE estimator
\begin{equation}
\hat{\bf g}_2 = \sqrt{\rho}[{\bf u}_1 \ \dots \ {\bf u}_{r_2}]^\H(\textbf{I}_{r_2}+\rho[{\bf u}_1 \ \dots\ {\bf u}_{r_2}][{\bf u}_1 \ \dots \ {\bf u}_{r_2}]^\H)^{-1}[y_2(1),\dots,y_2(r_1)]^\T.
\end{equation}
The estimation error is
$
\tilde{\bf g}_2 = {\bf g}_2 - \hat{\bf g}_2.
$
In time slots $t = r_2 + 1,\dots,M$, User~$2$ receives the signal
$
y_2(t)  = \sqrt{\rho}{\bf g}_2^\T{\bf u}_t x_t + w_2(t). 
$
User~$2$ uses the estimated channel to decode $[x_{r_1 + 1},\dots,x_M]$, achieving the rate
\begin{equation}
\Delta R_2 = \frac{M - r_2}{T} \E[\log\bigg(1 +  \frac{\rho}{\rho\E[\|\tilde{\bf g}_2^\T {\bf u}_t\|^2]+1} \|\hat{\bf g}_2^\T {\bf u}_t\|^2\bigg)].
\end{equation}

The received signal at User~$1$ in the pilot phase is
\begin{equation}
({\bf y}^\T_{1})_{[1:M]} = [y_1(1) \ \dots \ y_1(M)] = {\bf h}_1^\T\mathbf{X} + ({\bf w}^\T_{1})_{[1:M]}.
\end{equation}
User~$2$ estimates ${\bf h}_1^\T\mathbf{X}$ by $\frac{\rho}{\rho+1}({\bf y}^\T_{1})_{[1:M]}$ and feeds back to the base station. Because the base station knows $\tX$, it can obtain the estimation of ${\bf h}_1$ as
$
\hat{{\bf h}}_1 = \frac{\rho}{\rho + 1}\tX^{-\T} ({\bf y}_{1})_{[1:M]}. 
$
The estimation error is $\tilde{\bf h}_1 = {\bf h}_1 - \hat{\bf h}_1.$

Let $\hat{\bf h}_2 = \mathbf{U} \hat{\bf g}_2$ and $\tilde{\bf h}_2 = \mathbf{U} \tilde{\bf g}_1$. During the data phase, i.e. time slots $t = M+1,\dots,T$, the transmitted signal via conjugate beamforming is 
$
{\bf x}(t) = \sqrt{\frac{\rho}{2}}\frac{\hat{\bf h}^*_1}{\|\hat{\bf h}_1\|}s_1(t) + \sqrt{\frac{\rho}{2}}\frac{\hat{{\bf h}}^*_2}{\|\hat{{\bf h}}_2\|}s_2 (t),
$
where $s_k(t)$ is the data symbol for user $k \in \{1,2\}$ following the $\Cc\Nc(0,1)$ distribution. The received signals at the two users are
\begin{align}
y_{1}(t) &= \sqrt{\frac{\rho}{2}}\frac{{\bf h}_1^\T\hat{\bf h}^*_1}{\|\hat{\bf h}_1\|}s_1 (t) + \sqrt{\frac{\rho}{2}}\frac{{\bf h}_1^\T\hat{{\bf h}}^*_2}{\|\hat{{\bf h}}_2\|}s_2(t) + w_1(t), \\
y_{2}(t) &= \sqrt{\frac{\rho}{2}}\frac{{\bf h}_2^\T\hat{{\bf h}}^*_2}{\|\hat{{\bf h}}_2\|}s_2(t) + \sqrt{\frac{\rho}{2}}\frac{{\bf h}_2^\T \hat{\bf h}^*_1}{\|\hat{\bf h}_1\|}s_1(t) + w_{2}(t).
\end{align}
The achievable rate for User~$k$ is:
\begin{equation}
R_{k} = \bigg(1 - \frac{M}{T}\bigg)\E[\log\big(1 + \rho_{i}\|\hat{{\bf h}}_i\|^2\big)], \quad k=1,2,
\end{equation}
where the equivalent SNRs are defined as
$
\rho_1 \defeq \Big(\E[\frac{|\tilde{\bf h}_1^\T\hat{\bf h}^*_1|^2}{\|\hat{\bf h}_1\|^2} \!+\! \frac{|{\bf h}_1^\T\hat{{\bf h}}^*_2|^2}{\|\hat{{\bf h}}_2\|^2}] \!+\! \frac{2}{\rho}\Big)^{-1}
$ and 
$
\rho_2 \defeq \Big(\E[\frac{|\tilde{\bf h}_2^\T\hat{{\bf h}}^*_2|^2}{\|\hat{{\bf h}}_2\|^2} \!+\! \frac{|{\bf h}_2^\T \hat{\bf h}^*_1|^2}{\|\hat{\bf h}_1\|^2}] \!+\! \frac{2}{\rho}\Big)^{-1}.
$
The achievable sum rate is
\begin{equation}
R = R_1 + R_2 + \Delta R_2.
\end{equation}

For conventional transmission, the transmitter ignores the condition that two users need different number of pilots and sends $M$ pilots over $M$ time slots, the users estimate the channel and feedback to the transmitter. Then the transmitter communicates with the users via conjugate beamforming~\cite{massivemimobook}. Figure~\ref{fig_mmimo_2} shows the performance of the proposed scheme in comparison with the conventional one under Rayleigh fading, $M = 32$, $T = 64$, User~$1$ has fully correlated channel and User~$2$ has uncorrelated channel.
\begin{figure}[!ht]	
\centering
\includegraphics[width = \Figwidth]{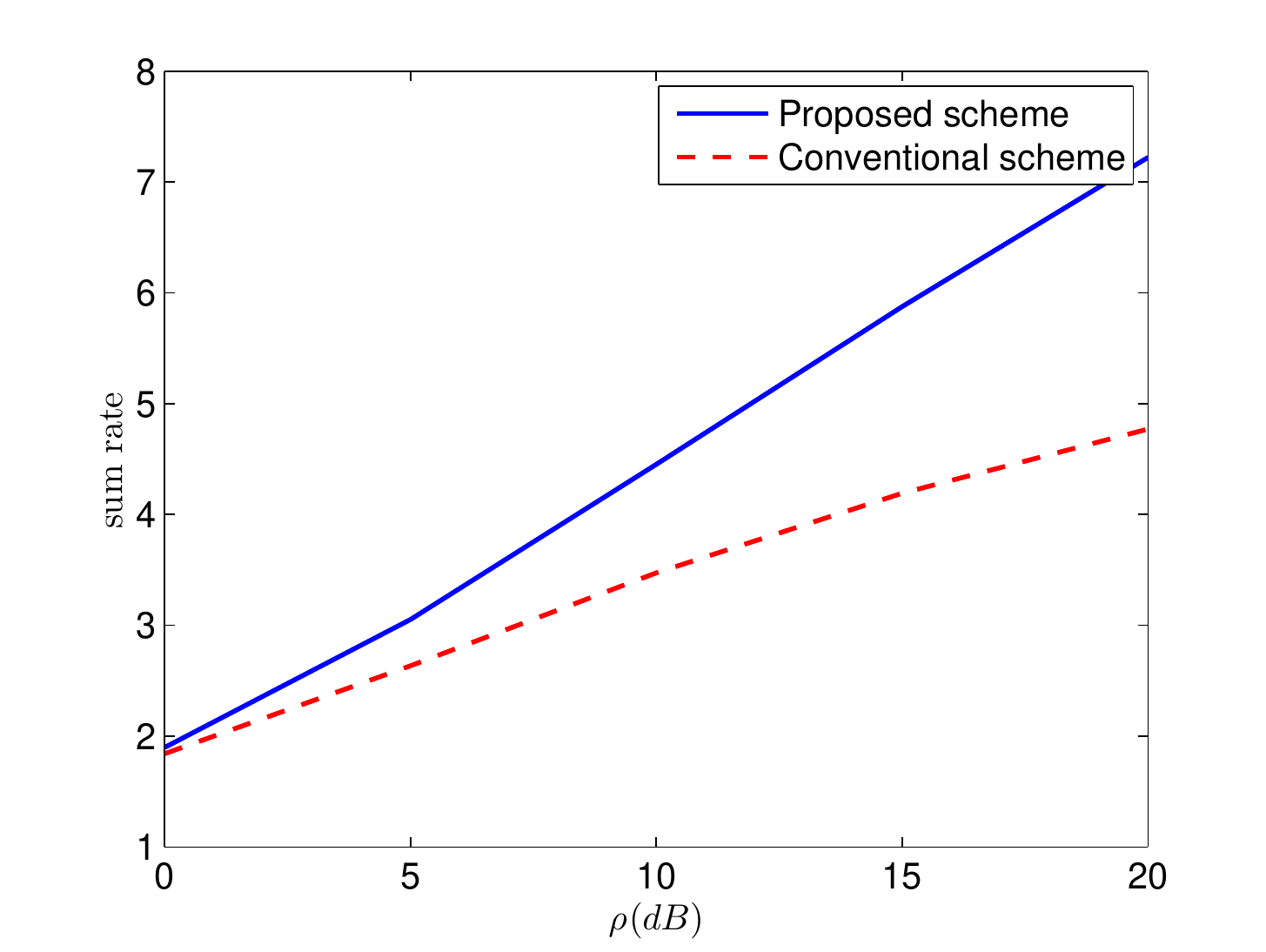}
\caption{The sum rate of the considered FDD massive MIMO system with the proposed scheme in comparison with the conventional scheme for $K = 2, M = 32$, $T = 64$, User~$1$ has fully correlated channel, and User~$2$ has uncorrelated channel.}
\label{fig_mmimo_2}
\end{figure}

We now generalize to the case where both users experience spatially correlated links and have partially overlapping eigenspaces. 
Recall that the eigendirections for the two users are $\tU_k$, where $\tU_k \in \mathbb{C}^{M \times r_k}$, for $k=1,2$. We assume without loss of generality that $r_1 \geq r_2$.
We find transmit eigendirections with orthonormal columns $\tV_{0}$ that are aligned with the common part of the two channel eigenspaces and $\tV_1, \tV_2$ that are aligned with the non-common parts, i.e., $\tV_{0} \in \mathbb{C}^{M \times r_{0}}$, $\tV_1 \in \mathbb{C}^{M \times (r_1 - r_{0})}$, $\tV_2 \in \mathbb{C}^{M \times (r_2 - r_{0})}$ such that
\begin{align}
\spanvzero & = \spanone \cap \spantwo,\\
\spanvone & = \spanone \cap \spantwo^\perp,\\
\spanvtwo & = \spantwo \cap \spanone^\perp.
\end{align}
Therefore, we can write ${\bf h}_k = [\tV_0 \ \tV_k] {\bf g}_k$ where ${\bf g}_k \in \CC^{r_k}$, $k = 1,2$.

The proposed scheme has two phases. The pilot phase has $r_1$ time slots, and the data phase has $T - r_1$ time slots. In the pilot phase, the base station sends pilots in the subspace of $\tV_{0}$ in time slots 1 to $r_0$, 
$
\tX_{[1:r_0]} = \sqrt{\rho}\tV_0^*.
$
The received signal at User~$k$ is
\begin{equation}
(\ty_k^\T)_{[1:r_0]} = \sqrt{\rho} {\bf h}_k^\T \tV_0^* +  (\tw_k^\T)_{[1:r_0]} = \sqrt{\rho}{\bf g}_k^\T \Bigg[\begin{matrix}
\tI_{r_0} \\
\ZeroMat_{(r_k - r_0)\times r_0}
\end{matrix}\Bigg]
+ (\tw_k^\T)_{[1:r_0]}.
\end{equation}
In the next $r_2 - r_{0}$ time slots, the base station sends pilots to two users simultaneously in subspaces $\tV_1$ and $\tV_2$, the transmitted signal is
\begin{equation}
\tX_{[r_0+1:r_2]} = \sqrt{\frac{\rho}{2}}\Bigg(\tV_1^*
\Bigg[\begin{matrix}
\tI_{r_2-r_0} \\
\ZeroMat_{(r_1 - r_2)\times (r_2-r_0)}
\end{matrix}\Bigg]
+ \tV_2^*\Bigg).
\end{equation}
The received signals at two users are:
\begin{align}
(\ty_1^\T)_{[r_0+1:r_2]} &= {\bf h}_1^\T \tX_{[r_0+1:r_2]}
+ (\tw_1^\T)_{[r_0+1:r_2]} 
= \sqrt{\frac{\rho}{2}}{\bf g}_1^\T 
\begin{bmatrix}
\ZeroMat_{r_0 \times (r_2-r_0)} \\
\tI_{r_2-r_0} \\
\ZeroMat_{(r_1 - r_2) \times (r_2-r_0)}
\end{bmatrix} 
+ (\tw_1^\T)_{[r_0+1:r_2]}, \\ 
(\ty_2^\T)_{[r_0+1:r_2]} &= {\bf h}_2^\T \tX_{[r_0+1:r_2]} + (\tw_2^\T)_{[r_0+1:r_2]}
= \sqrt{\frac{\rho}{2}}{\bf g}_2^\T
\begin{bmatrix}
\ZeroMat_{r_0 \times (r_2-r_0)} \\
\tI_{r_2-r_0}
\end{bmatrix} 
+ (\tw_2^\T)_{[r_0+1:r_2]}.
\end{align}
Based on $(\ty_2^\T)_{[1:r_2]}$, User~$2$ obtains a MMSE estimates $\hat{\bf g}_2 = \frac{\sqrt{\rho/2}}{\rho/2+1} (\ty_2)_{[1:r_2]}$ of ${\bf g}_2$ and feeds back to the base station. The estimation error is $\tilde{\tg}_2 = \tg_2 - \hat{\tg}_2$. 
In time slots $r_2 + 1$ to $r_1$, the base station sends pilots for User~$1$ in the remaining eigenspaces and sends data to User~$2$ via beamforming as
\begin{equation}
\tX_{[r_2+1:r_1]} = \sqrt{\frac{\rho}{2}} \Bigg(\tV_1^* \Bigg[ 
\begin{matrix}
\ZeroMat_{r_2 \times (r_1 - r_2)} \\
\tI_{r_1 - r_2}
\end{matrix} \Bigg] 
+  [\ZeroMat_{r_2 \times r_0} ~ \tV^*_2] \frac{\hat{{\bf g}}^*_2}{\|(\hat{{\bf g}}^\T_2)_{[r_0+1:r_2]}\|} \ts_{22}^\T \Bigg),
\end{equation} 
where $\ts_{22} \in \CC^{r_1-r_2}$ contains i.i.d. $\Cc\Nc(0,1)$ data symbols.
The received signal at User~$1$ is:
\begin{equation}
(\ty_1^\T)_{[r_2+1:r_1]}  = {\bf h}_1^\T \tX_{[r_2+1:r_1]}
+ (\tw_1^\T)_{[r_2+1:r_1]} 
= \sqrt{\frac{\rho}{2}} {\bf g}_1^\T 
\begin{bmatrix}
\ZeroMat_{r_2 \times (r_1 - r_2)} \\
\tI_{r_1 - r_2}
\end{bmatrix} 
 + (\tw_1^\T)_{[r_2+1:r_1]}.
\end{equation}
Based on $(\ty_1^\T)_{[1:r_1]}$, User~$1$ obtains a MMSE estimates $\hat{\bf g}_1 = \frac{\sqrt{\rho/2}}{\rho/2+1} (\ty_1)_{[1:r_1]}$ of ${\bf g}_1$ and feeds back to the base station. The estimation error is $\tilde{\tg}_1 = \tg_1 - \hat{\tg}_1$. The received signal at User~$2$ is
\begin{align}
(\ty_2^\T)_{[r_2+1:r_1]}  & = {\bf h}_2^\T \tX_{[r_2+1:r_1]} + (\tw_2^\T)_{[r_2+1:r_1]}   \\
& = \sqrt{\frac{\rho}{2}}{\bf g}_2^\T 
\begin{bmatrix}
\ZeroMat_{r_0 \times r_0} & \ZeroMat_{r_0 \times (r_2 - r_0)} \\
\ZeroMat_{(r_2 - r_0) \times r_0} & \tI_{r_2 - r_0}
\end{bmatrix}
\frac{\hat{{\bf g}}^*_2}{\|(\hat{{\bf g}}^\T_2)_{[r_0+1:r_2]}\|} \ts_{22}^\T + (\tw_2^\T)_{[r_2+1:r_1]} \\ 
&= \sqrt{\frac{\rho}{2}}\big\|(\hat{\bf g}_2^\T)_{[r_0+1:r_2]} \big\| \ts_{22}^\T
+ 
\sqrt{\frac{\rho}{2}}
\frac{(\tilde{\bf g}_2^\T)_{[r_0+1:r_2]} (\hat{{\bf g}}^*_2)_{[r_0+1:r_2]}}{\|(\hat{{\bf g}}^\T_2)_{[r_0+1:r_2]}\|} \ts_{22}^\T + (\tw_2^\T)_{[r_2+1:r_1]}.
\end{align}
User~$2$ decodes $\ts_{22}$  and achieves the rate
\begin{equation}
\Delta R_2 = \frac{r_1 - r_2}{T}\EE \bigg[\log\bigg(1 + \frac{\frac{\rho}{2}\|(\hat{\bf g}_2^\T)_{[r_0+1:r_2]} \|^2}{\frac{\rho}{2}\EE \big[\frac{|(\tilde{\bf g}_2^\T)_{[r_0+1:r_2]} (\hat{{\bf g}}^*_2)_{[r_0+1:r_2]}|^2}{\|(\hat{{\bf g}}^\T_2)_{[r_0+1:r_2]}\|^2}\big]+ 1}\bigg)\bigg]. 
\end{equation}

With the help of the feedback, the base station generates estimation for the two channels via
$
\hat{\bf h}_1 = [\tV_{0} \ \tV_1]\hat{{\bf g }}_1$, and
$\hat{{\bf h}}_2 = [\tV_{0} \ \tV_2]\hat{{\bf g }}_2.
$
The estimation errors are $\tilde{\bf h}_1 = \th_1 - \hat{\th}_1 $ and $\tilde{\bf h}_2 = \th_2 - \hat{\th}_2$. 
During the data phase, the transmitted signal via conjugate beamforming is 
\begin{equation}
{\bf X}_{[r_1+1:T]} = \sqrt{\frac{\rho}{2}}\frac{\hat{\bf h}^*_1}{\|\hat{\bf h}_1\|}\ts_1^\T + \sqrt{\frac{\rho}{2}}\frac{\hat{{\bf h}}^*_2}{\|\hat{{\bf h}}_2\|}\ts_2^\T.
\end{equation}
where $\ts_k \in \CC^{T-r_1}$, $k = 1,2,$ contains i.i.d. $\Cc\Nc(0,1)$ data symbols for User~$k$.
The received signals at the two users are
\begin{align}
(\ty_1^\T)_{[r_1+1:T]} &= \sqrt{\frac{\rho}{2}}\frac{{\bf h}_1^\T\hat{\bf h}^*_1}{\|\hat{\bf h}_1\|}\ts_1^\T + \sqrt{\frac{\rho}{2}}\frac{{\bf h}_1^\T\hat{{\bf h}}^*_2}{\|\hat{{\bf h}}_2\|}\ts_2^\T + (\tw_1^\T)_{[r_1+1:T]}, \\
(\ty_2^\T)_{[r_1+1:T]} &= \sqrt{\frac{\rho}{2}}\frac{{\bf h}_2^\T\hat{{\bf h}}^*_2}{\|\hat{{\bf h}}_2\|}\ts_2^\T + \sqrt{\frac{\rho}{2}}\frac{{\bf h}_2^\T\hat{\bf h}^*_1}{\|\hat{\bf h}_1\|}\ts_1^\T + (\tw_2^\T)_{[r_1+1:T]}.
\end{align}
User~$k$ decodes $\ts_k$ and achieves the rate
\begin{equation}
R_{k} = \bigg(1 - \frac{r_1}{T}\bigg)\E[\log\big(1 + \rho_{k}\|\hat{\bf h}_k\|^2\big)], \quad k = 1,2,
\end{equation}
with the equivalent SNRs $\rho_1 \defeq \Big(\E[\frac{|\tilde{\bf h}_1^\T\hat{\bf h}^*_1|^2}{\|\hat{\bf h}_1\|^2} + \frac{|{\bf h}_1^\T\hat{{\bf h}}^*_2|^2}{\|\hat{{\bf h}}_2\|^2}]+ \frac{2}{\rho}\Big)^{-1}$ and $\rho_2 \defeq \Big(\E[\frac{|\tilde{\bf h}_2^\T\hat{\bf h}^*_2|^2}{\|\hat{\bf h}_2\|^2} + \frac{|{\bf h}_2^\T\hat{{\bf h}}^*_1|^2}{\|\hat{{\bf h}}_1\|^2}]+ \frac{2}{\rho}\Big)^{-1}$.

The achievable sum rate is:
\begin{equation}
R = R_1 + R_2 + \Delta R_2.
\end{equation}

In the next subsections, we consider the $K$-user case. In this case, for a general (irregular) correlation structure, the signal design matching the correlations is complicated. Therefore, in order to emphasize the gain of correlation-based rate splitting and product superposition, we focus on some special configurations of the eigenspaces. 

\subsection{The $K$-User Case with Symmetric Eigenspace} \label{sec:mmimo_Kuser_symmetric}
The first considered special eigenspace configuration for the $K$-user case is the symmetric correlation structure as in ~\ref{sec:Kuser_noCSIR_partial}. We first present the case when $K = 3$. Under the symmetry assumption, we have $r_{\{1\}} = r_{\{2\}} = r_{\{3\}} \triangleq p_1$, $r_{\{1,2\}} = r_{\{1,3\}} = r_{\{2,3\}} \triangleq p_2$, and $r_{\{1,2,3\}} \triangleq p_3$.

Define the matrix $\tV$ as the collection of all the eigendirection vectors, which means
\begin{equation}
\tV = \big[\tV_{\{1\}} \ \tV_{\{2\}} \ \tV_{\{3\}}\ \tV_{\{1,2\}}\ \tV_{\{1,3\}}\ \tV_{\{2,3\}} \ \tV_{\{1,2,3\}}\big]
\end{equation}
where $\tV_\Jc \in \CC^{M \times r_\Jc}$ contains the eigenvectors spanning the subspaces of all users in $\Jc$.
Now we decompose the channel as ${\bf h}_k = [\tV_{\Jc}]_{k\in \Jc} {\bf g}_k$ where ${\bf g}_k \in \mathbb{C}^{r_k}$. For example, $\th_1 = \big[\tV_{\{1\}} \ \tV_{\{1,2\}}\ \tV_{\{1,3\}}\ \tV_{\{1,2,3\}}\big] \tg_1.$

In the first $p_1$ time slots, the base station sends pilots to three users simultaneously in subspaces $\tV_{\{1\}}$ ,$\tV_{\{2\}}$ and $\tV_{\{3\}}$. The transmitted signal is
\begin{equation}
\tX_{[1:p_1]} = \sqrt{\frac{\rho}{3}}(\tV_{\{1\}}^* + \tV_{\{2\}}^* + \tV_{\{3\}}^*).
\end{equation}

The received signal at User~$k$ is
\begin{equation}
(\ty_k^\T)_{[1:p_1]} = \sqrt{\frac{\rho}{3}}{\bf h}_k^\T \tV_{\{k\}}^* + (\tw_k^\T)_{[1:p_1]}.
\end{equation}
User~$k$ estimates ${\bf h}_k^\T \tV_{\{k\}}^*$ to obtain $\hat{\bf h}_k^\T \tV_{\{k\}}^*$ and feeds back to the base station. The estimation error is $\tilde{\bf h}_k^\T \tV_{\{k\}}^* = {\bf h}_k^\T \tV_{\{k\}}^* - \hat{\bf h}_k^\T \tV_{\{k\}}^*$. 
In the next $3p_2$ time slots, the base station sends pilots to users $i$ and $j$ in the subspace of $\tV_{\{i,j\}} (i \neq j)$ and data to the remaining user via conjugate beamforming. For example, in the first $p_2$ time slots, it sends
\begin{equation}
\tX_{[p_1+1:p_1+p_2]} = \sqrt{\frac{\rho}{2}}\tV_{\{2,3\}}^*  + \sqrt{\frac{\rho}{2}} \tV_{\{1\}}^* \frac{\tV_{\{1\}}^\T \hat{\bf h}^*_1}{\|\tV_{\{1\}}^\T \hat{\bf h}^*_1\|}  \ts_{11}^\T,
\end{equation} 
where $\ts_{11} \in \CC^{p_2}$ contains i.i.d. $\Cc\Nc(0,1)$ data symbols.
The received signal at User~$2$ or User~$3$ is
\begin{equation}
(\ty^\T_k)_{[p_1+1:p_1+p_2]} = \sqrt{\frac{\rho}{2}}{\bf h}_k^\T \tV_{\{2,3\}}^* + (\tw^\T_k)_{[p_1+1:p_1+p_2]}, \quad k = 2,3.
\end{equation}
User~$k~(k = 2,3)$ estimates ${\bf h}_k^\T \tV_{\{2,3\}}^*$ to obtain $\hat{\bf h}_k^\T \tV_{\{2,3\}}^*$ and feeds back to the base station. The received signal at User~$1$ is
\begin{align}
(\ty^\T_1)_{[p_1+1:p_1+p_2]} &= \sqrt{\frac{\rho}{2}} {\bf h}_1^\T \tV_{\{1\}}^* \frac{\tV_{\{1\}}^\T \hat{\bf h}^*_1}{\|\tV_{\{1\}}^\T \hat{\bf h}^*_1\|}  \ts_{11} + (\tw^\T_1)_{[p_1+1:p_1+p_2]} \\
&= \sqrt{\frac{\rho}{2}}  \|\tV_{\{1\}}^\T \hat{\bf h}^*_1\| \ts_{11} + \sqrt{\frac{\rho}{2}}  \tilde{\bf h}_1^\T \tV_{\{1\}}^* \frac{\tV_{\{1\}}^\T \hat{\bf h}^*_1}{\|\tV_{\{1\}}^\T \hat{\bf h}^*_1\|}  \ts_{11} + (\tw^\T_1)_{[p_1+1:p_1+p_2]}.
\end{align}
User~$1$ decodes $\ts_{11}$ and achieves the rate
\begin{equation}
\Delta R_1 = \frac{p_2}{T}\EE \bigg[\log\bigg(1 + \frac{\frac{\rho}{2}\|\tV_{\{1\}}^\T \hat{\bf h}^*_1\|^2}{\frac{\rho}{2}\EE\big[\frac{\tilde{|\bf h}_1^\T \tV_{\{1\}}^*\tV_{\{1\}}^\T \hat{\bf h}^*_1|^2}{\|\tV_{\{1\}}^\T \hat{\bf h}^*_1\|^2}\big] + 1}\bigg)\bigg].
\end{equation}
In the subsequent $p_2$ time slots, the channel coefficients in $\Vc_{\{1,3\}}$, $\Vc_{\{1,2\}}$ are estimated and and fed back, and the achievable rate for User~$2$ and User~$3$ can be calculated similarly.

In the following $p_3$ time slots, the base station transmits pilots in $\Vc_{\{1,2,3\}}$ as $
\tX_{[3p_2+1:3p_2+p_3]} = \sqrt{\rho}\tV^*_{\{1,2,3\}}.
$
User~$k$ receives $(\ty^T_k)_{[3p_2+1:3p_2+p_3]} = \sqrt{\rho}{\bf h}_k^\T \tV^*_{\{1,2,3\}} + (\ty^T_k)_{[3p_2+1:3p_2+p_3]}$, estimates ${\bf h}_k^\T \tV^*_{\{1,2,3\}}$ to obtain $\hat{\bf h}_k^\T \tV^*_{\{1,2,3\}}$ and feeds back to the base station.
From the feedbacks in the first $T_\tau = p_1 + 3p_2 + p3$ time slots, the base station obtains estimates $\hat{\th}_k$ of $\th_k$, $k = 1,2,3$. The estimation error is $\tilde{\th}_k = \th_k - \hat{\th}_k$.

During the data phase, the transmitted signal via conjugate beamforming is 
\begin{equation}
{\bf X}_{[T_\tau+1:T]} = \sqrt{\frac{\rho}{3}}\frac{\hat{\bf h}^*_1}{\|\hat{\bf h}_1\|}\ts_1^\T + \sqrt{\frac{\rho}{3}}\frac{\hat{{\bf h}}^*_2}{\|\hat{{\bf h}}_2\|}\ts_2^\T + \sqrt{\frac{\rho}{3}}\frac{\hat{{\bf h}}^*_3}{\|\hat{{\bf h}}_3\|}\ts_3^\T,
\end{equation}
where $\ts_k\in \CC^{T-T_\tau}$ contains i.i.d. $\Cc\Nc(0,1)$ data symbols for User~$k$. 
The received signals at User~$1$ is
\begin{align}
(\ty_1^\T)_{[T_\tau+1:T]} = \sqrt{\frac{\rho}{3}}\frac{{\bf h}_1^\T\hat{\bf h}^*_1}{\|\hat{\bf h}_1\|}\ts_1^\T + \sqrt{\frac{\rho}{3}}\frac{{\bf h}_1^\T\hat{{\bf h}}^*_2}{\|\hat{{\bf h}}_2\|}\ts_2^\T + \sqrt{\frac{\rho}{3}}\frac{{\bf h}_1^\T\hat{{\bf h}}^*_3}{\|\hat{{\bf h}}_3\|}\ts_3^\T + (\tw_1^\T)_{[T_\tau+1:T]}.
\end{align}
User~$1$ decodes $\ts_1$ and achieves the rate
\begin{equation}
R_{1} = \bigg(1 - \frac{p_1 + 3p_2 + p_3}{T}\bigg)\EE \bigg[\log\bigg(1 + \frac{\frac{\rho}{3}\|\hat{\bf h}_1\|^2}{\frac{\rho}{3}\E[\frac{|\tilde{\bf h}_1^\T\hat{\bf h}^*_1|^2}{\|\hat{\bf h}_1\|^2} + \frac{|{\bf h}_1^\T\hat{{\bf h}}^*_2|^2}{\|\hat{{\bf h}}_2\|^2} + \frac{|{\bf h}_1^\T\hat{{\bf h}}^*_3|^2}{\|\hat{{\bf h}}_3\|^2}]+ 1}\bigg)\bigg].
\end{equation}
The achievable rate of User~$2$ and User~$3$ can be calculated in the same way.

The achievable sum rate is
\begin{equation}
R = \sum_{k = 1}^3 (R_k + \Delta R_k).
\end{equation}

Now we extend this scheme to the $K$-user scenario. Following the signaling structure developed in the $3$-user case, the transmit scheme has three phases.  In the first phase, some pilot signals are transmitted. In the second phase, the remaining pilots are transmitted while at the same time, some users also receive data. In the third phase, channel state is known (due to pilots transmitted in the earlier two phases) and the base station beamforms to all users. The pilots and data arrangement is similar to the achievable scheme for Theorem~\ref{thm:K_sym_nCSIR}. 

The first phase has $\sum_{l = 1}^{\floor{K/2}} \chi(G(K,l))p_l$ time slots, in the first $\chi (G(K,1)) p_1 = p_1 $ time slots, the base station sends $\frac{\rho}{K}\sum_{i = 1}^K \tV_{\{i\}}$. In the same way, during the following time slots, the base station sends pilots which will not interfere with each other. The users estimate the channel coefficients in these subspaces and feed back to the base station.

The second phase has $\sum_{l = \floor{K/2} + 1}^{K} \chi(G(K,l))p_l$ time slots, where $\chi(G(K,l)) = \binom{K}{l}$. In this phase, the base station sends pilot in some eigendirections and simultaneously beamforms to the users which are not interfered by the pilots. For example, when sending the pilots in $\Vc_{\{1,2,\dots,K-2\}}$, the transmitted signal is
\begin{equation}
\tX = \sqrt{\frac{\rho}{3}}\tV_{\{1,2,\dots,K-2\}}^* + \sqrt{\frac{\rho}{3}}\tV_{\{K-1\}}^* \frac{\tV_{\{K-1\}}^\T \hat{\bf h}^*_{K-1}}{\|\tV_{\{K-1\}}^\T \hat{\bf h}^*_{K-1}\|}\ts_{K-1}^\T+
\sqrt{\frac{\rho}{3}}\tV_{\{K\}}^* \frac{\tV_{\{K\}}^\T \hat{\bf h}^*_{K}}{\|\tV_{\{K\}}^\T \hat{\bf h}^*_{K}\|} \ts_{K}^\T,
\end{equation}
where the equivalent channels $\tV_{\{K-1\}}^\T \hat{\bf h}^*_{K-1}$ and $\tV_{\{K\}}^\T \hat{\bf h}^*_{K}$ have been estimated and fed back in the first phase.
During these time slots, User~$1$ to User~$K-2$ can estimate their channel coefficients in the direction of $\tV_{\{1,2,\dots,K-2\}}$, while user~$K-1$ can decode $\ts_{K-1}$ and User~$k$ can decode $\ts_{K}$.

In the third phase, which has $T - T_\tau(K,0)$ time slots, the base station beamforms to all users with the estimated channel by sending
\begin{equation}
{\bf x}(t) = \sqrt{\frac{\rho}{K}}\sum_{k = 1}^K \frac{\hat{{\bf h}}_k}{\|\hat{{\bf h}}_k\|}s_k(t)
\end{equation}
at time slot $t = T - T_\tau(K,0) + 1, \dots, T$.

Finally, the total rate that can be achieved is the sum of the rates achieved during phases two and three.

\subsection{The $K$-User Case with On-Off Correlation} \label{sec:mmimo_Kuser}
The second special correlation configuration is motivated as follows. Experience shows that small values of correlation are often inconsequential to the rate and thus can be treated as uncorrelation in signal design. Furthermore, interference-free pilot reuse is only made possible under rank deficient correlation matrices, i.e., some transmit antenna gains are fully deterministic conditioned on the others. Therefore, we consider a $K$-user channel where the pairs of transmit antennas are either uncorrelated or fully correlated for each user, and refer to it as {\it on-off correlation}. Specifically, consider the channel vector ${\bf h}_k = [h_{k,1} \ h_{k,2}\ \dots\ h_{k,M}]^\T$ of any User~$k$, for any $i,j \in [M]$, we assume that either $h_{k,i} = h_{k,j}$ (fully correlated) or $\E[h_{k,i}^*h_{k,j}] = 0$ (uncorrelated). 

Consider the case where the channel coefficients of User~$1$ are fully correlated, the channel coefficients of User~$k$ are uncorrelated, while the remaining $K - 2$ users have fully correlated channel coefficients with respect to some antennas. Let us group the antennas into $L + 1$ groups: the first group has the first antenna, the $l$-th group has $\frac{M-1}{L}$ antennas from $\frac{(M-1)(l-1)}{L}+2$ to $\frac{(M-1)l}{L}+1$. We assign the users to each group as follows: User~$k$ is assigned to group $l$ if the channel coefficients of User~$k$ corresponding to the antennas in group $l$ are fully correlated, i.e $h_{k,\frac{(M-1)(l-1)}{L}+1+i}=h_{k,\frac{(M-1)(l-1)}{L}+1+j}$, for $1 \leq i,j \leq \frac{M-1}{L}$. Because User~$1$ has fully correlated channel coefficients, it is assigned to every group. 

The base station transmits the following signal in the pilot phase:
\begin{equation}
\mathbf{X}_{[1:M]} = \sqrt\rho \ \diag[v_0,v_1,v_1{\bf u}_1^\T,v_2,v_2{\bf u}_2^\T,\dots,v_L,v_L{\bf u}^\T_L],
\end{equation}
where $v_0 = 1$ and $v_l \in \mathbb{C}$, ${\bf u}_l \in \mathbb{C}^{\frac{M-1}{L}-1}$, $l = 1,2,...,L$ are mutually independent random variables following the distribution $\mathcal{CN}(0,1)$. Here $\{v_l\}_{l=1}^L$ are the symbols for User~$1$ and ${\bf u}_l$ is for one of the users in group $l$. The received signal at User~$k$ is:
\begin{equation}
(\ty^\T_k)_{[1:M]} = \sqrt{\rho} {\bf h}_k^\T \textbf{X} + (\tw^\T_k)_{[1:M]}.
\end{equation}
User~$k$ estimates $\textbf{X}^\T{\bf h}_k$ via MMSE and feeds back the estimated version $\frac{\sqrt{\rho}}{\rho+1} ({\ty_k})_{[1:M]}$ to the base station. Because the base station knows $\textbf{X}$, it obtains an estimated version of the channel of User~$k$ as $\hat{\th}_k = \frac{\sqrt{\rho}}{\rho+1} \textbf{X}^{-\T} ({\ty_k})_{[1:M]}$. The estimation error is $\tilde{\th}_k = \th_k - \hat{\th}_k$.

Denote the fully correlated channel coefficient of User~$1$ as $\bar{h}_1 \defeq h_{1,1} = h_{1,2} = \dots = h_{1,M}$. In the first time slot, User~$1$ receives $y_1(1) = \sqrt{\rho} \bar{h}_1 + w_1(1)$. It estimates $\bar{h}_1$ by $\hat{\bar{h}}_1 = \frac{\sqrt{\rho}}{\rho+1} y_1(1)$ and the estimation error is $\tilde{\bar{h}}_1 = \bar{h}_1 - \hat{\bar{h}}_1$. We have that $\hat{\bar{h}}_1 \sim \Cc\Nc\big(0,\frac{\rho}{\rho+1}\big)$ and $\tilde{\bar{h}}_1 \sim  \Cc\Nc\big(0,\frac{1}{\rho+1}\big)$. In the time slots $\frac{(M-1)(l-1)}{L}+2$, $l = 1,\dots,L$, User~$1$ receive $y_1(\tfrac{(M-1)(l-1)}{L}+2) = \sqrt{\rho} \bar{h}_1 v_l + w_1(\tfrac{(M-1)(l-1)}{L}+2)$. User~$1$ can decode $\{v_l\}$ and achieves the rate
\begin{align}
\Delta R_1 &= \frac{L}{T}\E[{\log\Bigg(1 + \frac{\rho \big|\hat{\bar{h}}_1 \big|^2}{\rho\E[\big|\tilde{\bar{h}}_1\big|^2]+1}\Bigg)}] \\
&= \frac{L}{T}\E[{\log\bigg(1 + \frac{\rho(\rho+1)}{2\rho+1}\big|\hat{\bar{h}}_1\big|^2\bigg)}] \\
&= \frac{L}{T} \log(e) \exp\Big(\frac{2\rho+1}{\rho^2}\Big)E_1\Big(\frac{2\rho+1}{\rho^2}\Big),
\end{align}
where $E_1(x) \defeq \int_{x}^{\infty}\frac{e^{-t}}{t} {\rm d} t$ is the exponential integral function.

In addition, if User~$k$ is assigned to group $l+1$, $l = 1,\dots,L$, denote $h_{k,\frac{(M-1)l}{L}+2} v_{l+1} = h_{k,\frac{(M-1)l}{L}+3} v_{l+1} = \dots = h_{k,\frac{(M-1)(l+1)}{L}+1} v_{l+1} \defeq \bar{h}_{k,l+1}$. In time slot $\frac{(M-1)l}{L}+2$, the received signal of User~$k$ is
\begin{equation}
y_k\big(\tfrac{(M-1)l}{L}+2\big) = \sqrt{\rho} \bar{h}_{k,l+1} + w_k\big(\tfrac{(M-1)l}{L}+2\big).
\end{equation}
User~$k$ can estimate the equivalent channel $\bar{h}_{k,l+1}$ by $\hat{\bar{h}}_{k,l+1} = \frac{\sqrt{\rho}}{\rho+1}y_k\big(\tfrac{(M-1)l}{L}+2\big)$ and the estimation error is $\tilde{\bar{h}}_{k,l+1} = \bar{h}_{k,l+1} - \hat{\bar{h}}_{k,l+1}$. We have that $\E[|\hat{\bar{h}}_{k,l+1}|^2] = \frac{\rho}{\rho+1}$ and $\E[|\hat{\bar{h}}_{k,l+1}|^2] = \frac{1}{\rho+1}$.
In the next $\frac{M-1}{L}-1$ time slots, User~$k$ receives
\begin{align}
y_k\big(\tfrac{(M-1)l}{L}+2+t\big) &= \sqrt{\rho} \bar{h}_{k,l+1} u_{l+1,t} + w_k\big(\tfrac{(M-1)l}{L}+2+t\big)
\end{align}
for $t = 1,2,\dots,\frac{M-1}{L}-1$.
Therefore, User~$k$ can decode ${\bf u}_{l+1}$ and achieve the rate
\begin{align}
\Delta R_{l+1} &= \frac{(M-1)/L-1}{T}\E[{\log\Bigg(1 + \frac{\rho \big|\hat{\bar{h}}_{k,l+1}\big|^2}{\rho\E[\big|\tilde{\bar{h}}_{k,l+1}\big|^2]+1}\Bigg)}] \\
&= \frac{(M-1)/L-1}{T}\E[{\log\bigg(1 + \frac{\rho(\rho+1)}{2\rho+1}\big|\hat{\bar{h}}_{k,l+1}\big|^2\bigg)}].
\end{align}

In the beamforming phase, the base station beamforms to the users according to the estimated channel with equal power. The transmitted signal is
\begin{equation}
{\bf X}_{[M+1:T]} = \sqrt{\frac{\rho}{K}}\sum_{k=1}^{K} \frac{\hat{\bf h}^*_k}{\|\hat{\bf h}_k\|} \ts_k^\T,
\end{equation}
where $\ts_k\in \CC^{T-M}$ contains i.i.d. $\Cc\Nc(0,1)$ data symbols for User~$k$. 
The received signal at User~$k$ is:
\begin{align}
(\ty^\T_k)_{[M+1:T]} &= {\bf h}_k^\T {\bf X}_{[M+1:T]} + (\tw^\T_k)_{[M+1:T]} \\
&= \sqrt{\frac{\rho}{K}} \sum_{l=1}^{K} \frac{{\bf h}_k^\T \hat{\bf h}_l}{\|\hat{\bf h}_l\|} \ts_l^\T  + (\tw^\T_k)_{[M+1:T]}. 
\end{align}

User~$k$ decodes $\ts_k$ and achieves the rate
\begin{equation}
R_{i} = \bigg(1 - \frac{M}{T}\bigg)\EE \bigg[\log\bigg(1 + \frac{\frac{\rho}{K}\|\hat{{\bf h}}_k\|^2}{\frac{\rho}{K}\E[\frac{|\tilde{\th}_k^\T \hat{\th}^*_k|^2}{\|\hat{\th}_k\|^2} + \sum_{l\ne k}\frac{|{\th}_k^\T \hat{\th}^*_l|^2}{\|\hat{\th}_l\|^2}] + 1}\bigg)\bigg].
\end{equation}

Finally, the achievable sum rate is:
\begin{equation}
R = \sum_{k = 1}^{K} R_i + \sum_{l = 1}^{L+1} \Delta R_l.
\end{equation}

Figure~\ref{fig_mmimo_10} shows the performance gain of the proposed scheme with respect to the conventional one under the following configuration: $K = 10$, $L = 9$, $M = 64$, $T = 128$.
\begin{figure}[!ht]
\centering
\includegraphics[width = \Figwidth]{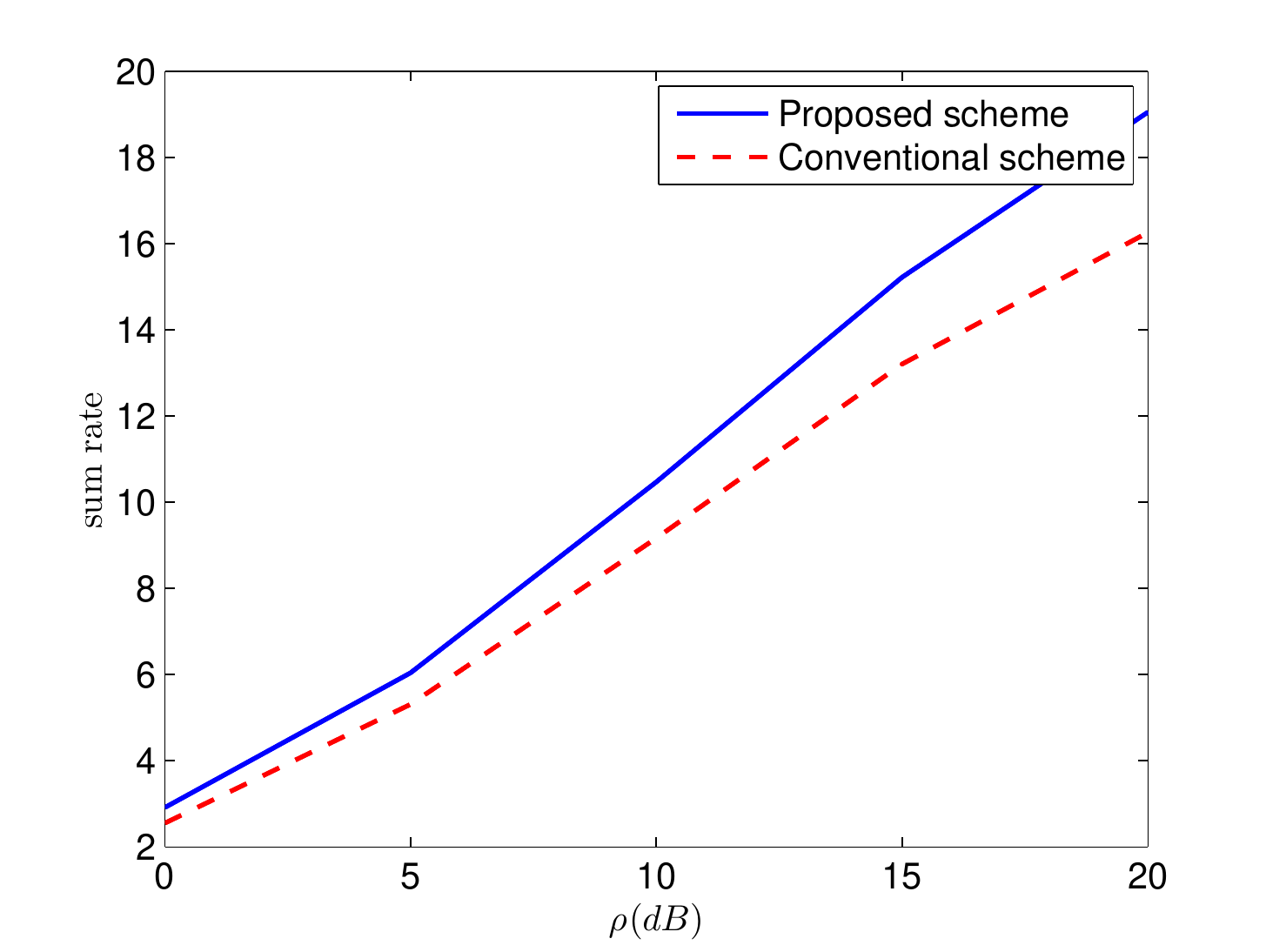}
\caption{The sum rate of the considered FDD massive MIMO system in on-off correlated fading with the proposed scheme in comparison with the conventional scheme for $K = 10, M = 64$, $T = 128$.}
\label{fig_mmimo_10}
\end{figure}

\subsection{Discussion: Correlation Diversity in Massive MIMO}
Our work focuses on gains that can be gleaned from the allocation of pilots. 
Broadly, our work has gains when the spatial correlation matrices between the users are dis-similar. The more the dissimilarity of the correlation matrices, the higher the gains provided by our technique. 
The metric for similarity in our work is the alignment of  the null spaces of the transmit correlation matrices corresponding to different users.  
A more detailed analysis of the gains depends naturally on antenna numbers as well as other factors; we omit a detailed listing of these cases in the interest of brevity.

In massive MIMO, when receivers have non-identical transmit correlations, 
designing the training sequences to match these non-identical channel correlation matrices can be challenging. 
Under this condition, Jiang et al. \cite{Jiang_2015} propose a scheme for massive MIMO 
in which the pilots are optimized according to a mutual information metric, and optimal length of the pilots is found by exhaustive search. When users have correlation matrices with different ranks, our method will have significant gains (in multiplexing gain) over \cite{Jiang_2015}.

\color{black}
\section{\textcolor{black}{Discussion and} Conclusion} \label{sec:conclusion}
This paper extends the scope of transmit correlation diversity to a broader set of conditions involving transmit correlation matrices with fully and partially overlapping eigenspaces. 
Furthermore, we present transmission schemes that harvest these generalized correlation diversity gains. We demonstrate the utility of both pre-beamforming and product superposition for correlation diversity. This arises from a careful decomposition of transmission spaces into several components. Along non-overlapping eigenspaces, simultaneous and non-interfering transmission is possible, as noted by earlier work. In the overlapping part, one may utilize the techniques employed in this paper. Careful design of this decomposition is necessary to allow the effective carving of the transmission signal space, allowing efficient operation of the proposed techniques. These ideas were developed in the context of a two-user system and were extended to multi-user systems. The application of these ideas in a massive MIMO system was explored.

\textcolor{black}{In the interest of completeness,  we mention imperfect or partial CSIT\cite{Caire:TIT2010,5571881} as another situation in which the transmitter knows something about the channel, but not everything. In transmit correlation diversity, training is concentrated on the part of the channel that remains unknown, while the imperfect/partial CSIT literature investigates how much of the channel knowledge can be abandoned in the interest of feedback efficiency, and what is the cost of this abandonment. In that sense, the two areas of investigation might be considered the dual of each other. Partial CSIT varies from one channel realization to the next, and is subject to fading speed and efficiency of feedback, while transmit correlation diversity reflects longer-term statistics that can be collected in the receiver over many realizations, and due to its slower variation, can be communicated with transmitter at higher precision. The methods and techniques used in addressing correlation diversity in this paper are largely distinct from the literature of imperfect/partial CSIT.}

\section{Acknowledgement}

The authors gratefully acknowledge Dr.\ Maxime Guillaud for his contribution to this investigation.

\appendices

\section{Proof of Theorem~\ref{thm:single_user}}
\label{proof:rate_single}

We prove by constructing pilot-based schemes that can achieve \eqref{eq:Rate_su_noRatTX}, \eqref{eq:Rate_su_RatTX_orthPilot}, and \eqref{eq:Rate_su_RatTX_optPilot}. 

\subsection{Case 1: Transmitter Ignores Correlation}	
The transmitter can ignore $\Rm$ and form the transmitted signal as if the channel is uncorrelated, but the performance still depends on correlation. Within each coherence block, the transmitter first sends an orthogonal pilot matrix $\Xm_\tau \in \CC^{M\times M}$ such that $\Xm_\tau\Xm^\H_\tau = M \Id_{M}$ during the first $M$ channel uses (this is optimal for uncorrelated fading~\cite[Sec.~III-A]{Hassibi_2003}), and then sends i.i.d.~$\Cc\Nc(0,1)$ data matrix $\rvMat{X}_\delta \in \CC^{M \times(T - M)}$ during the remaining $T-M$ channel uses. That is, 
\begin{align}
\rvMat{X} = \left[\sqrt{\frac{\rho_\tau}{M}}\Xm_\tau \ \sqrt{\frac{\rho_\delta}{M}} \rvMat{X}_\delta \right],
\end{align}
where $\rho_\tau$ and $\rho_\delta$ are the average power used for training and data phases, respectively, and satisfy the power constraint $\rho_\tau M + \rho_\delta(T-M) \le \rho T$.	

In the training phase, the receiver observes $\rvMat{Y}_\tau \defeq \rvMat{Y}_{[1:M]} = \sqrt{\frac{\rho_\tau}{M}}\rvMat{H}\Xm_\tau + \rvMat{W}_{[1:M]}$. 
Following Lemma~\ref{lem:MMSE-estimator}, it performs a linear MMSE channel estimator as
\begin{align}
\hat{\rvMat{H}} &= \sqrt{\frac{\rho_\tau}{M}} \rvMat{Y}_\tau \Big(\frac{\rho_\tau}{M}\Xm_\tau^\H \Rm \Xm_\tau + \Id_M\Big)^{-1} \Xm_\tau^\H \Rm. 
\end{align}
The estimate $ \hat{\rvMat{H}}$ and the estimation error $\tilde{\rvMat{H}} = \rvMat{H} - \hat{\rvMat{H}}$ have zero mean and row covariance 
\begin{align}
\frac{1}{N}\EE [\hat{\rvMat{H}}^\H \hat{\rvMat{H}}] &= \frac{\rho_\tau}{M} \Rm \Xm_\tau \Big(\frac{\rho_\tau}{M}\Xm_\tau^\H \Rm \Xm_\tau + \Id_M\Big)^{-1} \Xm_\tau^\H \Rm = \rho_\tau \Rm(\Id_M + \rho_\tau\Rm)^{-1}\Rm, \label{eq:tmp287} \\
\frac{1}{N}\EE [\tilde{\rvMat{H}}^\H \tilde{\rvMat{H}}] &= \Rm - \rho_\tau \Rm(\Id_M + \rho_\tau\Rm)^{-1}\Rm. 
\end{align}

In the data transmission phase, the received signal is 
\begin{align}  
\rvMat{Y}_\delta \defeq \rvMat{Y}_{[M+1:T]}= \sqrt{\frac{\rho_\delta}{M}}\rvMat{H} \rvMat{X}_\delta + \rvMat{W}_{[M+1:T]} = \sqrt{\frac{\rho_\delta}{M}} \hat{\rvMat{H}} \rvMat{X}_\delta + {\rvMat{W}}_\delta,
\end{align} 
where ${\rvMat{W}}_\delta \defeq \sqrt{\frac{\rho_\delta}{M}}\tilde{\rvMat{H}} \rvMat{X}_\delta + \rvMat{W}_{[M+1:T]}$ is the combined noise consisting of additive noise and channel estimation error. With MMSE estimator, $\rvMat{W}_\delta$ and $\rvMat{X}_\delta$ are uncorrelated because
\begin{align}
\EE [\rvMat{X}_\delta {\rvMat{W}}^\H_\delta |\Xm_\tau, \rvMat{Y}_\tau] &= \EE \Big[{\rvMat{X}_\delta  \Big(\sqrt{\frac{\rho_\delta}{M}}\rvMat{X}^\H_\delta \tilde{\rvMat{H}}^\H + \rvMat{W}^\H_\delta\Big) \big| \Xm_\tau, \rvMat{Y}_\tau}\Big] \\
&= \sqrt{\frac{\rho_\delta}{M}}\EE \Big[{\rvMat{X}_\delta  \rvMat{X}^\H_\delta (\rvMat{H} - \hat{\rvMat{H}}) \big| \Xm_\tau, \rvMat{Y}_\tau}\Big] \\
&= \mathbf{0},
\end{align}
since $\EE \big[{\rvMat{H} - \hat{\rvMat{H}} \big| \Xm_\tau, \rvMat{Y}_\tau}\big] = 0$.
From Lemma~\ref{lem:worst-case-noise}, a lower bound on the achievable rate is obtained by replacing ${\rvMat{W}}_\delta$ by i.i.d. Gaussian noise with the same variance
\begin{align}
\sigma^2_{\rvMat{W}_\delta} = \frac{1}{N(T-M)} \trace \big(\EE \big[{\rvMat{W}}^\H_\delta {\rvMat{W}}_\delta\big]\big) 
&= \frac{\rho_\delta}{M} \trace[\Rm - \rho_\tau \Rm(\Id_M + \rho_\tau\Rm)^{-1}\Rm] + 1 \\
&= \frac{\rho_\delta}{M} \trace[(\Sigmam^{-1} + \rho_\tau\Id_r)^{-1}] + 1.
\end{align}

Thus, the achievable rate is lower bounded by
\begin{align}
R &= \frac{T-M}{T}\EE \bigg[ \log\det\bigg(\Id_{N} + \frac{\rho_\delta}{M \sigma^2_{{\rvMat{W}}_\delta}} \hat{\rvMat{H}} \hat{\rvMat{H}}^\H \bigg) \bigg].
\end{align}
From \eqref{eq:tmp287}, $\hat{\rvMat{H}}$ has correlation matrix $\rho_\tau \Rm(\Id_M + \rho_\tau\Rm)^{-1}\Rm$. This shows \eqref{eq:Rate_su_noRatTX}.

\subsection{Case 2: Transmitter Exploits Correlation}	

By exploiting $\Rm$, the transmitter can project the signal onto the eigenspace of $\Rm$ and can also adapt the pilot symbols. The transmitter builds a precoder $\Vm \in \CC^{M \times r}$ with $r$ orthonormal columns such that $\Span[\Vm] = \Span[\Um]$. Let $\Phim = \Um^\H \Vm$. The transmitted signal is
\begin{align}
\rvMat{X} = \Vm \bigg[\sqrt{\frac{\rho_\tau}{r}}\Xm_\tau \ \sqrt{\frac{\rho_\delta}{r}} \rvMat{X}_\delta \bigg],
\end{align}
where $\Xm_\tau \in \CC^{r \times r}$ such that $\rank[\Xm_\tau] = r$ and $\trace[\Xm_\tau^\H \Xm_\tau] = r^2$ is the pilot matrix, and $\rvMat{X}_\delta \in \CC^{r \times(T - r)}$ is the data matrix containing $\Cc\Nc(0,1)$ entries. 
The average pilot and data powers satisfy $\rho_\tau r + \rho_\delta(T-r) \le \rho T$. 

The received signal during the training phase is then 
$\rvMat{Y}_\tau \defeq \rvMat{Y}_{[1:r]} = \sqrt{\frac{\rho_\tau}{r}} \rvMat{G} \Sigmam^{\frac12} \Phim \Xm_\tau  + \rvMat{W}_{[1:r]}.$
The equivalent channel ${\bf \Omega} \triangleq \rvMat{G} \Sigmam^{\frac12} \Phim$ has correlation matrix $\bar{\Rm} = \Phim^\H \Sigmam \Phim = \Vm^\H \Rm \Vm$. According to Lemma~\ref{lem:MMSE-estimator}, the MMSE channel estimate for the equivalent channel ${\bf \Omega}$ is given by
\begin{align}
\eqhat 
&= \sqrt{\frac{\rho_\tau}{r}} \rvMat{Y}_\tau \Big(\frac{\rho_\tau}{r}\Xm_\tau^\H \bar{\Rm} \Xm_\tau + \Id_r\Big)^{-1} \Xm_\tau^\H \bar{\Rm}.
\end{align}
The estimate $\eqhat$ and the estimation error $\eqtilde = \rvMat{G}\Sigmam^{\frac12}\Phim - \eqhat$ have  zero mean and row covariance 
\begin{align}
\frac{1}{N} \EE [\eqhat^\H \eqhat] 
&= \frac{\rho_\tau}{r} \bar{\Rm} \Xm_\tau \Big(\frac{\rho_\tau}{r}\Xm_\tau^\H \bar{\Rm} \Xm_\tau + \Id_r\Big)^{-1} \Xm_\tau^\H \bar{\Rm}, \label{eq:tmp360}\\
\frac{1}{N} \EE [\eqtilde^\H \eqtilde] &= \bar{\Rm} - \frac{\rho_\tau}{r} \bar{\Rm} \Xm_\tau \Big(\frac{\rho_\tau}{r}\Xm_\tau^\H \bar{\Rm} \Xm_\tau + \Id_r\Big)^{-1} \Xm_\tau^\H \bar{\Rm} = \big(\bar{\Rm}^{-1} + \frac{\rho_\tau}{r}\Xm_\tau \Xm_\tau^\H \big)^{-1}. 
\end{align}

In the data transmission phase, the received signal is 
\begin{align}  
\rvMat{Y}_\delta \defeq \rvMat{Y}_{[r+1:T]} = \sqrt{\frac{\rho_\delta}{r}}\rvMat{G} \Sigmam^{\frac12} \Phim \rvMat{X}_\delta + \rvMat{W}_{[r+1:T]} = \sqrt{\frac{\rho_\delta}{r}} \eqhat \rvMat{X}_\delta + {\rvMat{W}}_\delta,
\end{align} 
where ${\rvMat{W}}_\delta \defeq \sqrt{\frac{\rho_\delta}{r}}\eqtilde \rvMat{X}_\delta + \rvMat{W}_{[r+1:T]}$. From Lemma~\ref{lem:worst-case-noise}, a lower bound on the achievable rate is obtained by replacing ${\rvMat{W}}_\delta$ with i.i.d. Gaussian noise with the same variance
\begin{align}
\sigma^2_{{\rvMat{W}}_\delta} = \frac{1}{N(T-r)} \trace \big(\EE [{\rvMat{W}}^\H_\delta {\rvMat{W}}_\delta] \big)
&= \frac{\rho_\delta}{r}\trace \Big(\big(\bar{\Rm}^{-1} + \frac{\rho_\tau}{r}\Xm_\tau \Xm_\tau^\H \big)^{-1}\Big) + 1. 
\label{eq:tmp373}
\end{align}
The corresponding achievable rate lower bound is
\begin{align}
R &= \frac{T-r}{T}\EE \bigg[\log\det\bigg(\Id_{N} + \frac{\rho_\delta}{r \sigma^2_{{\rvMat{W}}_\delta}} \eqhat \eqhat^\H \bigg)\bigg],
  \end{align}
where the rows of $\eqhat$ obey $\Cc \Nc \big({0}^T,\bar{\Rm} - \Bm\big)$ with $\Bm \defeq \big(\bar{\Rm}^{-1} + \frac{\rho_\tau}{r}\Xm_\tau \Xm_\tau^\H \big)^{-1}$ and are independent with each other.

Taking $\Xm_\tau$ such that $\Xm_\tau \Xm_\tau^\H= r\Id_r$ (i.e., orthogonal pilots), we have $\Bm = \big(\bar{\Rm}^{-1} + \rho_\tau \Id_r\big)^{-1}$, and the achievable rate $R$ is given in \eqref{eq:Rate_su_RatTX_orthPilot}.

We can also optimize the pilot $\Xm_\tau$ so as to maximize $R$. The pilot matrix $\Xm_\tau$ affects the achievable rate bound primarily through the effective SNR
\begin{align}
\rho_{\rm eff} = \frac{\rho_\delta}{r \sigma^2_{{\rvMat{W}}_\delta}} \frac{1}{N} \EE \big(\trace \big[\eqhat^\H \eqhat\big]\big) =
\frac{\rho_\delta \trace (\bar{\Rm} - \Bm )}{\rho_\delta \trace (\Bm) + r},
\end{align}
which decreases with $\trace[\Bm]$. Therefore, to maximize $R$, we would like to minimize $\trace[\Bm]$. That is
\begin{align}
\min_{\trace (\Xm_\tau^\H\Xm_\tau) = r^2} \trace \Big(\big(\bar{\Rm}^{-1} + \frac{\rho_\tau}{r}\Xm_\tau \Xm_\tau^\H \big)^{-1} \Big).
\end{align}
Using Lagrange multiplier $\lambda$, we minimize 
\begin{align}
L(\Xm_\tau,\lambda) = \trace \Big(\big(\bar{\Rm}^{-1} + \frac{\rho_\tau}{r}\Xm_\tau \Xm_\tau^\H \big)^{-1}\Big) + \lambda\big(\trace (\Xm_\tau\Xm_\tau^\H) - r^2\big).
\end{align}
Solving $\frac{\partial L(\Xm_\tau,\lambda)}{\partial \Xm_\tau \Xm_\tau^\H} = 0$, we obtain the minimizer $
\Xm_\tau \Xm_\tau^\H = \sqrt{\frac{r}{\rho_\tau\lambda}}\Id_r   - \frac{r}{\rho_\tau} \bar{\Rm}^{-1}.
$
Using the constrain $\trace (\Xm_\tau^\H\Xm_\tau) = r^2$, we find that 
$
\frac{\rho_\tau}{r}\Xm_\tau \Xm_\tau^\H = \big(\rho_\tau + \frac{1}{r}\trace (\bar{\Rm}^{-1}) \big) \Id_{r} - \bar{\Rm}^{-1}.
$
With this, $\Bm = \big(\rho_\tau + \frac{1}{r} \trace (\bar{\Rm}^{-1}) \big)^{-1} \Id_{r}$, and the rate $R$ is given in \eqref{eq:Rate_su_RatTX_optPilot}. 
The effective SNR is now written as
\begin{align}
\rho_{\rm eff} = \frac{\rho_\delta}{\rho_\delta r \big(\rho_\tau + \frac{1}{r} \trace[\bar{\Rm}^{-1}] \big)^{-1} + r} \Big[\trace[\bar{\Rm}] - r \Big(\rho_\tau + \frac{1}{r} \trace[\bar{\Rm}^{-1}] \Big)^{-1}\Big].
\end{align}
Remark~\ref{remark:powerallocation} follows from an optimization of ($\rho_\tau, \rho_{d}$) as follows. Let $\rho_\tau r = (1-\alpha) \rho T$ and $\rho_\delta (T-r) = \alpha \rho T$ for $\alpha \in (0,1)$, we can derive that
\begin{align}
\rho_{\rm eff} = \frac{\rho T \trace[\bar{\Rm}]}{r(T-2r)} \frac{-\alpha^2 + a \alpha }{-\alpha + b},
\end{align}
where $a \defeq 1 + \frac{\trace (\bar{\Rm}^{-1})}{\rho T} - \frac{r^2}{\rho T \trace (\bar{\Rm})}$ and $b \defeq \frac{T-r}{T-2r}\Big(1 + \frac{\trace (\bar{\Rm}^{-1})}{\rho T}\Big)$. Noting that $T - 2r \ge 0$, we obtain the optimal value of $\alpha$ that maximizes $\rho_{\rm eff}$ as given in \eqref{eq:p2p_power_allocation}. This completes the proof.

\section{Proof of Theorem~\ref{thm:BC_rate_split}}
\label{proof:rate_part}
This achievable rate region is fully characterized by the mutual information $I(\rvMat{Y}_k;\rvMat{X}_k,\rvMat{X}_0)$, $I(\rvMat{Y}_k;\rvMat{X}_k \cond \rvMat{X}_0)$, and $I(\rvMat{Y}_k;\rvMat{X}_0 \cond \rvMat{X}_k)$, $k\in \{1,2\}$.
We cacluate the achievable rates for the following input distribution:
\begin{align}
\rvMat{X}_0 &= \left[\sqrt{\rho_{0\tau}}\Id_{s_0} \ \mathbf{0}_{s_0\times s_1} \ \sqrt{\frac{\rho_{0 \delta }}{s_0}}\rvMat{S}_{0}\right], \\
\rvMat{X}_1 &= \left[\mathbf{0}_{s_1 \times s_0} \ \sqrt{\rho_{1\tau}}\Id_{s_1} \ \sqrt{\frac{\rho_{1 \delta }}{s_1}}\rvMat{S}_{1}\right], \\
\rvMat{X}_2 &= \left[\mathbf{0}_{s_2 \times s_0} \ \sqrt{\rho_{2\tau}}\Id_{s_2} \ \sqrt{\frac{\rho_{2 \delta }}{s_2}} \rvMat{S}_{2 }\right],
\end{align}
where  $\rvMat{S}_{0} \in \CC^{s_0 \times (T-s_1-s_0)}$, $\rvMat{S}_{1} \in \CC^{s_1 \times (T-s_1-s_0)}$, and $\rvMat{S}_{2 } \in \CC^{s_2\times (T-s_2-s_0)}$ are data matrices containing independent $\Cc\Nc(0,1)$ symbols, for powers $\rho_{i\tau}, \rho_{\delta}$, $i\in \{0,1,2\}$, such that 
\begin{align}
\rho_{0\tau}s_0 + \rho_{0 \delta } (T-s_1-s_0) + \sum_{i=1}^{2} \big[\rho_{i\tau} s_i + \rho_{i\delta} (T-s_i-s_0)\big] = \rho T. \label{eq:add:power_constraint}
\end{align}

\label{sec:2user_rateSplitting_user1}
The received signal at User~$1$ is
\begin{align}
\rvMat{Y}_1 &= \rvMat{G}_1 \Sigmam^{\frac12}_1 \Phim_1 \begin{bmatrix}
\sqrt{\rho_{0\tau}}\Id_{s_0} & \mathbf{0} & \sqrt{\frac{\rho_{0 \delta }}{s_0}}\rvMat{S}_{0} \\
\mathbf{0}  & \sqrt{\rho_{1\tau}}\Id_{s_1} & \sqrt{\frac{\rho_{1 \delta }}{s_1}}\rvMat{S}_{1}
\end{bmatrix}
+ \rvMat{W}_1 \\
&= \Bigg[\underbrace{\rvMat{G}_1 \Sigmam^{\frac12}_1 \Phim_1 \Pm_{1\tau}^{\frac12}
+ \rvMat{W}_{1[1:s_1 + s_0]}}_{\rvMat{Y}_{1\tau}} 
\quad
\underbrace{\rvMat{G}_1 \Sigmam^{\frac12}_1 \Phim_1 \Pm_{1\delta}^{\frac12} \Bigg[\begin{matrix}
\rvMat{S}_{0} \\
\rvMat{S}_{1}
\end{matrix} \Bigg]
+ \rvMat{W}_{1[s_1+s_0+1;T]}}_{\rvMat{Y}_{1\delta}}\Bigg],
\end{align}
where $\Pm_{1\tau} \defeq \Bigg[\begin{matrix}
{\rho_{0\tau}}\Id_{s_0} & \mathbf{0} \\
\mathbf{0}  & {\rho_{1\tau}}\Id_{s_1}
\end{matrix}\Bigg]$ and $\Pm_{1\delta} \defeq \Bigg[\begin{matrix}
{\frac{\rho_{0 \delta }}{s_0}}\Id_{s_0} & \mathbf{0} \\
\mathbf{0} & {\frac{\rho_{1 \delta }}{s_1}}\Id_{s_1}
\end{matrix} \Bigg]$ are the power matrices for the pilot and data, respectively. 

The equivalent channel ${\bf \Omega}_1 \triangleq \rvMat{G}_1 \Sigmam^{\frac12}_1 \Phim_1$ has correlation matrix $\bar{\Rm}_1$.
Following Lemma~\ref{lem:MMSE-estimator}, User~$1$ performs a MMSE channel estimation based on $\rvMat{Y}_{1\tau}$ as
\begin{align}
\eqhatone = 
\rvMat{Y}_{1 \tau} \big(\Pm_{1\tau}^{\frac12}\bar{\Rm}_1\Pm_{1\tau}^{\frac12} + \Id_{s_1+s_0}\big)^{-1} \Pm_{1\tau}^{\frac12}\bar{\Rm}_1.
\end{align}
The estimate $\eqhatone$ and the estimation error $\eqtildeone = \rvMat{G}_1\Sigmam^{\frac12}_1 \Phim_1 - \eqhatone$ have  zero mean and row covariance 
\begin{align}
\frac{1}{N_1} \EE [\eqhatone^\H \eqhatone] &= \bar{\Rm}_1 \Pm_{1\tau}^{\frac12}\Big(\Pm_{1\tau}^{\frac12}\bar{\Rm}_1\Pm_{1\tau}^{\frac12} + \Id_{s_1+s_0}\Big)^{-1}\Pm_{1\tau}^{\frac12} \bar{\Rm}_1, \\
\frac{1}{N_1} \EE [\eqtildeone^\H \eqtildeone] &= \bar{\Rm}_1 - \bar{\Rm}_1 \Pm_{1\tau}^{\frac12}\Big(\Pm_{1\tau}^{\frac12}\bar{\Rm}_1\Pm_{1\tau}^{\frac12} + \Id_{s_1+s_0}\Big)^{-1}\Pm_{1\tau}^{\frac12} \bar{\Rm}_1 = \big(\bar{\Rm}_1^{-1} + \Pm_{1\tau} \big)^{-1}.
\end{align}

\underline{\it Lower bounding $I(\rvMat{Y}_1; \rvMat{X}_1, \rvMat{X}_0)$:}
The received signal during the data transmission phase can be written as
\begin{align}
\rvMat{Y}_{1\delta} = \hat{\rvMat{G}}_1 \Sigmam^{\frac12}_1 \Phim_1 \Pm_{1\delta}^{\frac12} \Bigg[\begin{matrix}
\rvMat{S}_{0} \\
\rvMat{S}_{1}
\end{matrix} \Bigg]
+ \rvMat{W}_{1\delta}, 
\end{align}
where ${\rvMat{W}}_{1\delta} \defeq \eqtildeone \Pm_{1\delta}^{\frac12}\Bigg[\begin{matrix}
\rvMat{S}_{0} \\
\rvMat{S}_{1}
\end{matrix} \Bigg]
+ \rvMat{W}_{1[s_1+s_0+1:T]}$ is the combined noise and residual interference due to channel estimation error.
Define $\bar{\bf \Omega}_1 \in \mathbb{C}^{N_1 \times (s_1 + s_0)}$ with independent rows obeying ${\Cc \Nc} \big({\bf 0}^T, \Pm_{1\tau}^{\frac12}(\Pm_{1\tau}^{\frac12}\bar{\Rm}_1\Pm_{1\tau}^{\frac12} + \Id_{s_1+s_0})^{-1}\Pm_{1\tau}^{\frac12}\big)$. By a similar analysis using Lemma~\ref{lem:worst-case-noise} as for \eqref{eq:Rate_su_RatTX_orthPilot} in Theorem~\ref{thm:single_user}, we have 
\begin{align}
\lefteqn{
I(\rvMat{Y}_1; \rvMat{X}_1, \rvMat{X}_0)
}\notag\\
&= I(\rvMat{Y}_{1\delta}; \rvMat{S}_{1}, \rvMat{S}_{0} \cond \rvMat{Y}_{1\tau}) +  \underbrace{I(\rvMat{Y}_{1\tau}; \rvMat{S}_{1}, \rvMat{S}_{0})}_{=0}\\
&= I(\rvMat{Y}_{1\delta}; \rvMat{S}_{1}, \rvMat{S}_{0} \cond \eqhatone) \\
&\ge \Big(T - s_1 -s_0 \Big) \EE\bigg[\log \det \bigg( \Id_{N_1} +\frac{1}{\trace \big(\big(\bar{\Rm}_1^{-1} + \Pm_{1\tau} \big)^{-1} \Pm_{1\delta}\big) + 1}  \eqhat_1 \Pm_{1\delta}  \eqhat_1^{\H} \bigg)\bigg] \\
&\ge \Big(T - s_1 -s_0 \Big) \EE\bigg[\log \det \bigg( \Id_{N_1} +\frac{1}{\trace \big(\big(\bar{\Rm}_1^{-1} + \Pm_{1\tau} \big)^{-1} \Pm_{1\delta}\big) + 1}  \bar{{\bf \Omega}}_1 \bar{\Rm}_1 \Pm_{1\delta} \bar{\Rm}_1^\H \bar{{\bf \Omega}}_1^{\H} \bigg)\bigg].
\label{eq:add:I(Y1;S1,S0)}
\end{align}

\underline{\it Lower bounding $I(\rvMat{Y}_1; \rvMat{X}_1 \cond \rvMat{X}_0)$:}
We rewrite $\rvMat{Y}_{1\delta}$ as
\begin{align} \label{eq:received_signal_user1}
\rvMat{Y}_{1\delta} &= \sqrt{\frac{\rho_{1 \delta }}{s_1}} \rvMat{G}_1 \Sigmam^{\frac12}_1 \Phim_{11} \rvMat{S}_{1}
+ \sqrt{\frac{\rho_{0 \delta }}{s_0}} \rvMat{G}_1 \Sigmam^{\frac12}_1 \Phim_{10} \rvMat{S}_{0}
+ \rvMat{W}_{1\delta}.
\end{align}
While decoding $\rvMat{S}_{1}$, the term $ \sqrt{\frac{\rho_{0 \delta }}{s_0}} \rvMat{G}_1 \Sigmam^{\frac12}_1 \Phim_{10}\rvMat{S}_{0}$ is an interference. Given the knowledge of $\rvMat{S}_{0}$ and the channel estimate $\eqhatone = \big[\eqhatonezero \ \eqhatoneone \big]$, where $\eqhatonezero$ and $\eqhatoneone$ are respectively the estimates of $\rvMat{G}_1 \Sigmam^{\frac12}_1 \Phim_{10}$ and $\rvMat{G}_1 \Sigmam^{\frac12}_1 \Phim_{11}$, the receiver can remove partly the interference to obtain
\begin{align}
\rvMat{Y}_{1\delta} - \sqrt{\frac{\rho_{0 \delta }}{s_0}} \eqhatonezero\rvMat{S}_{0}
&= \sqrt{\frac{\rho_{1 \delta }}{s_1}} \rvMat{G}_1 \Sigmam^{\frac12}_1 \Phim_{11} \rvMat{S}_{1}
+ \sqrt{\frac{\rho_{0 \delta }}{s_0}} \Big[\rvMat{G}_1 \Sigmam^{\frac12}_1 \Phim_{10} - \eqhatonezero \Big]\rvMat{S}_{0}
+ \rvMat{W}_{1[s_1+s_0+1:T]} \\
&= \sqrt{\frac{\rho_{1 \delta }}{s_1}} \eqhatoneone \rvMat{S}_{1} + {\rvMat{W}}_{1\delta}.
\end{align}
With a similar analysis using Lemma~\ref{lem:worst-case-noise} as for \eqref{eq:Rate_su_RatTX_orthPilot} in Theorem~\ref{thm:single_user}, 
\begin{align}
\lefteqn{
I(\rvMat{Y}_1; \rvMat{X}_1 \cond \rvMat{X}_0)
}\notag\\
&= I\big(\rvMat{Y}_{1\delta}; \rvMat{S}_{1} \big| \rvMat{S}_{0},\rvMat{Y}_{1\tau}\big) \\
&= I\big(\rvMat{Y}_{1\delta}; \rvMat{S}_{1} \big| \rvMat{S}_{0},\eqhatone\big) \\
&= I\big(\rvMat{Y}_{1\delta} - \sqrt{\frac{\rho_{0 \delta }}{s_0}} \eqhatonezero\rvMat{S}_{0}; \rvMat{S}_{1} \ \big| \ \rvMat{S}_{0},\eqhatone\big) \\
&= I\big(\sqrt{\frac{\rho_{1 \delta }}{s_1}} \eqhatoneone \rvMat{S}_{1} + {\rvMat{W}}_{1\delta}; \rvMat{S}_{1}  \ \big|\ \eqhatoneone\big) \\
&\ge \Big(T - s_1 -s_0\Big) \EE \bigg[\log \det \bigg( \Id_{N_1} + \frac{\rho_{1\delta}}{s_1 \big[\trace \big((\bar{\Rm}_1^{-1} + \Pm_{1\tau} )^{-1} \Pm_{1\delta}\big) + 1\big]} \bar{{\bf \Omega}}_1 \bar{\Rm}_{11} \bar{\Rm}_{11}^\H\bar{{\bf \Omega}}_1^{\H}\bigg)\bigg].
\label{eq:add:I(Y1;S1|S0)}
\end{align}

\underline{\it Lower bounding $I(\rvMat{Y}_1; \rvMat{X}_0 \cond \rvMat{X}_1)$:}
Given $\rvMat{S}_{1}$ and the channel estimate $\eqhatone = \big[\eqhatonezero \quad \eqhatoneone\big]$, the receiver can remove partly the interference in \eqref{eq:received_signal_user1} to obtain
\begin{align}
\rvMat{Y}_{1\delta} - \sqrt{\frac{\rho_{1 \delta }}{s_1}} \eqhatoneone\rvMat{S}_{1}
&= \sqrt{\frac{\rho_{0 \delta }}{s_0}} {\rvMat{G}}_1 \Sigmam^{\frac12}_1 \Phim_{10}\rvMat{S}_{0}
+ \sqrt{\frac{\rho_{1 \delta }}{s_1}} \Big[\rvMat{G}_1 \Sigmam^{\frac12}_1 \Phim_{11} - \eqhatoneone \rvMat{S}_{1}\Big]
+ \rvMat{W}_{1[s_1+s_0+1:T]} \\
&= \sqrt{\frac{\rho_{0 \delta }}{s_0}} \eqhatonezero\rvMat{S}_{0} + {\rvMat{W}}_{1\delta}.
\end{align}
Using reasoning similar to \eqref{eq:Rate_su_RatTX_orthPilot} in Theorem~\ref{thm:single_user}, 
\begin{align}
\lefteqn{
I(\rvMat{Y}_1; \rvMat{X}_0 \cond \rvMat{X}_1) 
}\notag \\
&= I\big(\rvMat{Y}_{1\delta}; \rvMat{S}_{0} \big| \rvMat{S}_{1},\rvMat{Y}_{1\tau}\big) \\
&= I\big(\rvMat{Y}_{1\delta}; \rvMat{S}_{0} \big| \rvMat{S}_{1},\eqhatone\big) \\
&= I\big(\rvMat{Y}_{1\delta} - \sqrt{\frac{\rho_{1 \delta }}{s_1}} \eqhatoneone\rvMat{S}_{1}; \rvMat{S}_{0} \ \big| \ \rvMat{S}_{1},\eqhatone\big) \\
&= I\big(\sqrt{\frac{\rho_{0 \delta }}{s_0}} \eqhatonezero\rvMat{S}_{0} + {\rvMat{W}}_{1\delta}; \rvMat{S}_{0} \ \big| \ \eqhatonezero \big) \\
&\ge \Big(T - s_1 - s_0\Big)\EE \bigg[\log \det \bigg( \Id_{N_1} + \frac{\rho_{0\delta}}{s_0\big[\trace \big( (\bar{\Rm}_1^{-1} + \Pm_{1\tau} \big)^{-1} \Pm_{1\delta}\big) + 1\big]} \bar{{\bf \Omega}}_1 \bar{\Rm}_{10} \bar{\Rm}_{10}^\H \bar{{\bf \Omega}}_1^{\H}\bigg) \bigg].
\label{eq:add:I(Y1;S0|S1)}
\end{align}

The received signal at User~$2$ is
\begin{align}
&\rvMat{Y}_2
= \rvMat{G}_2 \Sigmam^{\frac12}_2 \Phim_2 \begin{bmatrix}
\sqrt{\rho_{0\tau}}\Id_{s_0} & \mathbf{0} & \mathbf{0}_{s_0 \times (s_1-s_2)} & \sqrt{\frac{\rho_{0 \delta }}{s_0}}\rvMat{S}_{0} \\
\mathbf{0}  & \sqrt{\rho_{2\tau}}\Id_{s_2} & \sqrt{\frac{\rho_{2 \delta }}{s_2}}\rvMat{S}_{2  a} &\sqrt{\frac{\rho_{2 \delta }}{s_2}}\rvMat{S}_{2 b}
\end{bmatrix}
+ \rvMat{W}_2 \\
&= \Bigg[\!\underbrace{\rvMat{G}_2 \Sigmam^{\frac12}_2 \Phim_2 \Pm_{2\tau}^{\frac12}
\!+\! \rvMat{W}_{2[1:s_2+s_0]}}_{\rvMat{Y}_{2\tau }} 
\
\underbrace{\sqrt{\frac{\rho_{2 \delta }}{s_2}} \rvMat{G}_2 \Sigmam^{\frac12}_2 \Phim_{22} \rvMat{S}_{2 a}
\!+\! \rvMat{W}_{2[s_2+s_0+1:s_1+s_0]}}_{\rvMat{Y}_{2\delta a}} 
\
\underbrace{\rvMat{G}_2 \Sigmam^{\frac12}_2 \Phim_2 \Pm_{2\delta}^{\frac12} \Bigg[\begin{matrix}
\rvMat{S}_{0} \\
\rvMat{S}_{2 b}
\end{matrix} \Bigg]
\!+\! \rvMat{W}_{2[s_1+s_0+1:T]}}_{\rvMat{Y}_{2\delta b}} \!\Bigg],
\end{align}
where $\rvMat{S}_{2 a}$ and $\rvMat{S}_{2 b}$ are respectively the first $s_1-s_2$ columns and the remaining $T-s_1-s_0$ columns of $\rvMat{S}_{2}$; $\Pm_{2\tau} \defeq \Bigg[\begin{matrix}
{\rho_{0\tau}}\Id_{s_0} & \mathbf{0} \\
\mathbf{0}  & {\rho_{2\tau}}\Id_{s_2}
\end{matrix}\Bigg]$ and $\Pm_{2\delta} \defeq \Bigg[\begin{matrix}
{\frac{\rho_{0 \delta }}{s_0}}\Id_{s_0} & \mathbf{0} \\
\mathbf{0} & {\frac{\rho_{2 \delta }}{s_2}}\Id_{s_2}
\end{matrix} \Bigg]$ are the power matrices for the pilot and data, respectively. 
Following Lemma~\ref{lem:MMSE-estimator}, user 2 performs a MMSE channel estimation of ${\bf \Omega}_2 \triangleq \rvMat{G}_2\Sigmam^{\frac12}_2 \Phim_2 = [{\bf \Omega}_{20} \ {\bf \Omega}_{22}] = \big[\rvMat{G}_2\Sigmam^{\frac12}_2 \Phim_{20} \ \rvMat{G}_2\Sigmam^{\frac12}_2 \Phim_{22}\big]$ based on $\rvMat{Y}_{2\tau}$ as
\begin{align}
\eqhattwo = 
\rvMat{Y}_{2 \tau} \big(\Pm_{2\tau}^{\frac12}\bar{\Rm}_2\Pm_{2\tau}^{\frac12} + \Id_{s_2+s_0}\big)^{-1} \Pm_{2\tau}^{\frac12}\bar{\Rm}_2.
\end{align}
The estimate $\eqhattwo = \big[\eqhattwozero ~ \eqhattwotwo\big]$ and the estimation error $\eqtildetwo = \rvMat{G}_2\Sigmam^{\frac12}_2 \Phim_2 - \eqhattwo$ have  zero mean and row covariance 
\begin{align}
\frac{1}{N_2} \EE [\eqhattwo^\H \eqhattwo] &= \bar{\Rm}_2 \Pm_{2\tau}^{\frac12}\Big(\Pm_{2\tau}^{\frac12}\bar{\Rm}_2\Pm_{2\tau}^{\frac12} + \Id_{s_2+s_0}\Big)^{-1} \Pm_{2\tau}^{\frac12} \bar{\Rm}_2, \\
\frac{1}{N_2} \EE [\eqtildetwo^\H \eqtildetwo] &= \bar{\Rm}_2 - \bar{\Rm}_2 \Pm_{2\tau}^{\frac12}\Big(\Pm_{2\tau}^{\frac12}\bar{\Rm}_2\Pm_{2\tau}^{\frac12} + \Id_{s_2+s_0}\Big)^{-1} \Pm_{2\tau}^{\frac12} \bar{\Rm}_2 = \big(\bar{\Rm}_2^{-1} + \Pm_{2\tau} \big)^{-1}.
\end{align}

\underline{\it Lower bounding $I(\rvMat{Y}_2; \rvMat{X}_2, \rvMat{X}_0)$:}
Using the chain rule, 
\begin{align}
I(\rvMat{Y}_2; \rvMat{X}_2, \rvMat{X}_0) &= I(\rvMat{Y}_{2\tau},\rvMat{Y}_{2\delta a}, \rvMat{Y}_{2\delta b};\rvMat{S}_{0},\rvMat{S}_{2 a},\rvMat{S}_{2 b}) \\
&= I(\rvMat{Y}_{2\delta a}, \rvMat{Y}_{2\delta b};\rvMat{S}_{0},\rvMat{S}_{2 a},\rvMat{S}_{2 b} \cond \rvMat{Y}_{2\tau}) + \underbrace{I(\rvMat{Y}_{2\tau};\rvMat{S}_{0},\rvMat{S}_{2 a},\rvMat{S}_{2 b})}_{=0} \\
&= I(\rvMat{Y}_{2\delta a}, \rvMat{Y}_{2\delta b};\rvMat{S}_{0},\rvMat{S}_{2 a},\rvMat{S}_{2 b} \cond \eqhattwo) \\
&= I(\rvMat{Y}_{2\delta a};\rvMat{S}_{2 a} \cond \eqhattwo) + \underbrace{I(\rvMat{Y}_{2\delta a};\rvMat{S}_{0},\rvMat{S}_{2 b} \cond \rvMat{S}_{2 a}, \eqhattwo}_{=0}) \notag\\
&\quad+ 
\underbrace{I(\rvMat{Y}_{2\delta b};\rvMat{S}_{0},\rvMat{S}_{2 b} \cond \rvMat{Y}_{2\delta a}, \eqhattwo)}_{\ge I(\rvMat{Y}_{2\delta b};\rvMat{S}_{0},\rvMat{S}_{2 b} \cond \eqhattwo)} + \underbrace{I(\rvMat{Y}_{2\delta b};\rvMat{S}_{2 a} \cond \rvMat{S}_{0},\rvMat{S}_{2 b}, \rvMat{Y}_{2\delta a}, \eqhattwo)}_{= 0}\\ 
&\ge I(\rvMat{Y}_{2\delta a};\rvMat{S}_{2 a} \cond \eqhattwotwo) + I(\rvMat{Y}_{2\delta b};\rvMat{S}_{0},\rvMat{S}_{2 b} \cond \eqhattwo).
\label{eq:add:I(Y2;S2,S0)}
\end{align}
Define $\bar{\bf \Omega}_2 \in \mathbb{C}^{N_2 \times (s_2 + s_0)}$ with independent rows obeying ${\Cc \Nc} \big({\bf 0}^T,  \Pm_{2\tau}^{\frac12}(\Pm_{2\tau}^{\frac12}\bar{\Rm}_2\Pm_{2\tau}^{\frac12} + \Id_{s_2+s_0})^{-1}\Pm_{2\tau}^{\frac12}\big)$. Following analysis similar to \eqref{eq:Rate_su_RatTX_orthPilot} in Theorem~\ref{thm:single_user}, 
\begin{align}
&I(\rvMat{Y}_{2\delta a};\rvMat{S}_{2 a} \cond \eqhattwotwo) \notag \\
&\ge \Big(s_1 - s_2\Big)\EE \bigg[
\log \det \bigg(\Id_{N_2} + \frac{\rho_{2 \delta }}{\rho_{2 \delta}\trace \big(\bar{\Rm}_{22}^\H (\bar{\Rm}_2 +  \bar{\Rm}_2 \Pm_{2\tau} \bar{\Rm}_2)^{-1}\bar{\Rm}_{22} \big) +s_2 }\bar{{\bf \Omega}}_2 \bar{\Rm}_{22} \bar{\Rm}_{22}^\H \bar{{\bf \Omega}}_2^{\H} \bigg)\bigg]
\label{eq:add:I(Y2;S2,S0)_a}
\end{align}
and 
\begin{align}
&I(\rvMat{Y}_{2\delta b};\rvMat{S}_{0},\rvMat{S}_{2 b} \cond \eqhattwo) \notag \\
&\ge \Big(T - s_1 - s_0\Big) \EE \bigg[\log \det \bigg(\Id_{N_2} + \frac{1}{\trace \big((\bar{\Rm}_2^{-1} + \Pm_{2\tau})^{-1} \Pm_{2\delta}\big)+ 1} \bar{{\bf \Omega}}_2 \bar{\Rm}_{2}\Pm_{2\delta} \bar{\Rm}_{2}^\H \bar{{\bf \Omega}}_2^{\H}\bigg) \bigg].
\label{eq:add:I(Y2;S2,S0)_b}
\end{align}

\underline{\it Lower bounding $I(\rvMat{Y}_2; \rvMat{X}_2 \cond \rvMat{X}_0)$:}
We write $\rvMat{Y}_{2\delta} \defeq [\rvMat{Y}_{2\delta a} \ \rvMat{Y}_{2\delta b}]$ as
\begin{align}
\rvMat{Y}_{2\delta} &= \sqrt{\frac{\rho_{2 \delta }}{s_2}} \rvMat{G}_2 \Sigmam^{\frac12}_2 \Phim_{22} \rvMat{S}_{2}
+ \sqrt{\frac{\rho_{0 \delta }}{s_0}} \rvMat{G}_2 \Sigmam^{\frac12}_2 \Phim_{20} [\mathbf{0} \ \rvMat{S}_{0}]
+ \rvMat{W}_{2[s_2+s_0+1:T]}.
\end{align}
Similar to $I(\rvMat{Y}_1; \rvMat{X}_1 \cond \rvMat{X}_0)$, using interference cancellation and wort-case additive noise, 
\begin{align}
\lefteqn{
I(\rvMat{Y}_2; \rvMat{X}_2 \cond \rvMat{X}_0)
}\notag \\ 
&= I\big(\rvMat{Y}_{2\delta}; \rvMat{S}_{2} \big| \rvMat{S}_{0},\eqhattwo\big) \\
&= I\big(\rvMat{Y}_{2\delta} - \sqrt{\frac{\rho_{0 \delta }}{s_0}} \eqhattwozero[\mathbf{0} \ \rvMat{S}_{0}]; \rvMat{S}_{2} \big| \rvMat{S}_{0},\eqhattwo\big) \\
&\ge \Big(s_1 - s_2\Big) \EE \bigg[ \log \det \bigg(\Id_{N_2} + \frac{\rho_{2 \delta }}{\rho_{2 \delta}\trace \big(\bar{\Rm}_{22}^\H (\bar{\Rm}_2 + \bar{\Rm}_2 \Pm_{2\tau} \bar{\Rm}_2)^{-1}\bar{\Rm}_{22}\big)+s_2} \bar{{\bf \Omega}}_2 \bar{\Rm}_{22} \bar{\Rm}_{2} \bar{\Rm}_{2}^\H\bar{\Rm}_{22}^\H \bar{{\bf \Omega}}_2^{\H}\bigg)\bigg] \notag\\
&+ \Big(T - s_1 - s_0\Big) \EE \bigg[\log \det \bigg(\Id_{N_2} + \frac{\rho_{2 \delta}}{s_2\big[ \trace \big( (\bar{\Rm}_2^{-1} + \Pm_{2\tau})^{-1}\Pm_{2\delta} \big) + 1\big]} 
\bar{{\bf \Omega}}_2 \bar{\Rm}_{22} \bar{\Rm}_{2} \bar{\Rm}_{2}^\H \bar{\Rm}_{22}^\H \bar{{\bf \Omega}}_2^{\H} \bigg)\bigg].
\label{eq:add:I(Y2;S2|S0)}
\end{align}

\underline{\it Lower bounding $I(\rvMat{Y}_2; \rvMat{X}_0 \cond \rvMat{X}_2)$:}
Again, using interference cancellation and a similar analysis as for \eqref{eq:Rate_su_RatTX_orthPilot} in Theorem~\ref{thm:single_user}, 
\begin{align}
\lefteqn{
I(\rvMat{Y}_2; \rvMat{X}_0 \cond \rvMat{X}_2)
}\notag \\ 
&\ge I\big(\rvMat{Y}_{2\delta b}; \rvMat{S}_{0} \big| \rvMat{S}_{2b},\eqhattwo\big) \\
&= I\big(\rvMat{Y}_{2\delta b} - \sqrt{\frac{\rho_{2 \delta }}{s_2}} \eqhattwotwo \rvMat{S}_{2b}; \rvMat{S}_{0} \ \big|\ \rvMat{S}_{2b},\eqhattwo\big) \\
&\ge \Big(T - s_1 - s_0\Big) \EE \bigg[\log \det \bigg(\Id_{N_2} + \frac{\rho_{0 \delta }}{s_0\big[\trace \big( (\bar{\Rm}_2^{-1} +  \Pm_{2\tau} )^{-1} \Pm_{2\delta}\big) + 1\big]}  \bar{{\bf \Omega}}_2\bar{\Rm}_{20} \bar{\Rm}_{2} \bar{\Rm}_{2}^\H \bar{\Rm}_{20}^\H\bar{{\bf \Omega}}_2^{\H}  \bigg) \bigg].
\label{eq:add:I(Y2;S0|S2)}
\end{align}

Substituting \eqref{eq:add:I(Y2;S2,S0)_a} and \eqref{eq:add:I(Y2;S2,S0)_b} into \eqref{eq:add:I(Y2;S2,S0)}, then substituting \eqref{eq:add:I(Y1;S1,S0)}, \eqref{eq:add:I(Y1;S1|S0)}, \eqref{eq:add:I(Y1;S0|S1)}, \eqref{eq:add:I(Y2;S2,S0)}, \eqref{eq:add:I(Y2;S2|S0)}, and \eqref{eq:add:I(Y2;S0|S2)} into \eqref{eq:add:R1_bound}-\eqref{eq:add:sum_rate_bound}, and taking the convex hull over all possible power allocation satisfying \eqref{eq:add:power_constraint} and all feasible values of $s_0, s_1,s_2$, an achievable rate region is found with rate splitting for the broadcast channel. This concludes the proof of Theorem~\ref{thm:BC_rate_split}.

\section{Proof of Theorem~\ref{thm:BC_rate_product}}
\label{proof:rate_product}
Under product superposition, the input to the channel is constructed as follows:
\begin{align}
\rvMat{X} = [\Vm_0 \ \Vm_2]\rvMat{X}_1 \rvMat{X}_2,
\end{align}
with 
\begin{align}
\rvMat{X}_1 &= \begin{bmatrix}
\sqrt{\nu_{1\tau}}\Id_{s_0} & \sqrt{\frac{\nu_{1\delta}}{s_0}} \rvMat{S}_1 \\
\mathbf{0} & \sqrt{\nu_{1a}}\Id_{s_2}
\end{bmatrix}, \\
\rvMat{X}_2 &= 
\bigg[\sqrt{\rho_{2\tau}}\Id_{s_2+s_0} \quad \sqrt{\frac{\rho_{2\delta}}{s_2+s_0}}\rvMat{S}_2\bigg],
\end{align}
where $\rvMat{S}_1 \in \CC^{s_0 \times s_2}$ and $\rvMat{S}_2 \in \CC^{(s_2 + s_0) \times (T-s_2-s_0)}$ are the data matrices of User~$1$ and User~$2$ respectively, both contain i.i.d. $\Cc\Nc(0,1)$ symbols. As in earlier developments, integers $s_0,s_1,s_2$ are designed to allocate transmit dimensions to the components of product superposition, and take values in the range $s_0 \leq r_0$ and $s_2 \leq r_2 - r_0$.

The power constraint $\E[{\trace[\rvMat{X}^\H \rvMat{X}]}] \le \rho T$ translates to
\begin{align}
(s_0 \nu_{1\tau} + s_2(\nu_{1\delta}+\nu_{1a}))\Big(\rho_{2\tau} + \frac{T-s_2-s_0}{s_2+s_0}\rho_{2\delta} \Big) \le \rho T. \label{eq:prod:power_constraint}
\end{align}

In the first $s_2+s_0$ channel uses, User~$1$ receives
\begin{align}
\rvMat{Y}_{1[1:s_2+s_0]} &= \sqrt{\rho_{2\tau}}\rvMat{G}_1 \Sigmam_1^{\frac12} \Phim_{10} \bigg[\sqrt{\nu_{1\tau}}\Id_{s_0} \ \sqrt{\frac{\nu_{1\delta}}{s_0}} \rvMat{S}_1\bigg] + \rvMat{W}_{1[1:s_2+s_0]} \\
&=\Bigg[\underbrace{\sqrt{\nu_{1\tau}\rho_{2\tau}}\rvMat{G}_1 \Sigmam_1^{\frac12} \Phim_{10} + \rvMat{W}_{1[1:s_0]}}_{\rvMat{Y}_{1\tau}}
\quad 
\underbrace{\sqrt{\frac{\nu_{1\delta}\rho_{2\tau}}{s_0}}\rvMat{G}_1 \Sigmam_1^{\frac12} \Phim_{10} \rvMat{S}_1 + \rvMat{W}_{1[s_0+1:s_2+s_0]}}_{\rvMat{Y}_{1\delta}}
\Bigg].
\end{align}
Following Lemma~\ref{lem:MMSE-estimator}, User~$1$ estimates the equivalent channel $\rvMat{G}_1 \Sigmam_1^{\frac12} \Phim_{10}$ using a MMSE estimator based on $\rvMat{Y}_{1\tau}$ as
\begin{align}
\eqhatonezero = \sqrt{\nu_{1\tau}\rho_{2\tau}} \rvMat{Y}_{1\tau} \Big(\nu_{1\tau}\rho_{2\tau} \breve{\Rm}_{10} + \Id_{s_0}\Big)^{-1} \breve{\Rm}_{10}.
\end{align}
The estimate $ \eqhatonezero$ and the estimation error $\eqtildeonezero = {\rvMat{G}_1 \Sigmam_1^{\frac12} \Phim_{10}} - \eqhatonezero$ have  zero mean and row covariance 
\begin{align}
\frac{1}{N_1} \EE [\eqhatonezero^\H \eqhatonezero] &= \nu_{1\tau}\rho_{2\tau} \breve{\Rm}_{10} \Big(\nu_{1\tau}\rho_{2\tau} \breve{\Rm}_{10} + \Id_{s_0}\Big)^{-1} \breve{\Rm}_{10}, \\
\frac{1}{N_1} \EE [\eqtildeonezero^\H \eqtildeonezero] &= \breve{\Rm}_{10} - \nu_{1\tau}\rho_{2\tau} \breve{\Rm}_{10} \Big(\nu_{1\tau}\rho_{2\tau} \breve{\Rm}_{10} + \Id_{s_0}\Big)^{-1} \breve{\Rm}_{10}  \Big(\breve{\Rm}_{10}^{-1} + \nu_{1\tau}\rho_{2\tau} \Id_{s_0} \Big)^{-1}.
\end{align}
Using data processing inequality, 
\begin{align}
I(\rvMat{Y}_1;\rvMat{X}_1) &\ge I(\rvMat{Y}_{1[1:s_2+s_0]};\rvMat{X}_1) = I(\rvMat{Y}_{1\delta};\rvMat{S}_1 \cond \rvMat{Y}_{1\tau}) = I(\rvMat{Y}_{1\delta};\rvMat{S}_1 \cond \eqhatonezero).
\end{align}
Then, using the worst-case noise argument and Lemma~\ref{lem:worst-case-noise}, the following lower bound on $I(\rvMat{Y}_{1\delta};\rvMat{S}_1 \cond \eqhatonezero)$, is established, giving an achievable rate for User~$1$:
\begin{align}
R_1 
& = \frac{s_2}{T} \EE \bigg[ \log \det  \bigg(\Id_{N_1} + \frac{\nu_{1\delta} \rho_{2\tau}}{s_0 + \nu_{1\delta} \rho_{2\tau}\trace \big( (\breve{\Rm}_{10}^{-1} + \nu_{1\tau}\rho_{2\tau} \Id_{s_0} )^{-1}\big)} \eqhatonezero \eqhatonezero^\H \bigg)\bigg]. 
\label{eq:prod:R1}
\end{align}

The received signal at User~$2$ is 
\begin{align}
\rvMat{Y}_2 &= \rvMat{G}_2 \Sigmam_2^{\frac12} \Phim_2 \rvMat{X}_1  \bigg[\sqrt{\rho_{2\tau}}\Id_{s_2+s_0} \quad \sqrt{\frac{\rho_{2\delta}}{s_2+s_0}}\rvMat{S}_2\bigg] + \rvMat{W}_2 \\
&= \Bigg[\underbrace{\sqrt{\rho_{2\tau}}\rvMat{G}_{2e}  + \rvMat{W}_{2[1:s_2+s_0]}}_{\rvMat{Y}_{2\tau}}
\quad
\underbrace{\sqrt{\frac{\rho_{2\delta}}{s_2+s_0}}\rvMat{G}_{2e}\rvMat{S}_2  + \rvMat{W}_{2[s_2+s_0+1:T]}}_{\rvMat{Y}_{2\delta}}
\Bigg],
\end{align}
where $\rvMat{G}_{2e} \defeq \rvMat{G}_2 \Sigmam_2^{\frac12} \Phim_2 \rvMat{X}_1$ is the equivalent channel with the correlation matrix 
\begin{align}
\Rm_{2e} \defeq \frac{1}{N_2} \E[\rvMat{G}_{2e}^\H \rvMat{G}_{2e}] = \begin{bmatrix}
\nu_{1\tau} \breve{\Rm}_{20} & \sqrt{\nu_{1\tau} \nu_{1a}} \Phim_{20}^\H \Sigmam_2 \Phim_{22} \\
\sqrt{\nu_{1\tau} \nu_{1a}} \Phim_{22}^\H \Sigmam_2 \Phim_{20} & \frac{\nu_{1\delta}}{s_0} \trace[\breve{\Rm}_{20}] \Id_{s_2} + \nu_{1a} \breve{\Rm}_{22}
\end{bmatrix}.
\end{align}
Following Lemma~\ref{lem:MMSE-estimator}, User~$2$ estimates the equivalent channel $\rvMat{G}_{2e}$ using a MMSE estimator based on $\rvMat{Y}_{2\tau}$ as
\begin{align}
\hat{\rvMat{G}}_{2e} = \sqrt{\rho_{2\tau}} \rvMat{Y}_{2\tau} \big(\rho_{2\tau} \Rm_{2e} + \Id_{s_2+s_0}\big)^{-1} \Rm_{2e}. \label{eq:prod:MMSE_estimator}
\end{align}
The estimate $\hat{\rvMat{G}}_{2e}$ and the estimation error $\tilde{\rvMat{G}}_{2e} = {\rvMat{G}}_{2e} - \hat{\rvMat{G}}_{2e}$ have zero mean and row covariance 
\begin{align}
\frac{1}{N_2}\EE [\hat{\rvMat{G}}_{2e}^\H \hat{\rvMat{G}}_{2e}] &= \rho_{2\tau} \Rm_{2e} \big(\rho_{2\tau}\Rm_{2e} + \Id_{s_2+s_0}\big)^{-1} \Rm_{2e}, \\
\frac{1}{N_2}\EE [\tilde{\rvMat{G}}_{2e}^\H \tilde{\rvMat{G}}_{2e}] &= \Rm_{2e} - \rho_{2\tau} \Rm_{2e} \big(\rho_{2\tau}\Rm_{2e} + \Id_{s_2+s_0}\big)^{-1} \Rm_{2e} = \big(\Rm_{2e}^{-1} + \rho_{2\tau} \Id_{s_2+s_0} \big)^{-1}.
\end{align}
Using the worst-case noise argument and Lemma~\ref{lem:worst-case-noise},  the following achievable rate for User~$2$ is established:
\begin{align}
R_2 & = \Big(1-\frac{s_2+s_0}{T}\Big) \EE \bigg[ \log\det \bigg( {\Id_{N_2} + \frac{\rho_{2\delta}}{s_2 + s_0 + \rho_{2\delta}\trace \big( (\Rm_{2e}^{-1} + \rho_{2\tau} \Id_{s_2 + s_0} )^{-1}\big)} \hat{\rvMat{G}}_{2e} \hat{\rvMat{G}}_{2e}^\H } \bigg) \bigg]
\label{eq:prod:R2}
\end{align}
where the distribution of $\hat{\rvMat{G}}_{2e}$ is imposed by \eqref{eq:prod:MMSE_estimator}. 

From \eqref{eq:prod:R1} and \eqref{eq:prod:R2}, the rate pair $(R_1,R_2)$ is achievable. By swapping the users' role, another achievable rate pair is obtained. The overall achievable rate region is the convex hull of these pairs over all possible power allocations satisfying \eqref{eq:prod:power_constraint} and all feasible values of $s_0, s_1,s_2$. This concludes the proof of Theorem~\ref{thm:BC_rate_product}.

\section{Proof of Theorem~\ref{thm:BC_rate_hybrid}}
\label{proof:rate_hybrid}
The transmitted signal is 
\begin{align}
\rvMat{X} = [\Vm_0\ \Vm_1]\
\rvMat{X}'_{2} \rvMat{X}_1
+ \Vm_2 \rvMat{X}_2 
\end{align}
with
\begin{align}
\rvMat{X}_1 &= \bigg[\sqrt{\rho_{1\tau}}\Id_{s_1+s_0} \  \sqrt{\frac{\rho_{1\delta}}{s_1+s_0}}\rvMat{S}_{1}\bigg] \in \CC^{(s_1 + s_0) \times T}, \\
\rvMat{X}_2 &= \bigg[\mathbf{0}_{s_2 \times s_0} \  \sqrt{\rho_{2\tau}}\Id_{s_2} \  \sqrt{\frac{\rho_{2\delta}}{s_2}}\rvMat{S}_{2}\bigg] \in \CC^{s_2 \times T}, \\
\rvMat{X}'_{2} &= \begin{bmatrix}
\sqrt{\nu_{2\tau}}\Id_{s_0} & \Big[\mathbf{0}_{s_0 \times s_2} \ \sqrt{\frac{\nu_{2\delta}}{s_0}}\rvMat{S}'_{2}\Big] \\ \mathbf{0} & \sqrt{\nu_{2a}}\Id_{s_1}
\end{bmatrix} \in \CC^{(s_1+s_0) \times (s_1+s_0)},
\end{align}
where $\rvMat{S}_1 \in \CC^{(s_1+s_0) \times (T-s_1-s_0)}$, $\rvMat{S}_2 \in \CC^{s_2\times (T-s_2-s_0)}$, and $\rvMat{S}'_{2}\in \CC^{s_0\times(s_1-s_2)}$ are data matrices containing $\Cc\Nc(0,1)$ entries. 
The power constraint $\E[{\trace[\rvMat{X}^\H \rvMat{X}]}] \le \rho T$ translates to
\begin{align}
\big(s_0\nu_{2\tau} + s_1\nu_{2a} + (s_1-s_2) \nu_{2\delta} \big) \Big(\rho_{1\tau} + \frac{T-s_1-s_0}{s_1+s_0} \rho_{1\delta}\Big) + s_2\rho_{2\tau} + (T-s_2-s_0)\rho_{2\delta}\le \rho T. \label{eq:hybrid:power_constraint}
\end{align}

We begin by analyzing the rate of User~$1$. The received signal at User~$1$ is 
\begin{align}
\rvMat{Y}_1 
&= \rvMat{G}_1 \Sigmam^{\frac12}_1 \Phim_{1} \rvMat{X}'_{2} \bigg[\sqrt{\rho_{1\tau}}\Id_{s_1+s_0} \  \sqrt{\frac{\rho_{1\delta}}{s_1+s_0}}\rvMat{S}_{1}\bigg] + \rvMat{W}_1 \\
&= \Bigg[\underbrace{\sqrt{\rho_{1\tau}}\rvMat{G}_{1e} + \rvMat{W}_{1[1:s_1+s_0]}}_{\rvMat{Y}_{1\tau}}
\quad
\underbrace{\sqrt{\frac{\rho_{1\delta}}{s_1+s_0}} \rvMat{G}_{1e} \rvMat{S}_{1} + \rvMat{W}_{1[s_1+s_0+1:T]}}_{\rvMat{Y}_{1\delta}}\Bigg],
\end{align}
where ${\rvMat{G}_{1e}} \defeq \rvMat{G}_1 \Sigmam^{\frac12}_1 \Phim_{1} \rvMat{X}'_{2}$ is the equivalent channel with correlation matrix 
\begin{align}
\Rm_{1e} &\defeq \frac{1}{N_1} \E[\rvMat{G}_{1e}^\H \rvMat{G}_{1e}] 
= \begin{bmatrix}
\nu_{2\tau} \breve{\Rm}_{10} & \sqrt{\nu_{2\tau}\nu_{2a}} \Phim_{10}^\H \Sigmam_1 \Phim_{11}\\
\sqrt{\nu_{2\tau}\nu_{2a}} \Phim_{11}^\H \Sigmam_1 \Phim_{10} & \Bigg[\begin{matrix}
\mathbf{0} & \mathbf{0} \\
\mathbf{0} & \frac{\nu_{2\delta}}{s_0} \trace[\breve{\Rm}_{10}] \Id_{s_1-s_2}
\end{matrix}\Bigg]
+ \nu_{2a} \breve{\Rm}_{22}
\end{bmatrix}.
\end{align}
Following Lemma~\ref{lem:MMSE-estimator}, User~$1$ estimates the equivalent channel $\rvMat{G}_{1e}$ using a MMSE estimator based on $\rvMat{Y}_{1\tau}$ as
\begin{align}
\hat{\rvMat{G}}_{1e} = \sqrt{\rho_{1\tau}} \rvMat{Y}_{1\tau} \big(\rho_{1\tau} \Rm_{1e} + \Id_{s_1+s_0}\big)^{-1} \Rm_{1e}. \label{eq:hybrid:MMSE_estimate}
\end{align}
The estimate $\hat{\rvMat{G}}_{1e}$ and the estimation error $\tilde{\rvMat{G}}_{1e} = {\rvMat{G}}_{1e} - \hat{\rvMat{G}}_{1e}$ have zero mean and row covariance 
\begin{align}
\frac{1}{N_1}\EE [\hat{\rvMat{G}}_{1e}^\H \hat{\rvMat{G}}_{1e}] &= \rho_{1\tau} \Rm_{1e} \big(\rho_{1\tau}\Rm_{1e} + \Id_{s_1+s_0}\big)^{-1} \Rm_{1e}, \\
\frac{1}{N_1}\EE [\tilde{\rvMat{G}}_{1e}^\H \tilde{\rvMat{G}}_{1e}] &= \Rm_{1e} - \rho_{1\tau} \Rm_{1e} \big(\rho_{1\tau}\Rm_{1e} + \Id_{s_1+s_0}\big)^{-1} \Rm_{1e} = \big(\Rm_{1e}^{-1} + \rho_{1\tau} \Id_{s_1+s_0} \big)^{-1}.
\end{align}
Using the worst-case noise argument and Lemma~\ref{lem:worst-case-noise} as before, the following achievable rate for User~$1$ is obtained:
\begin{align}
R_1 = \bigg(1-\frac{s_1+s_0}{T}\bigg) \EE \bigg[\log\det\bigg(\Id_{N_1} + \frac{\rho_{1\delta}}{s_1 + s_0 + \rho_{1\delta}\trace \big( (\Rm_{1e}^{-1} + \rho_{1\tau} \Id_{s_1+s_0} \big)^{-1}\big)} \hat{\rvMat{G}}_{1e} \hat{\rvMat{G}}_{1e}^\H \bigg)\bigg],
\label{eq:tmp919}
\end{align}
where the distribution of $\hat{\rvMat{G}}_{1e}$ is imposed by \eqref{eq:hybrid:MMSE_estimate}. 

Now, we turn to analyzing the achievable rate for User~$2$.  The received signal at User~$2$ can be written as
\begin{align}
\rvMat{Y}_2 &= \rvMat{G}_2 \Sigmam^{\frac12}_2 \Phim_2 \begin{bmatrix}
\sqrt{\nu_{2\tau}\rho_{1\tau}}\Id_{s_0} & \mathbf{0} & \Big[\sqrt{\frac{\nu_{2\delta}\rho_{1\tau}}{s_0}}\rvMat{S}'_{2} \ \rvMat{A}\Big]\\
\mathbf{0} & \sqrt{\rho_{2\tau}}\Id_{s_2} & \sqrt{\frac{\rho_{2\delta}}{s_2}}\rvMat{S}_{2}\\
\end{bmatrix}
+ \rvMat{W}_2 \\
&= \Big[\rvMat{Y}_{2\tau} \ \underbrace{\rvMat{Y}_{2\delta a} \ \rvMat{Y}_{2\delta b}}_{\rvMat{Y}_{2\delta}}\Big],
\end{align}
where $\rvMat{A} \defeq \Big[\sqrt{\nu_{2\tau}}\Id_{s_0} \ \mathbf{0} \ \sqrt{\frac{\nu_{2\delta}}{s_0}}\rvMat{S}'_{2}\Big]\sqrt{\frac{\rho_{1\delta}}{s_1+s_0}}\rvMat{S}_{1}$ and
\begin{align}
\rvMat{Y}_{2\tau} &\defeq \rvMat{G}_2 \Sigmam^{\frac12}_2 \Phim_2 \Pm_{2\tau}^{\frac12} + \rvMat{W}_{2[1:s_2+s_0]}, \\ 
\rvMat{Y}_{2\delta a} &\defeq \rvMat{G}_2 \Sigmam^{\frac12}_2 \Phim_2 \begin{bmatrix}
\sqrt{\frac{\nu_{2\delta}\rho_{1\tau}}{s_0}}\rvMat{S}'_{2} \\
\sqrt{\frac{\rho_{2\delta}}{s_2}}\rvMat{S}_{2[1:s_1-s_2]} 
\end{bmatrix}
+ \rvMat{W}_{2[s_2+s_0+1:s_1+s_0]}, \\
\rvMat{Y}_{2\delta b} &\defeq \rvMat{G}_2 \Sigmam^{\frac12}_2 \Phim_2 \begin{bmatrix}
\rvMat{A} \\
\sqrt{\frac{\rho_{2\delta}}{s_2}}\rvMat{S}_{2[s_1-s_2+1:T-s_2]}
\end{bmatrix}
+\rvMat{W}_{2[s_1+s_0+1:T]},
\end{align}
where $\Pm_{2\tau} \defeq \Bigg[\begin{matrix}
\nu_{2\tau}\rho_{1\tau}\Id_{s_0} & \mathbf{0} \\
\mathbf{0} & \rho_{2\tau}\Id_{s_2}
\end{matrix}\Bigg]$.
The rate that User~$2$ can achieve is $\frac{1}{T}I(\rvMat{Y}_2; \rvMat{S}'_2, \rvMat{S}_2)$ bits/channel use with
\begin{align}
I(\rvMat{Y}_2; \rvMat{S}'_2, \rvMat{S}_2) &= I(\rvMat{Y}_{2\tau}, \rvMat{Y}_{2\delta}; \rvMat{S}'_2, \rvMat{S}_2) \\
&= \underbrace{I(\rvMat{Y}_{2\tau}; \rvMat{S}'_2, \rvMat{S}_2)}_{=0} + I(\rvMat{Y}_{2\delta}; \rvMat{S}'_2, \rvMat{S}_2 \cond \rvMat{Y}_{2\tau}) \\
&= I(\rvMat{Y}_{2\delta}; \rvMat{S}'_2, \rvMat{S}_2, \rvMat{A} \cond \rvMat{Y}_{2\tau}) - I(\rvMat{Y}_{2\delta}; \rvMat{A} \cond \rvMat{Y}_{2\tau}, \rvMat{S}'_2, \rvMat{S}_2), \label{eq:tmp951}
\end{align}
where the second and third equalities follow from the chain rule. 

Define $\bar{\bf \Omega}_2 \in \mathbb{C}^{N_2 \times (s_2 + s_0)}$ with independent rows obeying ${\Cc \Nc} \big({\bf 0}^T,  \bar{\Rm}_2^\H \big(\bar{\Rm}_2 + \Pm_{2\tau}^{-1}\big)^{-1} \bar{\Rm}_2\big)$ and $\bar{\bf \Omega}_{20} \in \mathbb{C}^{N_2 \times s_0}$ with independent rows obeying ${\Cc \Nc} \big({\bf 0}^T, \breve{\Rm}_{20}\big)$. For $I(\rvMat{Y}_{2\delta}; \rvMat{S}'_2, \rvMat{S}_2, \rvMat{A} \cond \rvMat{Y}_{2\tau})$, using the worst-case noise argument and Lemma~\ref{lem:worst-case-noise} as before, we have the bound
\begin{align}
\lefteqn{I(\rvMat{Y}_{2\delta}; \rvMat{S}'_2, \rvMat{S}_{2}, \rvMat{A} \big| \rvMat{Y}_{2\tau})
}\notag \\
&\ge \big(s_1 - s_2\big) \EE \bigg[ \log \det \bigg(\Id_{N_2} + \frac{1}{\trace \big( (\bar{\Rm}_2^{-1} + \Pm_{2\tau})^{-1} \Pm_{2\delta a}\big) + 1} \bar{\bf \Omega}_2 \Pm_{2\delta a} \bar{\bf \Omega}_2^\H\bigg) \bigg] \\
&+ \big(T - s_1 - s_0\big) \EE \bigg[ \log \det \bigg(\Id_{N_2} + \frac{1}{\trace \big( (\bar{\Rm}_2^{-1} + \Pm_{2\tau})^{-1} \Pm_{2\delta b}\big) + 1} \bar{\bf \Omega}_2 \Pm_{2\delta b} \bar{\bf \Omega}_2^\H \bigg) \bigg], \notag\\
\label{eq:tmp968}
\end{align}
where $\Pm_{2\delta a} \defeq \Bigg[\begin{matrix}
\frac{\nu_{2\delta}\rho_{1\tau}}{s_0} \Id_{s_0} & \mathbf{0} \\
\mathbf{0} & \frac{\rho_{2\delta}}{s_2}\Id_{s_2}
\end{matrix}\Bigg]$ and $\Pm_{2\delta b} \defeq \Bigg[\begin{matrix}
\frac{\rho_{1\delta}}{T-s_1-s_0} \big(\nu_{1\tau} + \nu_{2\delta}\frac{s_1-s_2}{s_0}\big) \Id_{s_0} & \mathbf{0} \\
\mathbf{0} & \frac{\rho_{2\delta}}{s_2}\Id_{s_2}
\end{matrix}\Bigg]$.

The term $I(\rvMat{Y}_{2\delta}; \rvMat{A} \big| \rvMat{Y}_{2\tau}, \rvMat{S}'_2, \rvMat{S}_2)$
can be upper bounded as follows:
\begin{align}
\lefteqn{\
I(\rvMat{Y}_{2\delta}; \rvMat{A} \cond \rvMat{Y}_{2\tau}, \rvMat{S}'_2, \rvMat{S}_2)
}\notag \\ 
&= I(\rvMat{Y}_{2\delta b}; \rvMat{A} \cond \rvMat{S}'_2, \rvMat{S}_2, \rvMat{Y}_{2\tau}) \label{eq:tmp1015}\\
&= I(\rvMat{Y}_{2\delta b}; \rvMat{A} \cond \rvMat{S}_2, \rvMat{Y}_{2\tau}) - I(\rvMat{Y}_{2\delta b}; \rvMat{S}'_2 \cond \rvMat{S}_2, \rvMat{Y}_{2\tau}) \label{eq:tmp1016}\\
&\le I(\rvMat{Y}_{2\delta b}; \rvMat{A} \cond \rvMat{S}_{2[s_1-s_2+1:T-s_2-s_0]}, \rvMat{Y}_{2\tau}) \label{eq:tmp1017}\\
&= h(\rvMat{A}|\rvMat{S}_{2[s_1-s_2+1:T-s_2-s_0]}, \rvMat{Y}_{2\tau}) - h(\rvMat{A}|\rvMat{S}_{2[s_1-s_2+1:T-s_2-s_0]}, \rvMat{Y}_{2\tau}, \rvMat{Y}_{2\delta b}) \\
&\le h(\rvMat{A}|\rvMat{S}_{2[s_1-s_2+1:T-s_2-s_0]}, \rvMat{Y}_{2\tau}) - h(\rvMat{A}|\rvMat{S}_{2[s_1-s_2+1:T-s_2-s_0]}, \rvMat{Y}_{2\tau}, \rvMat{Y}_{2\delta b},  \rvMat{G}_2 \Sigmam^{\frac12}_2 \Phim_2) \label{eq:tmp993}\\
&= h(\rvMat{A}|\rvMat{S}_{2[s_1-s_2+1:T-s_2-s_0]}, \rvMat{G}_2 \Sigmam^{\frac12}_2 \Phim_2) - h(\rvMat{A}|\rvMat{S}_{2[s_1-s_2+1:T-s_2-s_0]}, \rvMat{Y}_{2\delta b},  \rvMat{G}_2 \Sigmam^{\frac12}_2 \Phim_2) \label{eq:tmp994}\\
&= I(\rvMat{Y}_{2\delta b}; \rvMat{A} \cond \rvMat{S}_{2[s_1-s_2+1:T-s_2-s_0]}, \rvMat{G}_2 \Sigmam^{\frac12}_2 \Phim_2) \\
&= I\Big(\rvMat{Y}_{2\delta b} - \sqrt{\frac{\rho_{2\delta}}{s_2}}\rvMat{G}_2 \Sigmam^{\frac12}_2 \Phim_{22}\rvMat{S}_{2[s_1-s_2+1:T-s_2-s_0]}; \rvMat{A} ~\Big|~ \rvMat{S}_{2[s_1-s_2+1:T-s_2-s_0]}, \rvMat{G}_2 \Sigmam^{\frac12}_2 \Phim_{20}, \rvMat{G}_2 \Sigmam^{\frac12}_2 \Phim_{22}\Big) \\
&= I\Big(\rvMat{G}_2 \Sigmam^{\frac12}_2 \Phim_{20} \rvMat{A} + \rvMat{W}_{2[s_1+s_0+1:T]}; \rvMat{A} ~\Big|~ \rvMat{G}_2 \Sigmam^{\frac12}_2 \Phim_{20}\Big) \\
&=\big(T - s_1 - s_0\big) \EE \bigg[ \log \det \bigg(\Id_{N_2} + \rho_{1\delta}\Big(\nu_{2\tau} + \nu_{2\delta}\frac{s_1-s_2}{s_0}\Big) \bar{\bf \Omega}_{20} \bar{\bf \Omega}_{20}^\H\bigg) \bigg],\label{eq:tmp998}
\end{align}
where \eqref{eq:tmp1015} and \eqref{eq:tmp1016} follow from the Markov chains  $\rvMat{Y}_{2\delta a} \leftrightarrow \rvMat{S}'_2 \leftrightarrow \rvMat{A}$  and $\rvMat{Y}_{2\delta b} \leftrightarrow \rvMat{A} \leftrightarrow \rvMat{S}_2'$, respectively; \eqref{eq:tmp1017} holds because mutual information is non-negative and both $\rvMat{Y}_{2\delta b}$ and $\rvMat{A}$ are independent of $\rvMat{S}_{2[1:s_1-s_2]}$; \eqref{eq:tmp993} holds because conditioning reduces entropy; \eqref{eq:tmp994} holds because $\rvMat{A}$ is independent of both $\rvMat{Y}_{2\tau}$ and $\rvMat{G}_2 \Sigmam^{\frac12}_2 \Phim_2$, while given  $\rvMat{Y}_{2\delta b}$, $\rvMat{A}$ depends on $\rvMat{Y}_{2\tau}$ only through $\rvMat{G}_2 \Sigmam^{\frac12}_2 \Phim_2$;
and in the last equality, we used that $\EE [\rvMat{A}\rvMat{A}^\H] = \rho_{1\delta}\Big(\nu_{2\tau} + \nu_{2\delta}\frac{s_1-s_2}{s_0}\Big) \Id_{s_0}$. 

Substituting \eqref{eq:tmp968} and \eqref{eq:tmp998} into \eqref{eq:tmp951}, an achievable rate for User~$2$ is obtained. This rate and \eqref{eq:tmp919} give an achievable rate pair. Taking the convex hull of this pair over all possible power allocations satisfying \eqref{eq:hybrid:power_constraint} and all feasible values of $s_0, s_1,s_2$ provides an overall achievable rate region. This concludes the proof of Theorem~\ref{thm:BC_rate_hybrid}. 

\bibliographystyle{IEEEtran}
\bibliography{IEEEabrv,mybib}

\begin{thebibliography}{10}
\providecommand{\url}[1]{#1}
\csname url@samestyle\endcsname
\providecommand{\newblock}{\relax}
\providecommand{\bibinfo}[2]{#2}
\providecommand{\BIBentrySTDinterwordspacing}{\spaceskip=0pt\relax}
\providecommand{\BIBentryALTinterwordstretchfactor}{4}
\providecommand{\BIBentryALTinterwordspacing}{\spaceskip=\fontdimen2\font plus
\BIBentryALTinterwordstretchfactor\fontdimen3\font minus
  \fontdimen4\font\relax}
\providecommand{\BIBforeignlanguage}[2]{{%
\expandafter\ifx\csname l@#1\endcsname\relax
\typeout{** WARNING: IEEEtran.bst: No hyphenation pattern has been}%
\typeout{** loaded for the language `#1'. Using the pattern for}%
\typeout{** the default language instead.}%
\else
\language=\csname l@#1\endcsname
\fi
#2}}
\providecommand{\BIBdecl}{\relax}
\BIBdecl

\bibitem{Kermoal2002stochasticMIMO}
J.~P. {Kermoal}, L.~{Schumacher}, K.~I. {Pedersen}, P.~E. {Mogensen}, and
  F.~{Frederiksen}, ``A stochastic {MIMO} radio channel model with experimental
  validation,'' \emph{{IEEE} J. Sel. Areas Commun.}, vol.~20, no.~6, pp.
  1211--1226, Aug. 2002.

\bibitem{KaiYu2004modeling}
K.~Yu, M.~Bengtsson, B.~Ottersten, D.~McNamara, P.~Karlsson, and M.~Beach,
  ``Modeling of wide-band {MIMO} radio channels based on {NLoS} indoor
  measurements,'' \emph{IEEE Trans. Veh. Technol.}, vol.~53, no.~3, pp.
  655--665, May 2004.

\bibitem{Shiu_2000}
D.-S. Shiu, G.~J. Foschini, M.~J. Gans, and J.~M. Kahn, ``Fading correlation
  and its effect on the capacity of multielement antenna systems,''
  \emph{{IEEE} Trans. Commun.}, vol.~48, no.~3, pp. 502--513, Mar. 2000.

\bibitem{Jafar_2004}
S.~A. Jafar and A.~Goldsmith, ``Transmitter optimization and optimality of
  beamforming for multiple antenna systems,'' \emph{{IEEE} Trans. Wireless
  Commun.}, vol.~3, no.~4, pp. 1165--1175, Jul. 2004.

\bibitem{Jorswieck_2004}
E.~A. Jorswieck and H.~Boche, ``Channel capacity and capacity-range of
  beamforming in {MIMO} wireless systems under correlated fading with
  covariance feedback,'' \emph{{IEEE} Trans. Wireless Commun.}, vol.~3, no.~5,
  pp. 1543--1553, Sep. 2004.

\bibitem{Tulino_2005}
A.~M. Tulino, A.~Lozano, and S.~Verdu, ``Impact of antenna correlation on the
  capacity of multiantenna channels,'' \emph{{IEEE} Trans. Inf. Theory},
  vol.~51, no.~7, pp. 2491--2509, Jul. 2005.

\bibitem{chang_2006}
W.~{Chang}, S.~{Chung}, and Y.~H. {Lee}, ``Diversity-multiplexing tradeoff in
  rank-deficient and spatially correlated {MIMO} channels,'' in \emph{IEEE
  International Symposium on Information Theory (ISIT)}, Jul. 2006, pp.
  1144--1148.

\bibitem{Anese_2009}
E.~Dall'Anese, A.~Assalini, and S.~Pupolin, ``On the effect of imperfect
  channel estimation upon the capacity of correlated {MIMO} fading channels,''
  in \emph{{IEEE} Vehicular Technology Conference}, Apr. 2009, pp. 1--5.

\bibitem{Soysal_2010}
A.~Soysal, ``Tightness of capacity bounds in correlated {MIMO} systems with
  channel estimation error,'' in \emph{{IEEE} International Symposium on
  Personal, Indoor and Mobile Radio Communications}, Sep. 2010, pp. 667--671.

\bibitem{1237141}
{Hyundong Shin} and {Jae Hong Lee}, ``Capacity of multiple-antenna fading
  channels: spatial fading correlation, double scattering, and keyhole,''
  \emph{{IEEE} Trans. Inf. Theory}, vol.~49, no.~10, pp. 2636--2647, 2003.

\bibitem{995514}
A.~{Abdi} and M.~{Kaveh}, ``A space-time correlation model for multielement
  antenna systems in mobile fading channels,'' \emph{{IEEE} J. Sel. Areas
  Commun.}, vol.~20, no.~3, pp. 550--560, 2002.

\bibitem{Lee_2009}
J.-. {Lee}, J.-. {Ko}, and Y.-. {Lee}, ``Effect of transmit correlation on the
  sum-rate capacity of two-user broadcast channels,'' \emph{IEEE Trans.
  Commun.}, vol.~57, no.~9, pp. 2597--2599, Sep. 2009.

\bibitem{Lee_2010}
J.~W. Lee, H.~N. Cho, H.~J. Park, and Y.~H. Lee, ``Sum-rate capacity of
  correlated multi-user {MIMO} channels,'' in \emph{Information Theory and
  Applications Workshop ({ITA})}, Jan. 2010, pp. 1--5.

\bibitem{Naffouri_2009}
T.~Al-Naffouri, M.~Sharif, and B.~Hassibi, ``How much does transmit correlation
  affect the sum-rate scaling of {MIMO} {G}aussian broadcast channels?''
  \emph{{IEEE} Trans. Commun.}, vol.~57, no.~2, pp. 562--572, Feb. 2009.

\bibitem{Abdi_2002}
A.~{Abdi} and M.~{Kaveh}, ``A space-time correlation model for multielement
  antenna systems in mobile fading channels,'' \emph{{IEEE} J. Sel. Areas
  Commun.}, vol.~20, no.~3, pp. 550--560, Apr. 2002.

\bibitem{Nam_2012}
J.~Nam, J.~Y. Ahn, A.~Adhikary, and G.~Caire, ``Joint spatial division and
  multiplexing: Realizing massive {MIMO} gains with limited channel state
  information,'' in \emph{46th Annual Conference on Information Sciences and
  Systems ({CISS})}, Mar. 2012, pp. 1--6.

\bibitem{Nam_2014}
J.~Nam, ``Fundamental limits in correlated fading {MIMO} broadcast channels:
  Benefits of transmit correlation diversity,'' in \emph{{IEEE} International
  Symposium on Information Theory (ISIT)}, Jun. 2014, pp. 2889--2893.

\bibitem{Nam_2014_2}
J.~Nam, A.~Adhikary, J.~Y. Ahn, and G.~Caire, ``Joint spatial division and
  multiplexing: Opportunistic beamforming, user grouping and simplified
  downlink scheduling,'' \emph{{IEEE} J. Sel. Topics Signal Process.}, vol.~8,
  no.~5, pp. 876--890, Oct. 2014.

\bibitem{Nam_2017}
J.~Nam, G.~Caire, and J.~Ha, ``On the role of transmit correlation diversity in
  multiuser {MIMO} systems,'' \emph{{IEEE} Trans. Inf. Theory}, vol.~63, no.~1,
  pp. 336--354, Jan. 2017.

\bibitem{Adhikary_2014}
A.~Adhikary and G.~Caire, ``{JSDM} and multi-cell networks: Handling inter-cell
  interference through long-term antenna statistics,'' in \emph{48th Asilomar
  Conference on Signals, Systems and Computers}, Nov. 2014, pp. 649--655.

\bibitem{Adhikary_2015}
A.~Adhikary, H.~S. Dhillon, and G.~Caire, ``Massive-{MIMO} meets {HetNet}:
  Interference coordination through spatial blanking,'' \emph{{IEEE} J. Sel.
  Areas Commun.}, vol.~33, no.~6, pp. 1171--1186, Jun. 2015.

\bibitem{Adhikary_2014_2}
A.~Adhikary, E.~A. Safadi, and G.~Caire, ``Massive {MIMO} and inter-tier
  interference coordination,'' in \emph{Information Theory and Applications
  Workshop ({ITA})}, Feb. 2014, pp. 1--10.

\bibitem{Hassibi_2003}
B.~Hassibi and B.~M. Hochwald, ``How much training is needed in
  multiple-antenna wireless links?'' \emph{{IEEE} Trans. Inf. Theory}, vol.~49,
  no.~4, pp. 951--963, Apr. 2003.

\bibitem{Fan_2017}
F.~Zhang, M.~Fadel, and A.~Nosratinia, ``Spatially correlated {MIMO} broadcast
  channel: Analysis of overlapping correlation eigenspaces,'' in \emph{IEEE
  International Symposium on Information Theory (ISIT)}, Jun. 2017, pp.
  1097--1101.

\bibitem{Ngo_2017}
K.~{Ngo}, S.~{Yang}, and M.~{Guillaud}, ``An achievable {DoF} region for the
  two-user non-coherent {MIMO} broadcast channel with statistical {CSI},'' in
  \emph{2017 {IEEE} Information Theory Workshop ({ITW})}, Nov. 2017, pp.
  604--608.

\bibitem{Fan_2018}
F.~{Zhang} and A.~{Nosratinia}, ``Spatially correlated {MIMO} broadcast channel
  with partially overlapping correlation eigenspaces,'' in \emph{2018 {IEEE}
  International Symposium on Information Theory ({ISIT})}, Jun. 2018, pp.
  1520--1524.

\bibitem{Zheng_2002}
L.~Zheng and D.~N.~C. Tse, ``Communication on the {G}rassmann manifold: {A}
  geometric approach to the noncoherent multiple-antenna channel,''
  \emph{{IEEE} Trans. Inf. Theory}, vol.~48, no.~2, pp. 359--383, Feb. 2002.

\bibitem{Chiani2003capacity_spatially_correlated_MIMO}
M.~{Chiani}, M.~Z. {Win}, and A.~{Zanella}, ``On the capacity of spatially
  correlated {MIMO} {R}ayleigh-fading channels,'' \emph{{IEEE} Trans. Inf.
  Theory}, vol.~49, no.~10, pp. 2363--2371, Oct. 2003.

\bibitem{Luks1997AlgoNilpotent}
\BIBentryALTinterwordspacing
E.~M. Luks, F.~R\'{a}k\'{o}czi, and C.~R. Wright, ``Some algorithms for
  nilpotent permutation groups,'' \emph{J. Symb. Comput.}, vol.~23, no.~4, pp.
  335--354, Apr. 1997. [Online]. Available:
  \url{http://dx.doi.org/10.1006/jsco.1996.0092}
\BIBentrySTDinterwordspacing

\bibitem{Li_2012}
Y.~Li and A.~Nosratinia, ``Product superposition for {MIMO} broadcast
  channels,'' \emph{{IEEE} Trans. Inf. Theory}, vol.~58, no.~11, pp.
  6839--6852, Nov. 2012.

\bibitem{Li_2015}
------, ``Coherent product superposition for downlink multiuser {MIMO},''
  \emph{{IEEE} Trans. Wireless Commun.}, vol.~14, no.~3, pp. 1746--1754, Mar.
  2015.

\bibitem{Fadel_disparity}
M.~Fadel and A.~Nosratinia, ``Coherence disparity in broadcast and multiple
  access channels,'' \emph{{IEEE} Trans. Inf. Theory}, vol.~62, no.~12, pp.
  7383--7401, Dec. 2016.

\bibitem{ElGamal}
A.~El~Gamal and Y.-H. Kim, \emph{Network Information Theory}.\hskip 1em plus
  0.5em minus 0.4em\relax New York, NY, USA: Cambridge University Press, 2011.

\bibitem{Diestel2017graph_theory}
R.~Diestel, \emph{Graph Theory: 5th edition}, ser. Springer Graduate Texts in
  Mathematics.\hskip 1em plus 0.5em minus 0.4em\relax Springer-Verlag,
  {\copyright} Reinhard Diestel, 2017.

\bibitem{Garey1990NPcompleteness}
M.~R. Garey and D.~S. Johnson, \emph{Computers and Intractability; A Guide to
  the Theory of {NP}-Completeness}.\hskip 1em plus 0.5em minus 0.4em\relax New
  York, NY, USA: W. H. Freeman \& Co., 1990.

\bibitem{massivemimobook}
\BIBentryALTinterwordspacing
E.~Bj\"{o}rnson, J.~Hoydis, and L.~Sanguinetti, ``Massive {MIMO} networks:
  {Spectral}, energy, and hardware efficiency,'' \emph{Foundations and
  Trends{\textregistered} in Signal Processing}, vol.~11, no. 3-4, pp.
  154--655, 2017. [Online]. Available:
  \url{http://dx.doi.org/10.1561/2000000093}
\BIBentrySTDinterwordspacing

\bibitem{Jiang_2015}
Z.~{Jiang}, A.~F. {Molisch}, G.~{Caire}, and Z.~{Niu}, ``Achievable rates of
  {FDD} massive {MIMO} systems with spatial channel correlation,'' \emph{{IEEE}
  Trans. Wireless Commun.}, vol.~14, no.~5, pp. 2868--2882, 2015.

\bibitem{Caire:TIT2010}
G.~{Caire}, N.~{Jindal}, M.~{Kobayashi}, and N.~{Ravindran}, ``Multiuser {MIMO}
  achievable rates with downlink training and channel state feedback,''
  \emph{{IEEE} Trans. Inf. Theory}, vol.~56, no.~6, pp. 2845--2866, 2010.

\bibitem{5571881}
M.~{Ding} and S.~D. {Blostein}, ``Maximum mutual information design for {MIMO}
  systems with imperfect channel knowledge,'' \emph{{IEEE} Trans. Inf. Theory},
  vol.~56, no.~10, pp. 4793--4801, 2010.

\end{thebibliography}

\end{document}